\documentclass[10pt]{iopart}
\usepackage[utf8]{inputenc}
\usepackage{mathtools}
\usepackage{float}
\usepackage{graphicx}
\usepackage[margin=0.95in]{geometry}
\usepackage{physics}
\usepackage{amsmath}
\usepackage{bbm}
\usepackage{subcaption}
\usepackage{epstopdf}
\usepackage{bbold}
\usepackage{soul}

\usepackage[colorlinks=true, allcolors = blue]{hyperref}

\usepackage{fancyhdr}
\usepackage{mwe}
\usepackage{stfloats}
\usepackage{makecell}
\usepackage{amssymb}
\usepackage{animate}
\usepackage{comment}
\usepackage{listings}
\usepackage{amsthm}
\usepackage{setspace}
\usepackage[ruled,vlined]{algorithm2e}
\usepackage[outercaption]{sidecap}    
\sidecaptionvpos{figure}{t}
\usepackage{empheq}

\usepackage{tocloft}
\advance\cftsecnumwidth 6em\relax
\advance\cftsubsecindent 6em\relax
\advance\cftsubsubsecindent 5em\relax

\usepackage{afterpage}

\usepackage{color} 
\usepackage{multicol}
\usepackage{gensymb}

\usepackage{biblatex}
\bibliography{references.bib}

\usepackage[official]{eurosym}
\usepackage[framed,numbered]{matlab-prettifier}
\definecolor{mygreen}{RGB}{28,172,0} 
\definecolor{mylilas}{RGB}{170,55,241}
\definecolor{cy}{HTML}{0088A3}

\theoremstyle{definition}
\newtheorem{definition}{Definition}

\newtheorem{theorem}{Theorem}[section]

\newtheorem{lemma}[theorem]{Lemma}

\newtheorem*{remark}{Remark}
\newtheorem*{tada}{Theorem}
\newtheorem*{tada2}{Theorem}
\newtheorem{proposition}[theorem]{Proposition}

\DefineBibliographyStrings{english}{%
  mathesis = {MSc Thesis}
}

\begin{document}

\title{Spectral estimation for Hamiltonians: a comparison between classical imaginary-time evolution and quantum real-time evolution}
\author{ME Stroeks$^{1,2}$, J Helsen$^3$ and BM Terhal$^{1,2,4}$}
\address{$^1$ QuTech, Delft University of Technology, P.O. Box 5046, 2600 GA Delft, The Netherlands}
\address{$^2$ Faculty of EEMCS, Delft University of Technology, Mekelweg 4, 2628 CD Delft, The Netherlands}
\address{$^3$ CWI $\&$ QuSoft, Science Park 123, 1098 XG Amsterdam, The Netherlands}
\address{$^4$ JARA Institute for Quantum Information, Forschungszentrum Juelich, D-52425 Juelich, Germany}
\ead{m.e.h.m.stroeks@tudelft.nl, Jonas.Helsen@cwi.nl, b.m.terhal@tudelft.nl}

\begin{abstract}
We present a classical Monte Carlo (MC) scheme which efficiently estimates an imaginary-time, decaying signal for stoquastic (i.e. sign-problem-free) local Hamiltonians. The decay rates in this signal correspond to Hamiltonian eigenvalues (with associated eigenstates present in an input state) and can be classically extracted using a classical signal processing method like ESPRIT. We compare the efficiency of this MC scheme to its quantum counterpart in which one extracts eigenvalues of a general local Hamiltonian from a real-time, oscillatory signal obtained through quantum phase estimation circuits, again using the ESPRIT method. 
We prove that the ESPRIT method can resolve $S={\rm poly}(n)$ eigenvalues, assuming a $1/{\rm poly}(n)$ gap between them, with ${\rm poly}(n)$ quantum and classical effort through the quantum phase estimation circuits, assuming efficient preparation of the input state. We prove that our Monte Carlo scheme plus the ESPRIT method can resolve $S=O(1)$ eigenvalues, assuming a $1/{\rm poly}(n)$ gap between them, with ${\rm poly}(n)$ purely classical effort for stoquastic Hamiltonians, requiring some access structure to the input state. However, we also show that under these assumptions, i.e. $S=O(1)$ eigenvalues, assuming a $1/{\rm poly}(n)$ gap between them and some access structure to the input state, one can achieve this with ${\rm poly}(n)$ purely classical effort for {\em general} local Hamiltonians. These results thus quantify some opportunities and limitations of classical Monte Carlo methods for spectral estimation of Hamiltonians. We numerically compare the MC eigenvalue estimation scheme (for stoquastic Hamiltonians) and the QPE eigenvalue estimation scheme by implementing them for an archetypal stoquastic Hamiltonian system: the transverse field Ising chain. 
\end{abstract}

\maketitle

\tableofcontents

\section{Introduction}
\label{sec:intro}

In general, it is a computationally intractable task to obtain, by classical or quantum means, the eigenvalues of a Hamiltonian $H$ associated with a many-body quantum system. However, more restricted tasks related to estimating the spectrum of $H$ can be executed on a quantum computer by means of quantum phase estimation algorithms \cite{MikeIke, SHF:QPE,TomBarbara,somma:njp}, using the ability to simulate the real-time dynamics $e^{-iHt/\hbar}$ efficiently on a quantum computer via Trotterization \cite{Lloyd}.

\textit{Classical} alternatives are provided by quantum Monte Carlo methods \cite{book:QMC, QMCreview}. The efficiency of quantum Monte Carlo methods when used to simulate many-body systems is generally limited by the \textit{sign problem}. This can cause the variance of the estimator in the Monte Carlo algorithm to grow exponentially in the system size $n$, necessitating an exponential number of runs of the Monte Carlo algorithm.

A (ubiquitous) class of Hamiltonians that is sign-problem-free has been formalized under the name \textit{stoquastic} Hamiltonians \cite{Stoq}. Roughly speaking, a (real-valued) Hamiltonian is stoquastic (in a particular basis $\mathcal{B}$) if its off-diagonal elements are non-positive: $\bra{x}H\ket{y}\leq 0$, for $x \neq y$ (with $\ket{x}$, $\ket{y}$ being elements of $\mathcal{B}$). As a consequence, its associated Gibbs density matrix $e^{-\tau H}$ is an element-wise non-negative matrix (for $\tau \in \mathbb{R}_{+}$). This property makes it particularly suitable for Monte Carlo sampling as complexity results \cite{Stoq, ComplStoq, AB:stoq} and various algorithmic results \cite{BravyiGosset,SergeyQPEvsMC,CH:mixing,Crosson} have demonstrated. 

Since stoquastic Hamiltonians are sign-problem-free, it is of interest to see if one can indeed prove that (part of its) spectrum can be efficiently estimated through classical Monte Carlo methods. Conversely, can quantum algorithms, \textit{even} for stoquastic Hamiltonians, provide an advantage over Monte Carlo algorithms in carrying out this task? In this work, we address these questions by making a direct comparison between the task of estimating the spectral content of a stoquastic local Hamiltonian in an input state via a quantum circuit versus via a classical Monte Carlo scheme. In addition, we investigate to what extent this task can be efficiently carried out classically for a general local Hamiltonian. 

Central in our study is, first of all, the real-time {\em signal} 
\begin{equation}
    g_R(k) = \bra{\Phi}e^{-iHk\Delta t}\ket{\Phi} = \sum_{j=1}^{2^{n}}\bigl\lvert\braket{\psi_{j}}{\Phi}\bigr\rvert^{2}\Big(e^{-iE_{j}\Delta t}\Big)^{k},
    \label{eq:signal-QPE}
\end{equation}
for $k=0, 1\ldots, K$ and where $\ket{\Phi}$ is some pre-specified $n$-qubit input state. The estimation of $g_R(k)$ for various $k$ is a crucial step in the quantum phase estimation algorithm (QPE). In what follows we will fix $\Delta t$ so that the eigenstates $\ket{\psi_j}$ with nonzero or substantial overlap $|\bra{\psi_j} \Phi\rangle|^2 > 0$, showing up in the signal, have the property that $E_j \Delta t \in [0,2\pi)$. Thus, from now on, we assume that these $E_j$ are shifted and rescaled to lie in $[0,2\pi)$. 
We will assume that there are at most $S$ eigenvectors with nonzero $ |\bra{\psi_j} \Phi\rangle|^2$, where $S$ is desired to be ${\rm poly}(n)$ or less for overall efficiency. 
Identifying a state $\ket{\Phi}$ which has non-zero overlap on only a few ($S={\rm poly}(n)$ or $S=O(1)$) eigenstates and which obeys the assumptions in the following Theorems is not so simple, and can be considered one of the bottlenecks in using quantum phase estimation or other Monte Carlo methods to determine spectral information of the Hamiltonian.

Besides the real-time signal, one can define the imaginary-time \textit{signal}
\begin{equation}
g_I(k) = \bra{\Phi}e^{-Hk }\ket{\Phi} = \sum_{j=1}^{2^{n}}\bigl\lvert\braket{\psi_{j}}{\Phi}\bigr\rvert^{2}\Big(e^{-E_{j}}\Big)^{k},
\label{eq:signal-MC}
\end{equation}
where again we can assume that $E_j \in [0,2\pi)$. We will prove, {\em for local stoquastic Hamiltonians}, that the quantum cost of estimating $g_R(k)$ and the classical Monte Carlo cost of estimating $g_I(k)$ within error $\epsilon$ are \emph{approximately identical}, although the assumptions on our knowledge/preparation costs of $\ket{\Phi}$ are slightly different in the two cases. 

The two statements are as follows:

\bigskip
\begin{theorem}
For a local Hamiltonian acting on $n$ qubits, one can estimate $g_R(k)$ in Eq.~\eqref{eq:signal-QPE} with probability at least $1-\delta$ with sampling error $\epsilon$ and Trotter error $\epsilon_{\text{trot}}$ (and total error $\epsilon_{\rm tot}=\epsilon+\epsilon_{\text{trot}}$), using quantum circuits acting on $n+1$ qubits, 
where the depth of the quantum circuit scales as $\mathcal{O}\big(k^{1+o(1)}\big) \mathcal{O}\big(\epsilon_{\text{trot}}^{-o(1)}\big) \times {\rm poly}(n)$ and the number of times one executes the circuit is $\Theta(\epsilon^{-2} \log(4\delta^{-1}))$, under the assumption that $\ket{\Phi}$ is a state of $n$ qubits which can be generated by a ${\rm poly}(n)$-size quantum circuit. Hence to obtain $g_R(k)$ for $k=0,\ldots, K$, with error at most $\epsilon_{\rm tot}=\epsilon+\epsilon_{\rm trot}$ for all $k$, with probability $1-\delta$ requires using quantum circuits for $k=0, \ldots, K$, each acting on $n+1$ qubits, where the depth of the quantum circuit scales as $\mathcal{O}\big(k^{1+o(1)}\big) \mathcal{O}\big(\epsilon_{\text{trot}}^{-o(1)}\big)\times {\rm poly}(n)$ and each circuit is repeated $\Theta(\epsilon^{-2}\left[ \log(4\delta^{-1})+\log(K)\right])$ times.
\label{thm:re}
\end{theorem}

\bigskip
\begin{theorem}
For a local stoquastic Hamiltonian acting on $n$ qubits, one can estimate $g_I(k)$ in Eq.~\eqref{eq:signal-MC} with probability at least $1-\delta$ with total error $\epsilon_{\rm tot}=\epsilon+\epsilon_{\rm trot}$, using a classical MC algorithm on $n$-bit strings where the depth of the algorithm scales as $\mathcal{O}\big(k^{1+o(1)}\big) \mathcal{O}\big(\epsilon_{\text{trot}}^{-o(1)}\big)\times{\rm poly}(n)$ and the number of times one runs the algorithm is $\Theta(\epsilon^{-2} \text{log}(\delta^{-1}))$, under the assumption that $\ket{\Phi} = \sum_{x=1}^{2^{n}}\Phi(x)\ket{x}$ is a normalized state of $n$ qubits such that (1) $\frac{\Phi(y)}{\Phi(x)}$ can be efficiently (${\rm poly}(n)$) calculated for a {\em given} $x$ and $y$ and (2) we can efficiently draw samples from the probability distribution $P(x) = \bigl\lvert \Phi(x) \bigr\rvert^{2}$. Hence to obtain $g_I(k)$ for all $k=0,\ldots, K$, with error at most $\epsilon_{\rm tot}$ for each $k$, with probability $1-\delta$ requires using a classical MC algorithm on $n$-bit strings for $k=0, \ldots, K$, where the depth of each algorithm scales as $\mathcal{O}\big(k^{1+o(1)}\big) \mathcal{O}\big(\epsilon_{\text{trot}}^{-o(1)}\big)\times{\rm poly}(n)$ and the number of times one runs the algorithm (for each $k$) is $\Theta(\epsilon^{-2}\left[ \text{log}(\delta^{-1}) + \text{log}(K) \right])$.
\label{thm:im}
\end{theorem}

\bigskip
We then ask, given knowledge of either the real-time signal $g_R(k)$ or imaginary-time signal $g_I(k)$, what can be learnt about those eigenvalues $E_j$, whose associated eigenstates have nonzero overlap with the input state $\ket{\Phi}$? The signals $g_I(k)$ and $g_R(k)$ respectively correspond to a probabilistic sum of decaying components and a sum of oscillating components with decay rates and oscillation frequencies $E_{j}$ as a function of discrete `time' $k = 0,\ldots,K$. Hence a method which extracts those decay and oscillation rates from knowing $g_I(k)$ or $g_R(k)$ at various $k$ is needed. A method of choice which has already been used in quantum information theory is the matrix pencil method \cite{Sarkar, Sarkar2,pt:prony} (with equivalent methods known as ESPRIT and MUSIC). This method has been used for processing randomized benchmarking data \cite{OWE:RB, helsen2020general}, quantum phase estimation \cite{TomBarbara}, spectral tomography of superoperators \cite{francesco}, for processing experimental time-series data to identify Hamiltonian parameters \cite{hangleiter2021precise} or generally in processing discretely-sampled decaying Ramsey signals.

Using this method, it is known that if {\em either} $g_R(k)$ or $g_I(k)$ is known {\em exactly} for $k=0, \ldots, K$ where $K+1\geq 2S$, one can learn those eigenvalues $E_j$ and probabilities $|\bra{\psi_j} \Phi\rangle |^2$ {\em exactly}. However, in the presence of sampling and Trotter noise, the resolving power also depends on the gap between the eigenvalues $E_j$, the number $S$ of eigenvalues and whether we extract them from an oscillating or decaying signal. Our work is thus focused on understanding whether there are fundamental advantages in learning $g_R(k)$ with noise versus learning $g_I(k)$ with noise, as this quantifies the benefit of a quantum algorithm versus a classical algorithm for spectral estimation of (stoquastic) Hamiltonians. 

Not surprisingly, there are drawbacks to processing data from the imaginary-time evolution. As the signal decays exponentially, $k$ cannot be chosen too large otherwise the signal becomes smaller than the noise. Our goal is to quantify this precisely and show that, at least theoretically, a regime exists in which the Monte Carlo method may be competitive.

The first statement we make can be viewed as a summary of previous work, namely it combines Lemma \ref{prop:retimesignalest} via Theorem \ref{thm:re} and the performance of the ESPRIT method in Theorem \ref{thm:osc-esprit-gap} in the presence of a gap:

\bigskip
\begin{theorem}
Given a local Hamiltonian on $n$ qubits. Let the number of eigenvectors supported in some (efficient-to-prepare) input state $\ket{\Phi}$ be $S=p_1(n)$ (with $p_1(n)$ some polynomial in $n$), and each occurs with nonzero probability at least $1/{\rm poly}(n)$. Furthermore, assume that the $S$ eigenvalues $\{E_i\}$ with $E_i \in [0,2\pi)$ are sufficiently well-separated, i.e. at least by a gap $\Delta \geq C/K$ with constant $C$ and $K=\Theta(p_1(n))$. Then using Hadamard test (QPE) quantum circuits plus signal post-processing via ESPRIT, each requiring a ${\rm poly}(n)$ effort, one can resolve the eigenvalues $\{E_j\}$ with distance $d(\{E_i\}, \{\tilde{E}_j\})$ (defined in Eq.~\eqref{eq:matching_error}) at most $1/{\rm poly}(n)$.
\label{thm:QPE-total}
\end{theorem}

\bigskip
For local stoquastic Hamiltonians the combination of Lemma \ref{prop:imtimesignalest}, via Theorem \ref{thm:im} and the performance of the ESPRIT method in Theorem \ref{thm:final} in the presence of a gap leads to:

\bigskip
\begin{theorem}
Given a local stoquastic Hamiltonian on $n$ qubits. Let the number of eigenvectors supported in some efficient-to-sample (i.e. with ${\rm poly}(n)$ effort) input state $\ket{\Phi}$ be $S=O(1)$, and each occurs with nonzero probability at least $1/{\rm poly}(n)$. In addition, assume that for a fixed $x,y$ it is efficient to compute $\frac{\Phi(y)}{\Phi(x)}$. Furthermore, assume that the $S$ eigenvalues $\{E_i\}, E_i \in [0,2\pi)$ are sufficiently well-separated, i.e. at least by $\Delta \geq 1/{\rm poly}(n)$ with some ${\rm poly}(n)$. Then using a Monte Carlo algorithm plus signal post-processing via ESPRIT, each requiring (some) ${\rm poly}(n)$ effort, one can resolve the eigenvalues $\{E_j\}$ with distance $d(\{E_i\}, \{\tilde{E}_j\})$ at most $1/{\rm poly}(n)$.
\label{thm:MC-total}
\end{theorem}

\bigskip
Theorem \ref{thm:MC-total} immediately begs the question whether such a result could hold for general local Hamiltonians as well: the assumptions that there are only $S=O(1)$ eigenstates in the initial state, as well as the assumption of efficient access to the initial state appear rather strong. To address this question, we define another real-valued, decaying signal as 
\begin{equation}
g_D(k) = \bra{\Phi}\left(I-H/2\pi\right)^{k}\ket{\Phi} = \sum_{j=1}^{2^{n}}\bigl\lvert\braket{\psi_{j}}{\Phi}\bigr\rvert^{2}\Big(I-H/2\pi \Big)^{k}.
\label{eq:signal-MC-D}
\end{equation}
If $S=O(1)$ and if $g_D(k)$ can be estimated with some accuracy for $k=1,\ldots, K=O(1)$, we can also apply the ESPRIT method to extract these $S$ eigenvalues. We note that this requires that the eigenvalues $E_j$ are bounded away from $2\pi$. Hence if we use $g_D(k)$ we assume that we have shifted and rescaled the eigenvalues so that, say, the $E_j$s lie in $[0,\pi]$.

One can prove that for general local Hamiltonians, assuming $S=O(1)$ eigenvalues in $\ket{\Phi}$, one can estimate $g_D(k)$ with $\epsilon$ accuracy, under an assumption about the access to $\ket{\Phi}$ which is identical to the Monte Carlo case for stoquastic Hamiltonians (Theorem \ref{thm:MC-total}). In fact, this result shows that Theorem \ref{thm:MC-total} is not particular to local stoquastic Hamiltonians at all, if we only care about `nominally ${\rm poly}(n)$' algorithms. However, the computational cost of estimating $g_D(k)$ for general local Hamiltonians is significantly higher in practice compared to the Monte Carlo method for stoquastic Hamiltonians. The result expressed in Lemma \ref{lem:SVT} can be viewed as `dequantization' as it is similar in spirit to the Singular Value Transformation (SVT) tool (Theorem 3 in \cite{dequant:GG}). Theorem 3 in \cite{dequant:GG} is used to construct an algorithm that estimates the ground state energy of a Hamiltonian to $O(1)$ (in $n$) precision, given an initial state with only some constant overlap with the ground state. 

Applying the ESPRIT analysis to Lemma \ref{lem:SVT}, we will obtain the following Theorem:

\bigskip
\begin{theorem}
Given a local Hamiltonian on $n$ qubits. Let the number of eigenvectors supported in some efficient-to-sample (${\rm poly}(n)$ effort) input state $\ket{\Phi}$ be $S=O(1)$, and each occurs with nonzero probability at least $1/{\rm poly}(n)$. In addition, assume that for a fixed $x,y$ it is efficient to compute $\frac{\Phi(y)}{\Phi(x)}$. Furthermore, assume that the $S$ eigenvalues $\{E_i\}, E_i \in [0,\pi]$ are sufficiently well-separated, i.e. at least by $\Delta \geq 1/{\rm poly}(n)$ with some ${\rm poly}(n)$. Then using Lemma \ref{lem:SVT} plus signal post-processing via ESPRIT, each requiring (some) ${\rm poly}(n)$ classical effort, one can resolve the eigenvalues $\{E_j\}$ with distance $d(\{E_i\}, \{\tilde{E}_j\})$ at most $1/{\rm poly}(n)$.
\label{thm:MCD-total}
\end{theorem}

\bigskip

To investigate practical aspects of the MC scheme for stoquastic Hamiltonians and compare it to the quantum scheme, we numerically study the one-dimensional Ising chain in a transverse field $g$ \cite{Sachdev} in a \textit{proof-of-principle} setting. We numerically study, amongst several other aspects, the recovery of the ground-state and first-excited-state eigenvalues in the ($g>1$)-regime from the signals $g_{R}(k)$ and $g_{I}(k)$ (in the presence of sampling noise and Trotter error) using the ESPRIT method.

An overview of the paper is as follows. In Section \ref{sec:QPE}, we review the \textit{Hadamard} or \textit{overlap} quantum subroutine (Lemma \ref{prop:retimesignalest}) and we present the Monte Carlo algorithm (Lemma \ref{prop:imtimesignalest}) for stoquastic Hamiltonians with its proof, as well as stating a straightforward Lemma \ref{lem:SVT} on `dequantization'. Section \ref{sec:MPM-analysis} reviews the ESPRIT method and has an extensive \ref{app:MPM} in which we prove the performance of the ESPRIT method for imaginary-time decaying signals using many lemmas also needed in the real-time signal case. The arguments for Theorem \ref{thm:MCD-total} are presented in Section \ref{sec:MPM-analysis} as well. In Section \ref{simsection}, we numerically compare the quantum scheme and the Monte Carlo scheme (for stoquastic Hamiltonians) for determining part of the spectrum of a transverse field Ising chain. In Section \ref{sec:con}, we discuss our work and propose some directions for future study. Several appendices give additional background information and details. 

We note that very extensive literature exists on the Monte Carlo power method \cite{book:QMC} in which one applies a sequences of steps which gradually project an initial input state onto the ground state. In this method, unlike in our MC scheme of Lemma \ref{prop:imtimesignalest}, one renormalizes the state after each iteration, so that the signal does not die out. In our approach, we do not renormalize, but study the decay rates themselves. In terms of other previous work, we note that in \cite{SergeyQPEvsMC} the ground state energy of a stoquastic Hamiltonian was efficiently estimated by means of a projector Monte Carlo scheme, under an additional `guiding state' promise. In \cite{motta:imag} the authors consider the implementation of the imaginary-time evolution $\exp(-\tau H)$ on a quantum computer in order to prepare a ground state of any local Hamiltonian. Note that our goal is not to prepare any ground or excited state but rather only learn some eigenvalues.

In the remainder of this section, we will review a few definitions which are used in this paper.

\begin{definition}
\textbf{Stoquastic Hamiltonians} A (real-valued) Hamiltonian $H$ is (globally) stoquastic \cite{Stoq} in a basis $\mathcal{B}$ if all its off-diagonal elements are non-positive: $\bra{x}H\ket{y}\leq 0$, for $x\neq y$ (and states $\ket{x}$, $\ket{y}$ being elements of basis $\mathcal{B}$).
\end{definition}
In this work, we are interested in Hamiltonians that are \textit{local} and \textit{stoquastic}:
\begin{definition}
\textbf{Local Hamiltonians} A Hamiltonian $H$ associated with a system consisting of $n$ degrees of freedom (e.g. spins/qubits) is local if it admits a decomposition into a set of Hermitian operators $\{H_{i}\}$ -- i.e. $\sum_{i}^{N}H_{i}$ -- such that each $H_{i}$ acts non-trivially on $O(1)$ (not growing with $n$) degrees of freedom of the system.
\label{LHpaper}
\end{definition}
\noindent
We denote the maximum number of degrees of freedom on which each $H_{i}$ acts non-trivially (i.e. its \textit{locality}) by $k$ and note that the number of terms in a local Hamiltonian is $N=O(n^k)$.

For local Hamiltonians there is a slightly stronger notion of stoquasticity, called termwise stoquasticity, which can differ from the definition of stoquasticity given above, see \cite{Stoq,ComplStoq}. 
\begin{definition}
\textbf{Termwise stoquastic Hamiltonians} A (real-valued) $k$-local Hamiltonian $H$ is $m$-termwise stoquastic in a basis $\mathcal{B}$ if it admits a decomposition into (real-valued) $m$($\geq k$)-local terms $\{H_{a}\}$ such that each $H_{a}$ is stoquastic: $\forall a$, $\bra{x}H_{a}\ket{y}\leq 0$, for $x\neq y$ (and states $\ket{x}$, $\ket{y}$ being elements of basis $\mathcal{B}$).
\end{definition}
Most many-body Hamiltonians considered in physics which are stoquastic are $O(1)$-termwise stoquastic. The results in this paper apply to both termwise stoquastic as well as globally stoquastic Hamiltonians (using some small adaptions employing results in \cite{ComplStoq}), and we will refer to them simply as `stoquastic'. 

For a matrix $X$ we will use the operator or spectral norm $\norm{X}=\sqrt{\lambda_{\rm max}(X^{\dagger}X)}=\sigma_{\rm max}(X)$, where $\sigma_{\rm max}(X)$ is the largest singular value of $X$. We also refer to the Frobenius norm $\norm{X}_F=\sqrt{{\rm Tr} (X^{\dagger} X)}$ and the induced$-\infty$ norm $\norm{X}_{\infty}=\max_{i}\sum_j |X_{ij}|$. For an $m\times n$ matrix $X$, we use $\norm{X} \leq \sqrt{m} \norm{X}_{\infty}$ and $\norm{X} \leq \norm{X}_F$.

\section{Quantum scheme versus Monte Carlo scheme for spectral estimation}
\label{sec:QPE}
In this section we show how to estimate $g_R(k)$ on a quantum computer, and $g_I(k)$ for stoquastic Hamiltonians via a Monte Carlo algorithm, as well as how to estimate $g_D(k)$ inefficiently (in $k$) via a classical algorithm for general local Hamiltonians.

Lemma \ref{prop:retimesignalest} states a well-known quantum subroutine, namely the Hadamard or overlap test, while a new result, a Monte Carlo version of the routine, is proved in Lemma \ref{prop:imtimesignalest}. After these Lemmas, the proofs of Theorems \ref{thm:re} and \ref{thm:im} are given. Then we give Lemma \ref{lem:SVT} for general local Hamiltonians, using similar tools as in Lemma \ref{prop:imtimesignalest}.

We note that the overlap test is used in versions of quantum phase estimation which do not aim at preparing an energy eigenstate of the Hamiltonian, but rather only learn the spectral content in its input state, as in Refs.~\cite{TomBarbara, somma:njp, lin2021heisenberglimited}. Here we basically follow this approach for the real-time quantum evolution, which can in addition be randomized to save on implementation costs, see \cite{WBC:random}.

\bigskip
\begin{lemma}[Hadamard or Overlap Test]
Let $\mathcal{F}  \equiv \bra{\Phi} G_{1} G_{2}\: ...\: G_{L} \ket{\Phi}$, 
where:
\begin{enumerate}
    \item $\ket{\Phi} = \sum_{x=1}^{2^{n}}\Phi(x)\ket{x}$ is a state of $n$ qubits which can be generated by a ${\rm poly}(n)$-size quantum circuit. 
    \item Each $G_{l}$ is a $k$-local unitary matrix.
\end{enumerate}
 $\mathcal{F}$ can be estimated within error $\epsilon$ with probability at least $1-\delta$ with a quantum circuit with $\Theta(\epsilon^{-2}\log(4\delta^{-1})) \times \left[\Theta(L) + \text{poly(n)}\right]$ single and two-qubit gates. 
\label{prop:retimesignalest}
\end{lemma}

\bigskip
\begin{proof}
Figure \ref{qpecircuit} depicts the quantum circuit which is used. It involves an $n$-qubit register and a single ancillary qubit. The state of the composite system can be tracked through the circuit and the final state can be found to be (where $R(\theta) \equiv e^{-i\theta Z/2}$):
\begin{equation}
    \frac{1}{2}\Big(\big(e^{-i\theta/2}I+e^{i\theta/2}G_1 G_2 \ldots G_L \big)\ket{0}_{a}\otimes \ket{\Phi} + \big(e^{-i\theta/2}I-e^{i\theta/2}G_1 G_2 \ldots G_L\big)\ket{1}_{a}\otimes \ket{\Phi}\Big).
\end{equation}
A $Z$-measurement is now performed on the ancillary qubit, measuring either $\ket{0}$ or $\ket{1}$ with associated outcomes resp. $m=0$ or $m=1$. The probability to measure state $\ket{0}$ ($m=0$) on the ancillary qubit after application of the depicted gates is then given by:
\begin{eqnarray}
    \text{Pr}\big(m=0\:|\:\theta \big) = \frac{1}{2}+\frac{1}{4}\bigg( e^{i\theta}\bra{\Phi} G_1 G_2 \ldots G_L \ket{\Phi} + e^{-i\theta}\big( \bra{\Phi} G_1 G_2 \ldots G_L \ket{\Phi} \big)^{*} \bigg) = \nonumber \\
    \begin{cases}
    \frac{1}{2}+\frac{1}{2}\text{Re}\Big( \bra{\Phi} G_1 G_2 \ldots G_L \ket{\Phi} \Big),\text{ for }\theta = 0, \\
    \frac{1}{2}-\frac{1}{2}\text{Im}\Big( \bra{\Phi} G_1 G_2 \ldots G_L \ket{\Phi} \Big),\text{ for }\theta = \frac{\pi}{2}.
    \end{cases}
    \label{eq:qpeprobb}
\end{eqnarray}
In the final expression, we have restricted ourselves to $\theta = 0$ and $\theta = \frac{\pi}{2}$, which are the $\theta$ values of interest. Suppose that for $\theta = 0$ and $\theta=\pi/2$, the quantum circuits are repeated $|\Sigma|$ times to obtain a set $2|\Sigma|$ of independent realizations of the ancillary-qubit state to be measured and let $\bigl\lvert\Sigma_0^{\theta=0}\bigr\rvert$ and $\bigl\lvert\Sigma_0^{\theta=\pi/2}\bigr\rvert$ be the number of times the ancilla measurement returns 0 so that 
\begin{equation}
    \tilde{\mathcal{F}}=\left(2\frac{\bigl\lvert\Sigma_0^{\theta=0}\bigr\rvert}{|\Sigma|}-1\right)-i \left(2\frac{\bigl\lvert\Sigma_0^{\theta=\pi/2}\bigr\rvert}{|\Sigma|}-1\right)
\end{equation}
is our (unbiased) estimator, i.e. $\mathbb{E}(\tilde{\mathcal{F}}(t))=\mathcal{F}(t)$.
Then by means of the Chernoff bound we have
\begin{align}
\begin{split}
\text{Pr}\bigg( \Bigl\lvert \tilde{\mathcal{F}}-\mathcal{F} \Bigr\rvert \leq \epsilon \bigg) \geq&\: \text{Pr}\bigg(\Bigl\lvert \text{Re}(\tilde{\mathcal{F}} - \mathcal{F}) \Bigr\rvert \leq \epsilon/\sqrt{2} \bigg)\: \text{Pr}\bigg(\Bigl\lvert \text{Im}(\tilde{\mathcal{F}} - \mathcal{F}) \Bigr\rvert \leq \epsilon/\sqrt{2} \bigg) \\
=&\: \bigg( 1-\text{Pr}\bigg(\Bigl\lvert \text{Re}(\tilde{\mathcal{F}} - \mathcal{F}) \Bigr\rvert \leq \epsilon/\sqrt{2} \bigg)\bigg)\: \bigg( 1-\text{Pr}\bigg(\Bigl\lvert \text{Im}(\tilde{\mathcal{F}} - \mathcal{F}) \Bigr\rvert \leq \epsilon/\sqrt{2} \bigg) \bigg) \\
\geq&\: \bigg[\text{max}\Big(0,\big(1-2\:\text{exp}(-|\Sigma|\epsilon^{2}/4)\big)\Big)\bigg]^{2} \\
\geq&\: 1-4\:\text{exp}\big(-|\Sigma|\epsilon^{2}/4\big) = 1-\delta,
\end{split}
\end{align}
where the number of samples is chosen as $|\Sigma|=\Theta(\epsilon^{-2}\log(4\delta^{-1}))$.
\end{proof}

\begin{SCfigure}[][t]
    \centering
    \includegraphics[width=11cm]{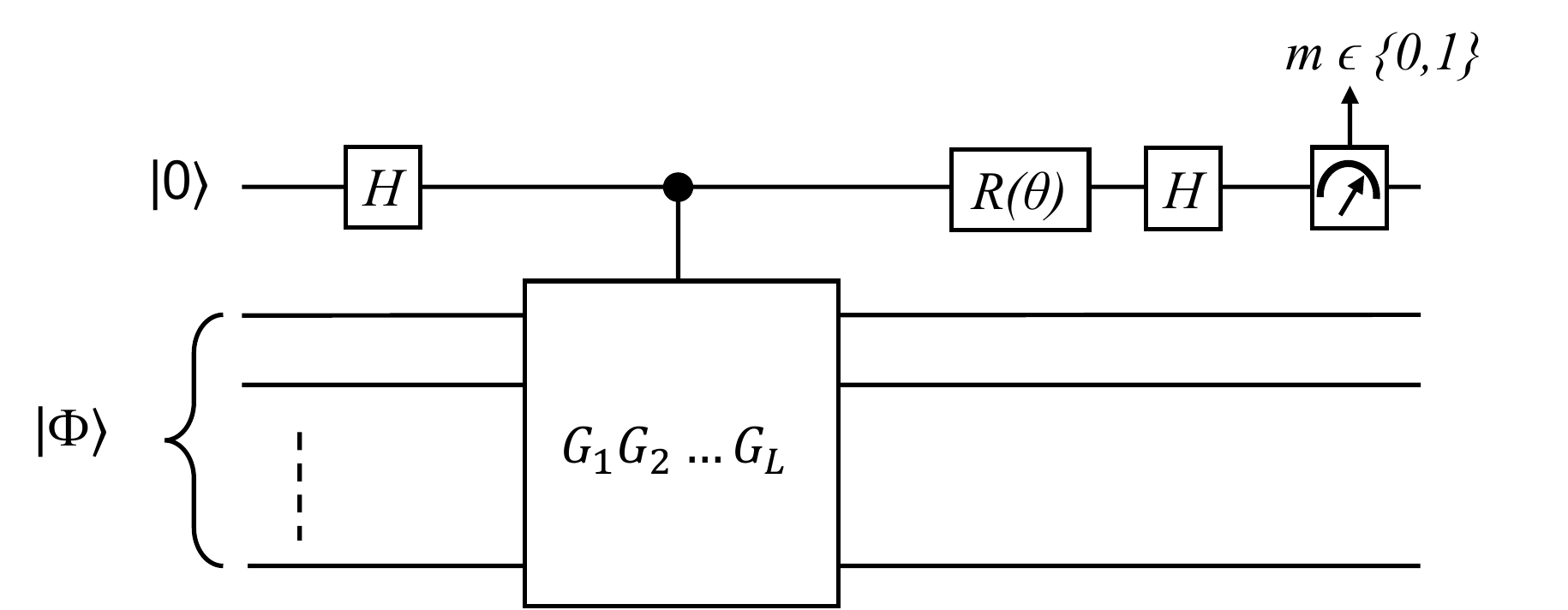}
    \caption{Basic circuit with a single ancillary qubit and an $n$-qubit register (initialized in state $\ket{\Phi}$).} 
    \label{qpecircuit}
\end{SCfigure}

Next, we will consider a classical Monte Carlo version of the quantum routine given above. The key result here is Lemma \ref{prop:imtimesignalest}. 
However, before we can state it, we collect a few facts about matrices $G_i$ which will be useful in the proof of Lemma \ref{prop:imtimesignalest}. The matrices $G_{i}$ that we will consider now can be seen as analogous to the (unitary) local real-time propagation operators $G_{i}$ considered earlier in Lemma \ref{prop:retimesignalest}, but are now local \emph{imaginary-time} propagation operators. The $G_{i}$ are no longer unitary but, for local stoquastic Hamiltonians, are elementwise nonnegative. They are of the form $G_{i}=e^{-a_{i}/M\:kH_{i}}$, where $a_{i}/M$ is a positive parameter set by the Trotterization scheme and $k = 0,1,\ldots,K$ denotes imaginary-time coordinate. We have the following proposition on further properties of these operators:

\begin{proposition}
Let $G_{i}=e^{-a_{i}/M\:kH_{i}}$, where $H_{i}$ is a stoquastic Hermitian matrix (a term in $H=\sum_i H_i$) which acts nontrivially on some subset of $O(1)$ qubits. Let the smallest eigenvalue of $H_{i}$ be 0, i.e. $\lambda_{\text{min}}(H_{i})=0$. We have:
\begin{itemize}
\item The matrix $G_{i}$ is an elementwise nonnegative and positive definite matrix, with eigenvalues in the interval $(0,1]$, acting nontrivially only on the same $O(1)$ qubits as $H_{i}$. 

\item If $G_i$ is reducible, then we can write $G_{i} = \oplus_{b=1}^{B_{i}}G_{i}^{b}$ with $B_i$ irreducible sub-matrices $G_{i}^{b}$. The set of bit string basis states on which the irreducible sub-matrix $G_{i}^{b}$ acts is denoted by $S_{i}^{b}$, where $\cup_{b}S_{i}^{b}\subseteq \{0,1\}^{n}$. 

\item From the Perron-Frobenius Theorem (Theorem 8.4.4 in \cite{HornJohnson}) it follows that for each nonnegative and irreducible sub-matrix there exists a unique and strictly positive eigenstate associated with its largest eigenvalue, i.e. 
\begin{equation}
\ket{\phi_i^b}= \sum_{x\in S_{i}^{b}}\phi_{i}^{b}(x)\ket{x},\quad  G_{i}^{b}\ket{\phi_i^b} = \lambda_{i}^{b}\ket{\phi_i^b},
\label{def:largest}
\end{equation}
where $\phi_{i}^{b}(x) > 0$, $\forall x \in S_{i}^{b}$. Since the spectrum of $G_i$ is the union of spectra of the submatrices $G_i^b$, the spectrum of each $G_i^b$ also lies in the interval $(0,1]$ and one of the blocks $b$ will contain the largest eigenvalue of $G_i$ equal to 1. In case $G_i$ is irreducible itself, there is a largest nonnegative eigenvector as in Eq.~\eqref{def:largest} which has support $\phi_i(x) > 0$ for all $x$. In this case, the corresponding eigenvalue will be $\lambda_i=1$.

\item Naturally, since $G_i$ acts nontrivially only on a subset of $O(1)$ qubits (and acts as $I$ on other qubits) one can efficiently compute the blocks $G_i^b$, its largest eigenvalue $\lambda_{i}^{b}$ and associated eigenstate $\ket{\phi_i^b}$ in each block $b$.
\end{itemize}
\label{prop:PF}
\end{proposition}

\bigbreak
\noindent
We prove the following:
\bigskip
\begin{lemma}
Let $\mathcal{F}  \equiv \bra{\Phi} G_{1}G_{2}\: ...\: G_{L} \ket{\Phi}$, 
where:
\begin{enumerate}
    \item $\ket{\Phi} = \sum_{x=1}^{2^{n}}\Phi(x)\ket{x}$ is a normalized state of $n$ qubits where $\Phi(x) \in \mathbb{C}$ ($\forall x$) and $\sum_{x}\bigl\lvert \Phi(x) \bigr\rvert^{2} = 1$. We assume that (1) $\frac{\Phi(y)}{\Phi(x)}$ can be efficiently (${\rm poly}(n)$) calculated for a {\em given} $x$ and $y$ and (2) we can efficiently draw samples from the probability distribution $P(x) = \bigl\lvert \Phi(x) \bigr\rvert^{2}$.
    \item Each $G_{l}=G_l$ is a $k$-local, positive-definite, (elementwise) nonnegative matrix with eigenvalues in $(0,1]$.
\end{enumerate}
$\mathcal{F}$ can be estimated within error $\epsilon$ with probability at least $1-\delta$ with a classical MC algorithm with runtime $\Theta(\epsilon^{-2}\text{log}(\delta^{-1}))\times \text{poly}(n) \times \Theta(L)$. 
\label{prop:imtimesignalest}
\end{lemma}

\bigskip
\begin{proof} The proof of Lemma \ref{prop:imtimesignalest} consists of two steps: To construct an estimator for $\mathcal{F}(\tau)$ and to show that the error of this estimator can be bounded according to the lemma. 

We rewrite the quantity of interest $\mathcal{F}$ as follows (where $L-1$ complete sets of basis states are inserted in between the $G_{l}$ operators in the final equality):
\begin{equation}
    \mathcal{F} =  
    \sum_{x_{0},x_{1},...,x_{L}} \lvert\Phi(x_{0})\rvert^{2}\frac{\Phi(x_{L})}{\Phi(x_{0})}\bra{x_{0}}G_{1}\ket{x_{1}}\bra{x_{1}}G_{2}\ket{x_{2}}\: ...\: \bra{x_{L-1}}G_{L}\ket{x_{L}},
\end{equation}
where we have set $\ket{x}=\ket{x_{0}}$ and $\ket{y}=\ket{x_{L}}$. $\mathcal{F}$ thus corresponds to the sum of an exponential number of products of (non-negative) matrix elements of $G_{1},...,G_{L}$, weighted by amplitudes in the state $\ket{\Phi}$. Evidently, only terms for which all the matrix elements in the product are non-zero contribute to the sum. 

We now consider the string of basis states $\ket{x_{0}},...,\ket{x_{L}}$ and associate with each step $\ket{x_{l-1}}$ to $\ket{x_{l}}$ in this string a probability 
\begin{equation}
P_{l}(x_{l-1}\to x_{l}) = \frac{1}{\lambda_{l}^{b}}\bra{x_{l-1}}G_{l}\ket{x_{l}}\frac{\phi_{l}^{b}(x_{l})}{\phi_{l}^{b}(x_{l-1})},
\end{equation}
where $b$ labels the sub-block in $G_l=\oplus_b G_l^b$ which contains the strings $x_{l-1}$ and $x_l$. Here $\phi_{l}^{b}(x) \equiv \braket{x_{l-1}}{\phi_l^b}$ with $\ket{\phi_l^b}$ defined in Proposition \ref{prop:PF}.


The probability distribution $P_{l}$ is thus non-negative as $\lambda_{l}^{b}\in (0,1]$, $G_{l}$ is element-wise non-negative and $\phi_l^b(x_l)>0$ and $\phi_l^b(x_{l-1})> 0$.
It can be shown to be normalized:
\begin{multline}
    \sum_{x_l}P_{l}(x_{l-1}\to x_{l}) = \sum_{x_l}\frac{1}{\lambda_{l}^{b}}\bra{x_{l-1}}G_{l}\ket{x_{l}}\frac{\phi_{l}^{b}(x_{l})}{\phi_{l}^{b}(x_{l-1})} = \sum_{x_{l}\in S_{l}^{b}}\frac{1}{\lambda_{l}^{b}}\bra{x_{l-1}}G_{l}^{b}\ket{x_{l}}\frac{\phi_{l}^{b}(x_{l})}{\phi_{l}^{b}(x_{l-1})} \\
    = \frac{1}{\phi_{l}^{b}(x_{l-1})}
    \frac{1}{\lambda_{l}^{b}}\bra{x_{l-1}}G_{l}^{b} \ket{\phi_{l}^b} = \frac{1}{\phi_{l}^{b}(x_{l-1})}
    \frac{1}{\lambda_{l}^{b}}\bra{x_{l-1}}\lambda_{l}^{b} \ket{\phi_{l}^{b}} = \frac{\braket{x_{l-1}}{\phi_{l}^{b}}}{\phi_{l}^{b}(x_{l-1})} = 1.
\end{multline}
We use $P_{l}(x_{l-1}\to x_{l})$ to rewrite $\mathcal{F}(\tau)$ as:
\begin{multline}
    \mathcal{F}(\tau) = \sum_{x_{0},x_{1},...,x_{L}} \underbrace{\lvert\Phi(x_{0})\rvert^{2} P_{1}(x_{0}\to x_{1})P_{2}(x_{1}\to x_{2})\: ...\: P_{L}(x_{L-1}\to x_{L})}_{\equiv \Pi(\boldsymbol{x}).}  \times \underbrace{\frac{\Phi(x_{L})}{\Phi(x_{0})} \prod_{l=1}^{L}\lambda_{l}^{b(l)}\frac{\phi_{l}^{b(l)}(x_{l-1})}{\phi_{l}^{b(l)}(x_{l})}}_{\equiv \mathcal{R}(\boldsymbol{x}).},
\end{multline}
where $\boldsymbol{x} \equiv (x_{0},x_{1},...,x_{L})$ and we have defined the quantities $\Pi(\boldsymbol{x})$ and $\mathcal{R}(\boldsymbol{x})$. Since $\lvert\Phi(x_{0})\rvert^{2}$ and each $P_{l}$ are probability distributions, $\Pi(\boldsymbol{x})$ is a probability distribution as well, i.e.
\begin{multline}
    \sum_{x_{0},x_{1},...,x_{L}}\Pi(\boldsymbol{x}) = \sum_{x_{0}}\Big(\lvert\Phi(x_{0})\rvert^{2}\:\sum_{x_{1}}\Big( P_{1}(x_0\to x_1) \:... \sum_{x_{L}}\Big(P_{L}(x_{L-1}\to x_{L})\Big)\:...\: \Big)\Big) = 1.
\end{multline}
Clearly, one can sample from $\Pi(\boldsymbol{x})$ by first sampling from $|\Phi(x_0)|^2$, then sampling from $P_1(x_0 \rightarrow x_1)$ to generate $x_1$ etc. until $x_L$.

By thus sampling from the probability distribution $\Pi(\boldsymbol{x})$ and obtaining a mean estimator for $\mathcal{F}(\tau)$ using the samples $\mathcal{R}(\boldsymbol{x})$, we can estimate $\mathcal{F}(\tau)$. We note that $\mathcal{F}(\tau) = \mathbb{E}\big(\mathcal{R}(\boldsymbol{x})\big)$. Since $\mathcal{R}(\boldsymbol{x}) \in \mathbb{C}$, a mean estimator over a finite number of samples will generally be complex-valued. Since $\mathcal{F}(\tau) \in \mathbb{R}$, we will instead obtain a mean estimator using samples $\text{Re}(\mathcal{R}(\boldsymbol{x}))$. The mean estimator that we shall use to estimate $\mathcal{F}(\tau)$ is the median-of-means estimator \cite{MOM}. Using a set $\Sigma$ of samples $\{\boldsymbol{x}\}$ (distributed according to $\Pi(\boldsymbol{x})$), the median-of-means estimator is defined as follows: Divide the set $\Sigma$ into $q$ subsets $s_{1},\ldots,s_{q}$ of size approximately $|\Sigma|/q$. Calculate the empirical mean of $\text{Re}(\mathcal{R}(\boldsymbol{x}))$ over the samples in each subset: $f_{j} = \frac{1}{|s_{j}|}\sum_{\boldsymbol{x}\in s_{j}}\text{Re}(\mathcal{R}(\boldsymbol{x}))$ for $j\in \{1,\ldots,q\}$ (each $f_{j}$ is an unbiased estimator of $\mathcal{F}(\tau)$). Now the median-of-means estimator is given by the median of these empirical means: $\hat{\mathcal{F}} = \mathrm{M}(f_{1},\ldots,f_{q})$. See \ref{app:MOM} for more details.

The algorithm that efficiently produces $\hat{\mathcal{F}} = \mathrm{M}(f_{1},\ldots,f_{q})$ is explicitly given in Algorithm \ref{algorithm1paper} below. 
Note that when $\Phi(x_0)$ is small for some $x_0$, the probability of drawing this $x_0$, $\lvert\Phi(x_{0})\rvert^{2}$, is very small, but the ratio $\frac{\Phi(x_{L})}{\Phi(x_{0})}$ in the estimator could get very large.

\smallbreak
\begin{algorithm}[H]
\textbf{Input:} Initial state $\ket{\Phi}$. Local propagation operators $\{G_{l}\}_{l=1}^{L}$. Sample size $|\Sigma|$. Number of subsets $q$. \\
\textbf{Output:} Median of means estimate of $\mathcal{F}(\tau)$. \\
\smallbreak
\SetAlgoLined
\For{$\sigma\in \{1,2,...,|\Sigma|\}$}{
Sample an initial basis state $\ket{x_{0}}$ from the probability distribution $\lvert\Phi(x_{0})\rvert^{2}$. The state $\ket{x_{0}}$ is part of $S_{1}^{b(1)}$. \\

\For{$l\in \{1,...,L\}$}{
Pick a state $\ket{x_l} \in S_{l}^{b(l)}$ with probability $P_{l}(x_{l-1}\to x_{l}) = \frac{1}{\lambda_{l}^{b(l)}}\bra{x_{l-1}}G_{l}\ket{x_{l}}\frac{\phi_{l}^{b(l)}(x_l)}{\phi_{l}^{b(l)}(x_{l-1})}$. The state $\ket{x_{l}}$ is part of $S_{l+1}^{b(l+1)}$ (for $l<L$).
}

Given $\{x_{l}\}_{l=1}^{L}$ (sampled from $\Pi(\boldsymbol{x})$), calculate $\mathcal{R}_{\sigma}(\boldsymbol{x}) = \frac{\Phi(x_{L})}{\Phi(x_{0})} \prod_{l=1}^{L}\lambda_{l}^{b(l)}\frac{\phi_{l}^{b(l)}(x_{l-1})}{\phi_{l}^{b(l)}(x_{l})}$.
}

Divide the $|\Sigma|$ samples into $q$ subsets $s_{1},\ldots,s_{q}$, such that $|s_{j}|\approx |\Sigma|/q$, $\forall j$.

\For{$j \in \{1,2,\ldots,q\}$}{Calculate $f_{j} = \frac{1}{|s_{j}|}\sum_{\boldsymbol{x}\in s_{j}}\text{Re}(\mathcal{R}(\boldsymbol{x}))$.}

Output $\hat{\mathcal{F}}=\mathrm{M}(f_{1},\ldots,f_{q})$.
\caption{Efficiently obtaining a median-of-means estimate of $\mathcal{F}(\tau)$ through sampling of the probability distribution $\Pi(\boldsymbol{x})$.}
\label{algorithm1paper}
\end{algorithm}
\smallbreak

Algorithm \ref{algorithm1paper} thus efficiently provides an estimate of $\mathcal{F}(\tau)$ (albeit biased). To complete the proof, we will show that the variance of $\text{Re}(\mathcal{R}(\boldsymbol{x})) \in \mathbb{R}$ can be bounded which in turn is used to bound the number of samples to get an estimate close to the mean, leading to Lemma \ref{prop:imtimesignalest}. 

For a complex random variable $Z=\mathcal{R}(\boldsymbol{x})$, $\mathbb{E}(Z) \equiv \mathbb{E}\big(\text{Re}(Z)\big) + i\mathbb{E}\big(\text{Im}(Z)\big)$ and $\text{Var}\big(Z\big) = \text{Var}\big( \text{Re}(Z) \big) + \text{Var}\big(\text{Im}(Z)) \big) \geq \text{Var}\big( \text{Re}(Z) \big)$. Hence we can bound the variance of random variable $\text{Re}(\mathcal{R}(\boldsymbol{x}))$ by bounding the variance of the random variable $\mathcal{R}(\boldsymbol{x})$. This variance is given by:
\begin{equation}
    \text{Var}(\mathcal{R}(\boldsymbol{x})) =
    \mathbb{E}\Big( \bigl\lvert \mathcal{R}(\boldsymbol{x}) \bigr\rvert^{2} \Big) - \underbrace{\Bigl\lvert \mathbb{E}\Big( \mathcal{R}(\boldsymbol{x}) \Big)\Bigr\rvert^{2}}_{=\: \lvert \mathcal{F}(\tau) \rvert^{2} \:=\:\mathcal{F}(\tau)^{2}.} \leq \mathbb{E}\Big( \bigl\lvert \mathcal{R}(\boldsymbol{x}) \bigr\rvert^{2} \Big),
\label{eq:varboundR}
\end{equation}
where the inequality holds because $\mathcal{F}^{2} \geq 0$ (since $\mathcal{F} \in \mathbb{R}$). To obtain an upper bound on the variance, we shall investigate this expression in more detail:
\begin{align}
\begin{split}
    \mathbb{E}\bigg( \bigl\lvert \mathcal{R}(\boldsymbol{x}) \bigr\rvert^{2} \bigg) &=\: \sum_{\boldsymbol{x}}\Pi(\boldsymbol{x})\bigl\lvert \mathcal{R}(\boldsymbol{x}) \bigr\rvert^{2} \\ &=\: \sum_{\boldsymbol{x}}\lvert\Phi(x_{L})\rvert^{2}\bra{x_{0}}G_{1}\ket{x_{1}}\bra{x_{1}}G_{2}\ket{x_{2}}\: ...\: \bra{x_{L-1}}G_{L}\ket{x_{L}}\prod_{l=1}^{L}\lambda_{l}^{b(l)}\frac{\phi_{l}^{b(l)}(x_{l-1})}{\phi_{l}^{b(l)}(x_{l})} \\ &=\: \sum_{\boldsymbol{x}}\lvert\Phi(x_{L})\rvert^{2}\:Q_{1}(x_{0},x_{1})Q_{2}(x_{1},x_{2})\: ...\: Q_{L}(x_{L-1},x_{L}), 
\end{split}
\label{eq:varbound}
\end{align}
where in the last equality we defined the non-negative quantity $Q_{l}(x,y) \equiv \bra{x}G_{l}\ket{y}\lambda_{l}^{b(l)}\frac{\phi_{l}^{b(l)}(x)}{\phi_{l}^{b(l)}(y)}$. Exploiting the Hermiticity of $G_{l}^{b}$, $Q_{l}(x,y)$ can be shown to have the following property:
\begin{equation}
    \sum_{x}Q_{l}(x,y) = \sum_{x} \bra{x}G_{l}\ket{y}\lambda_{l}^{b}\frac{\phi_{l}^{b}(x)}{\phi_{l}^{b}(y)} = \sum_{x\in S_{l}^{b}} \bra{x}G_{l}^{b}\ket{y}\lambda_{l}^{b}\frac{\phi_{l}^{b}(x)}{\phi_{l}^{b}(y)} = 
    \big(\lambda_{l}^{b}\big)^{2}\frac{\braket{\phi_l^b}{y}}{\phi_{l}^{b}(y)} = \big(\lambda_{l}^{b}\big)^{2} \leq 1.
    \label{eq:propQpaper}
\end{equation}
$Q_{l}(x,y)$ thus satisfies $0\leq Q_{l}(x,y) \leq 1$, $\forall x,y$ and $\forall l\in \{1,2,...,L\}$. By consecutively exploiting the property in Eq.~\eqref{eq:propQpaper} for all $Q_{l}$'s and the normalization property of state $\ket{\Phi}$ in the expression in Eq.~\eqref{eq:varbound}, we obtain
\begin{equation}
    \text{Var}(\mathcal{R}({\bf x})) \leq  \mathbb{E}\bigg( \bigl\lvert \mathcal{R}(\boldsymbol{x}) \bigr\rvert^{2} \bigg) \leq 1. \:\: \Rightarrow \:\: \text{Var}\Big(\text{Re}(\mathcal{R}({\bf x}))\Big) \leq 1.
\label{eq:unitybound}
\end{equation}

If we take the number of samples $|\Sigma|$, and divide them into $q$ subsets $s_{1},\ldots,s_{q}$ of size approximately $|\Sigma|/q$, then (by means of Chebyshev's inequality) each $f_{j}$ obeys $\lvert f_{j} - \mathcal{F} \rvert \leq \sqrt{\text{Var}\big(\text{Re}(\mathcal{R}(\boldsymbol{x}))\big)}\sqrt{4q/|\Sigma|} \leq \sqrt{4q/|\Sigma|}$ with probability at least $3/4$. Using Hoeffding's inequality and the definition of the mean, one can show that (see \ref{app:MOM}):
\begin{equation}
    \text{Pr}\Big( \bigl\lvert \hat{\mathcal{F}} - \mathcal{F} \bigr\lvert \leq \sqrt{4q/|\Sigma|} \Big) \geq 1-e^{-q/8}. 
\end{equation}
Hence $\mathcal{F}$ can be estimated with error $\epsilon$ with probability at least $1-\delta$ (with $q=8\:\text{log}(\delta^{-1})$) for $|\Sigma| = \Theta(\text{log}(\delta^{-1})\epsilon^{-2})$, where obtaining each sample takes a number of operations that scales linearly in $L$ and ${\rm poly}(n)$. This completes the proof of Lemma \ref{prop:imtimesignalest}.
\end{proof}

\begin{remark}
Note that if one would have chosen the empirical mean $\tilde{\mathcal{F}}=\frac{1}{|\Sigma|}\sum_{\boldsymbol{x}\: \in \Sigma}\text{Re}(\mathcal{R}(\boldsymbol{x}))$ as a mean estimator for $\mathcal{F}$ (instead of the median-of-means estimator), then using Eq.~\eqref{eq:unitybound} and Chebyshev's inequality, we obtain:
\begin{equation}
    \text{Pr}\Big( \bigl\lvert \tilde{\mathcal{F}} - \mathcal{F} \bigr\lvert \leq \epsilon \Big) \geq 1-\frac{{\rm Var}(\text{Re}(\tilde{\mathcal{F}}))}{\epsilon^2}\geq 1-\frac{1}{|\Sigma|\:\epsilon^2}.
\end{equation}
Hence $\mathcal{F}$ can be estimated using $\tilde{\mathcal{F}}$ with error $\epsilon$ with probability at least $1-\delta$, for $|\Sigma|=\Theta(\delta^{-1}\epsilon^{-2})$. Using the median-of-means estimator thus provides an exponential improvement in the required scaling of $|\Sigma|$ with $\delta^{-1}$. Note that if we could upper and lower bound the range of ${\rm Re}(\mathcal{R}(\mathbf{x}))$ by some constants, then we could have used a Chernoff-Hoeffding bound for the empirical mean $\tilde{\mathcal{F}}$ which gives the aforementioned (exponentially) better dependence of the run-time of the algorithm with $\delta^{-1}$ (as in Lemma \ref{prop:retimesignalest} where we do use a Chernoff-Hoeffding bound). 
\end{remark}

\begin{remark}
Note that the Lemma also applies to estimating $\bra{x} G_{1}G_{2}\: ...\: G_{L} \ket{x'}$ (with $1/{
\rm poly}(n)$ accuracy) as one simply starts the process at $x_0=x$ and $\mathcal{R}(\mathbf{x})$ is only nonzero when one arrives at $x_L=x'$. Similarly, one can estimate $\bra{\Phi_1} G_{1}G_{2}\: ...\: G_{L} \ket{\Phi_2}$ with $1/{\rm poly}(n)$ accuracy, assuming one can sample from $|\Phi_1(x)|^2$ (or $|\Phi_2(x)|^2$) {\em and} compute for a given $x$ and $y$, the ratio 
$\frac{\Phi_2(y)}{\Phi_1(x)}$. In addition, one can extend the Lemma to the case where the local propagation operators $G_{l}$ are not Hermitian, but are still nonnegative matrices, see \ref{AppC}.
\end{remark}

We stress that Lemma \ref{prop:imtimesignalest} provides an efficient classical algorithm provided that: For a given $x,y\in \{0,1\}^{n}$, one can efficiently determine $\frac{\Phi(y)}{\Phi(x)}$ and the state $\ket{\Phi}$ is such that one can efficiently draw samples from $P(x) = \bigl\lvert \Phi(x) \bigr\rvert^{2}$. In many practical settings, $\ket{\Phi}$ is such that one can define a function $f: \{0,1\}^{n} \to \mathbb{C}$ which takes as input the $n$-bit string $x$, and efficiently outputs the corresponding coefficient $\Phi(x)$. This is e.g. the case for (matrix) product states or for other ansatz classes of states. Then, given $x$ and $y$, the fraction $\frac{\Phi(y)}{\Phi(x)}$ can be efficiently obtained. 
Note that under this assumption one can set-up a Monte Carlo scheme based on the Metropolis algorithm to sample from $\bigl\lvert \Phi(x) \bigr\rvert^{2}$, although this scheme 
is only a heuristic strategy and its efficient convergence would have to be proved. A good class of states to which both Lemmas \ref{prop:imtimesignalest} and \ref{prop:retimesignalest} apply are of course product states.
Note that even when running the overlap test is too costly (as quantum circuits are noisy), but preparing the state $\ket{\Phi}$ is feasible, one could use this preparation to sample from $|\Phi(x)|^2$ for the application of the MC method. Of course the requirement of being able to compute $\frac{\Phi(y)}{\Phi(x)}$ remains. For the transverse field Ising model, an example of a $\ket{\Phi}$ which obeys these conditions will be given in Section \ref{simsection}. \\

{\em Proof of Theorems \ref{thm:re} and \ref{thm:im}}:
We require the Trotterization of $e^{-i k H}$ resp. $e^{-k H}$ into a string of local propagation operators $G_i$ which are unitary (in Lemma \ref{prop:retimesignalest}) resp.
Hermitian and non-negative (in Lemma \ref{prop:imtimesignalest}).
This non-unique decomposition of $e^{-ik H}$ and $e^{- k H}$ into an ordered string of local propagation operators depends on the Trotterization scheme and is discussed in \ref{sec:trotter} (and more extensively in \cite{childs+:trotter}). The Trotterization gives an error $\epsilon_{\rm trot}$ (in addition to the sampling error $\epsilon$ in Lemmas \ref{prop:retimesignalest} and \ref{prop:imtimesignalest}) and the number of local propagation operators (for each sample) $L$ in Lemmas \ref{prop:retimesignalest} and \ref{prop:imtimesignalest} will be $L=\text{poly}(n)\:\mathcal{O}\big(\Upsilon\:k^{1+1/p}\epsilon_{\text{trot}}^{-1/p}\big)$ (for real time) and $L=\text{poly}(n)\:\mathcal{O}\big(\Upsilon\:k^{1+1/p}\epsilon_{\text{trot}}^{-1/p}\big)$ (for imaginary time, provided that $M\geq 4\tau \Upsilon \big( \sum_{\gamma} \norm{H_{\gamma}} \big)$, where $M$ is the Trotter variable). 
$\Upsilon$ denotes the number of \textit{stages} in the Trotterization scheme of order $p$, and typically scales exponentially in $p$ (but $p$ is chosen a constant). For given order $p = O(1)$ of the Trotterization scheme, $L$ in Lemma \ref{prop:imtimesignalest} thus scales with the length of the time interval over which the system is simulated as $k^{1+o(1)}$ and $k^{1+o(1)}$, and with the imposed Trotter error as $\epsilon_{\text{trot}}^{-o(1)}$.
Then, if we wish to estimate $g_R(k)$ and $g_I(k)$ at multiple $k=0,\ldots, K$, we use that the probability that all $K$ estimates are up to uncertainty $\epsilon$ equals unity minus the probability that at least one of the estimates is beyond $\epsilon$ (which, by the union bound, is at most $K \delta$).

Finally, before we move on to extracting eigenenergy estimates from the (real-time and imaginary-time) signals using the ESPRIT method, we prove the Lemma related to the signal $g_D(k)$ in Eq.~\eqref{eq:signal-MC-D}. 

\bigskip
\begin{lemma}
Let $g_D(k)$ be defined as in Eq.~\eqref{eq:signal-MC-D} for a local $n$-qubit Hamiltonian $H$, with $E_j$ in $[0,\pi]$, and assume that (1) one can efficiently (i.e. with ${\rm poly}(n)$ effort) sample from $|\Phi(x)|^2$, and (2) given $x$ and $y$, one can compute $\Phi(y)/\Phi(x)$ efficiently. Then, $g_D(k)$ can be classically estimated within error $\epsilon$ with probability at least $1-\delta$ with $[{\rm poly}( n)]^{k}\times \Theta(\epsilon^{-2}\log(\delta^{-1}))$ classical computational effort.
\label{lem:SVT}
\end{lemma}

\bigskip
\begin{proof}
By definition of $g_{D}(k)$, we can write
\begin{align}
    g_D(k)=\sum_{x,y}|\Phi(x)|^2 \frac{\Phi(y)}{\Phi(x)} \bra{x} (I-H/2\pi)^k \ket{y}.
\end{align}
To estimate $g_D(k)$, one first draws an $x$ from $P(x)=|\Phi(x)|^2$, and then one collects all $y$ which are obtained after the application of $(I-H/2\pi)^k$ to $\bra{x}$. Each application of $I-H/2\pi$ maps the input string onto at most ${\rm poly}(n)$ new output strings, hence one obtains at most $[{\rm poly}(n)]^k$ such $y$'s after $k$ applications. Let $\boldsymbol{x}=(x_k=y,x_{k-1},\ldots, x_1, x_0=x)$ be a particular path of strings and let 
\begin{align}
    {\cal R}(x)\equiv \:&
    \sum_{y}\frac{\Phi(y)}{\Phi(x)}\bra{x}(I-H/2\pi)^{k}\ket{y} \notag \\ = \:&
    \sum_{x_1,\ldots, x_{k-1},y} \frac{\Phi(y)}{\Phi(x)}\bra{x} I-H/2\pi \ket{x_{1}}\bra{x_1} I-H/2\pi \ket{x_2} \ldots \bra{x_{k-1}} I-H/2\pi \ket{y},
\end{align}
so that $g_D(k)=\sum_{x} |\Phi(x)|^2 \, \mathcal{R}(x)$.
For each $x$ that is sampled from $P(x)$, one thus computes and outputs ${\rm Re}(\mathcal{R}(x))$ by summing over the contributions from \textit{all} paths $\boldsymbol{x}$ that start at string $x$. As in Lemma \ref{prop:imtimesignalest}, we need to establish how many samples $|\Sigma|$ we need to draw from $P(x)$ to obtain $g_D(k)$ within error $\epsilon$ with probability at least $1-\delta$. This analysis depends on the variance of the complex variable $\mathcal{R}(x)$ through Eq.~\eqref{eq:varboundR}, requiring us to upper bound \begin{align}
    \mathbb{E}\big( |\mathcal{R}(x)|^2\big)= \:& \sum_{x} |\Phi(x)|^2\:\bigg( \sum_{y} \frac{\Phi^{*}(y)}{\Phi^{*}(x)}\bra{y}(I-H/2\pi)^{k}\ket{x}\bigg)\bigg( \sum_{y'} \frac{\Phi(y')}{\Phi(x)}\bra{x}(I-H/2\pi)^{k}\ket{y'}\bigg) \notag \\ = \:& \sum_{x} \bra{\Phi}(I-H/2\pi)^{k}\ket{x}\bra{x}(I-H/2\pi)^{k}\ket{\Phi} \notag \\ = \:& \bra{\Phi} (I-H/2\pi)^{2k} \ket{\Phi} \leq 1,
\end{align}
where in the final line we have used that the eigenvalues of $H$ lie in $[0,\pi]$.
As in the proof of Lemma \ref{prop:imtimesignalest}, this establishes that ${\rm Var}\big({\rm Re}(\mathcal{R}(x))\big)\leq 1$. Then we can use the median-of-means estimator as in the proof of Lemma \ref{prop:imtimesignalest} and \ref{app:MOM} to establish that with probability at least $1-\delta$, $g_D(k)$ can be estimated with error at most $\epsilon$, taking $|\Sigma|=\Theta(\epsilon^{-2}\log(\delta^{-1}))$ samples from $P(x)=|\Phi(x)|^2$, and with $[{\rm poly}(n)]^k$ computational effort per sample.
\end{proof}

It is important to note that unlike in Lemma \ref{prop:imtimesignalest}, here we only sample $x$ and compute the rest as the estimator $\mathcal{R}(x)$, while in Lemma \ref{prop:imtimesignalest} we sample the whole path of length $L$. This is why the computational effort in Lemma \ref{prop:imtimesignalest} is efficient (linear) in $L$ and thus polynomial in $k$, while in Lemma \ref{lem:SVT} the computational effort is exponential in $k$. This is thus the difference between the stoquastic Hamiltononian case versus the general Hamiltonian case. Note also that one can take each $G_i$ in Lemma \ref{prop:imtimesignalest} to be $G=I-H/2\pi$ in principle, as it obeys condition (ii) when $H$ is stoquastic.

We note that in \cite{dequant:GG} the sampling-access assumption is formulated slightly differently, that is, one gets access to $\Phi(x)$ for a given $x$, which can be stronger than only knowing the ratio $\Phi(x)/\Phi(y)$ for a given $x$ and $y$. In addition, Ref.~\cite{dequant:GG} allows an additional error in the sampling access whereas we gloss over this here and assume perfect sampling-access (similar to the exact assumptions in the other Lemmas). 


\section{Classically processing the signal: the ESPRIT method}
\label{sec:MPM-analysis}

We turn to discussing the ESPRIT method \cite{li2019superresolution} which is a method like the matrix pencil method \cite{Sarkar,Sarkar2,Sarkar3} for processing a signal as in Eqs.~\eqref{eq:signal-QPE} and \eqref{eq:signal-MC} consisting of $S$ components. Indeed, suppose a set of values for the signal $g(k)$,
\begin{equation}
    g(k) = \sum_{j=1}^{S}c_{j}z_{j}^{k},
\end{equation}
where $|z_j|\leq 1$, for $k \in \{0,1,...,K\}$, $K$ even. The goal is to determine the $z_j$ and the coefficients $c_j > 0$ \footnote{Here we focus on determining the $z_j$, but given the $z_j$ one can determine the $c_j$ as well and methods for analyzing the performance also exist for this \cite{Moitra}.} using $g(k)$ for sufficiently many $k$. 
In case of the real-time signal $g_R(k)$, we have $z_j\equiv e^{-i E_j}$, in case of a purely-decaying imaginary-time signal $g_I(k)$, we have $z_{j} \equiv e^{-E_{j}} \in (e^{-2\pi},1]$ and for the purely-decaying signal $g_D(k)$ we have $z_j=\left(1-E_j/(2\pi)\right)^2$.

Due to sampling and Trotter noise, one is effectively given a noisy signal $y(k)$ (for $k \in \{0,1,...,K\}$), which is related to the original signal $g(k)$ by:
\begin{equation}
    y(k) := g(k) + \eta(k) = \sum_{j=1}^{S}c_{j}z_{j}^{k} + \eta(k),
\label{eq:noisysignal}
\end{equation}
where $\eta(k)$ denotes e.g. the sampling and Trotter noise, and we have $|\eta(k)| \leq \epsilon_{\rm tot}$ in Theorem \ref{thm:re} and \ref{thm:im} with high probability.

It is well-known that for a noiseless signal ($\eta(k)=0$), the $z_{j}$'s and the $c_j$'s can be resolved perfectly via ESPRIT and the matrix pencil method if we take $K+1\geq 2S$. Importantly, this result does not depend on whether the signal is oscillatory or decaying. For illustration, Figure \ref{MPMtest} depicts the results of application of the matrix pencil method to a noiseless signal. We consider separately a decaying signal and an oscillating signal, and for both cases we depict respectively the estimates of the decay rates and oscillation frequencies as a function of $K$. When $K+1\geq 2S$, the eigenvalues are indeed resolved both for the decaying and the oscillating signal. When $K+1<2S$, the eigenvalues are not resolved. We will see however, that in the presence of noise a decaying or oscillatory signal fares very differently.

\begin{figure}[b]
    \centering
    \includegraphics[width=0.55\linewidth]{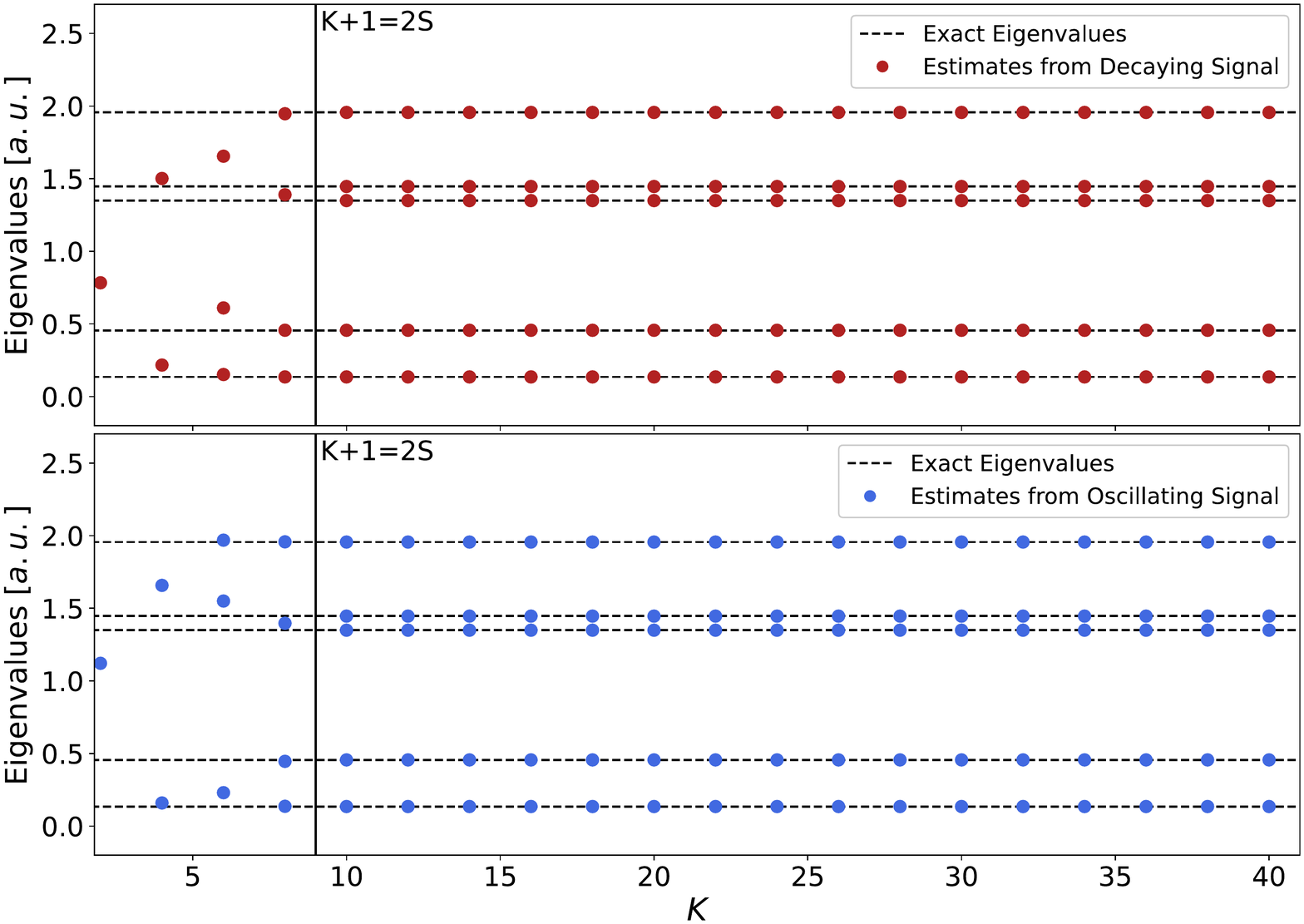}
    \caption{Estimates of the decay rates (of a decaying signal) and oscillation frequencies (of an oscillating signal) as a function of $K$. The estimates are obtained from applying the matrix pencil method \cite{Sarkar, Sarkar2, Sarkar3} to the noiseless signals $g(k) = \sum_{j=1}^{S}c_{j}z_{j}^{k}$, where $z_{j} = e^{-E_{j}}$ (in the case of the decaying signal) and $z_{j} = e^{-iE_{j}}$ (in the case of the oscillating signal) for $k=0,1,\ldots, K$ where $E_j \in [0,2\pi)$. All $c_{j}$'s are set equal to $1/S$ and the $E_{j}$'s have been randomly produced. The eigenvalues are recovered for $K+1\geq 2S$. }
    \label{MPMtest}
\end{figure}

Let us consider in more detail the task of obtaining the $z_{j}$'s from the signal $y(k)$ in Eq.~\eqref{eq:noisysignal}. The key object of study here is the Hankel matrix $H(y):= H(g)+H(\eta)$, containing all $K$ data points of the noisy signal $y(k)$ and a positive integer `matrix pencil' parameter $L$:
\begin{multline}
    H(y) = 
    \begin{pmatrix}
    y(0) & y(1) & \dots & y(K-L) \\
    y(1) & y(2) & \dots & y(K-L+1) \\
    \vdots & \vdots & & \vdots \\
    y(L) & y(L+1) & \dots & y(K)
    \end{pmatrix}_{(L+1)\times(K-L+1)} \\= \sum_{j=1}^{S}c_{j}
    \begin{pmatrix}
    1 & z_{j} & \dots & z_{j}^{K-L} \\
    z_{j} & z_{j}^{2} & \dots & z_{j}^{K-L+1} \\
    \vdots & \vdots & & \vdots \\
    z_{j}^{L} & z_{j}^{L+1} & \dots & z_{j}^{K}
    \end{pmatrix}+H(\eta),
\label{eq:Ymatrix}
\end{multline}
where $H(\eta)$ is purely due to the noise and has norm $\norm{H(\eta)}$. We can decompose the Hankel matrix $H(g)$ of the noiseless signal in terms of Vandermonde matrices $V_L$:
\begin{equation}
    H(g) = \sum_{j=1}^{S}c_{j}
    \begin{pmatrix}
    1 & z_{j} & \dots & z_{j}^{K-L} \\
    z_{j} & z_{j}^{2} & \dots & z_{j}^{K-L+1} \\
    \vdots & \vdots & & \vdots \\
    z_{j}^{L} & z_{j}^{L+1} & \dots & z_{j}^{K}
    \end{pmatrix}_{(L+1)\times (K-L+1)} = V_{L} C V_{K-L}^T, 
\label{eq:Vandermondedecomp}
\end{equation}
where 
\begin{align}
    C \equiv \text{diag}(c_{1},c_{2},...,c_{S}),
    \label{def:matrixC}
\end{align} and $V_L$ is
\begin{equation}
    V_{L} = \begin{pmatrix}
    1 & 1 & \dots & 1 \\
    z_{1} & z_{2} & \dots & z_{S} \\
    \vdots & \vdots & & \vdots \\
    z_{1}^{L} & z_{2}^{L} & \dots & z_{S}^{L}
    \end{pmatrix}_{(L+1)\times S}. 
    \label{eq:V}
\end{equation}
In general, methods such as ESPRIT (see the ESPRIT Algorithm \ref{alg:two}) rely on the parameter $L$ and for convenience we will keep it general in some of the analysis (specifically in \ref{app:MPM}). Our results will, however, focus on the choice $L=K/2$.  For $L=K/2$, we have 
\begin{align} 
    H(y) = 
    \begin{pmatrix}
    y(0) & y(1) & \dots & y(K/2) \\
    y(1) & y(2) & \dots & y(K/2+1) \\
    \vdots & \vdots & & \vdots \\
    y(K/2) & y(K/2+1) & \dots & y(K)
    \end{pmatrix}_{(K/2+1)\times(K/2+1)}
\end{align}
and
\begin{align}
    H(g) = V_{K/2}CV_{K/2}^{T} \in \mathbb{C}^{(K/2+1)\times (K/2+1)}.
\end{align}
Making contact with error bounds in the previous section, we see that (for $L=K/2$) 
\begin{equation}
    \forall k,\; \eta(k)\leq \epsilon_{\rm tot} \Rightarrow \norm{H(\eta)}=\sigma_{\rm max}(H(\eta)) \leq \norm{H(\eta)}_F \leq K \epsilon_{\rm tot}.
    \label{eq:upper}
\end{equation}

From the `Vandermonde decomposition' in Eq. \eqref{eq:Vandermondedecomp} of the Hankel matrix encoding a real-time or imaginary-time signal, one can develop numerical algorithms to extract the the decay rates $z_i$. One such algorithm is ESPRIT (given in Algorithm \ref{alg:two}), which specifically exploits the relation between the Vandermonde decomposition of $H(y)$ and its singular value decomposition. 
\SetKwComment{Comment}{/* }{ */}
\begin{algorithm}
\caption{ESPRIT algorithm.}\label{alg:two}
\KwData{Time signal $y$, number of decay rates or oscillation frequencies $S$.}
\KwResult{List $\tilde{z}_1, \ldots ,\tilde{z}_S$.}
$K \gets \mathrm{length}(y)$\Comment*[r]{We will assume $K$ is even for simplicity.}
$L \gets K/2$\Comment*[r]{Not the most general choice, however it works well in practice.}
$H(y) \gets \text{Hankel matrix built from } y$\;
$\tilde{U}, \tilde{\Sigma}, \tilde{W} \gets \mathrm{SVD}(H(y))$\Comment*[r]{Make sure $\tilde{\Sigma}$ is decreasingly ordered.}
$\tilde{U}_S\gets \text{First }S\text{ columns of }\tilde{U}$\Comment*[r]{Remember $\tilde{U}$ is a $(L+1) \times (L+1)$ unitary matrix}
$\tilde{U}_0\gets \text{First }L\text{ rows of }\tilde{U}_S $\;
$\tilde{U}_1\gets \text{Last }L\text{ rows of }\tilde{U}_S $\;
$\tilde{\Psi}\gets\tilde{U}_0^+\tilde{U}_1 $\Comment*[r]{Make $S \times S$ signal matrix $\tilde{\Psi}$, $+$ denotes Moore-Penrose inverse.}
$\tilde{z}_1, \ldots, \tilde{z}_S \gets\text{eigenvalues of signal matrix } \tilde{\Psi}$.
\end{algorithm}

We will see that this algorithm comes with recovery guarantees on the parameters $z_1, \ldots, z_{S}$, in both the real-time and imaginary-time signal case, provided the noise vector $\eta$ is small enough. The strength of these guarantees differs significantly between the two types of signal, and we will discuss them separately in the next sections. 
From the $\tilde{z}_j$'s we can then (for both the real-time and imaginary-time signal) extract $\tilde{E_j}$'s, which denote the $S$ estimates for $\{E_i \in [0,2\pi)\}_{i=1}^S$ returned by the classical post-processing algorithm. The error in the energy estimates is set as the optimal matching distance \cite{bhatia2013matrix} 
\begin{equation}\label{eq:matching_error}
    d(\{E_i\},\{\tilde{E}_j\})=\frac{1}{2\pi}\min_{\pi \in \text{Perm}_S} \max_j |\tilde{E}_{\pi(j)}-E_j|,
\end{equation}
i.e. the returned list is optimally matched with the actual eigenvalues and the error is set by the largest mismatch.

\subsection{Real-time (oscillatory) signal}
In this section we discuss the performance of ESPRIT on real-time (oscillatory) signals. This performance has been well studied in the signal processing literature. Here, we will follow the analysis of \cite{li2019superresolution}, which provides Theorem \ref{thm:osc-esprit-gap} relating $\norm{H(\eta)}$ in Eq.~\eqref{eq:upper} and the energy matching error defined in Eq.~\eqref{eq:matching_error}. 

The performance of ESPRIT in the oscillatory signal case relies on lower bounding the smallest nonzero singular value of the Vandermonde matrix $V_{L=K/2}$ in Eq.~\eqref{eq:V}, (or similarly upperbounding the condition number $\kappa(V_{K/2})=\sigma_{\rm max}(V_{K/2})/\sigma_{\rm min}(V_{K/2})$). The smallest nonzero singular value of the Vandermonde matrix $V_{K/2}$ will depend on $K$, $S$ and the location of the poles $z_j$. For the real-time signal, the $z_j$ lie on the unit circle whereas for the imaginary-time signal the $z_j$ lie in the interval $(e^{-2\pi},1]$.
Let the minimal gap between the $E_i$ be defined as 
\begin{align}
    \Delta=\frac{1}{2\pi} \min_{j \neq k}|E_j-E_k|.
    \label{eq:gap}
\end{align}

It has been proved \cite{Moitra} for $z_j=e^{-i E_j}$ that
\begin{equation}
 \Delta\geq \frac{C}{K} \Rightarrow  \sigma_{\rm min}^2(V_{K/2})\geq \frac{C-1}{C} K, 
 \label{eq:gap-del}
\end{equation}
for some constant $C> 1$. Note that if there are $S$ eigenvalues $E_j \in [0,2\pi)$ in the signal, it is clear that the minimal gap $\Delta\leq 1/S$, hence one should at least take $K\geq C S$. Based on this bound, Theorem 4 in \cite{li2019superresolution} says:

\bigskip
\begin{theorem}[\cite{li2019superresolution}]
Let $(g+\eta)(k)$ be a real-time signal with $k=0,\ldots,K$, and with $g(k) = \sum_{i=1}^{S}c_{i}z_{i}^{k}$, $c_i > 0\: \forall i$, $c_{\min}=\min_i c_i$ and $\eta(k)$ a small noise vector. Let $z_j=e^{-i E_j}$ with $j=1, \ldots, S$ and $E_j \in [0,2\pi)$ $\forall j$, and $K \geq 2C/\Delta$ for some constant $C > 2$ with gap $\Delta$, and $K+1 \geq 2S$. 
If
\begin{equation}
    \norm{H(\eta)} \leq c_{\rm min}K \,h_1(S,C,K),
    \label{eq:condH-gap}
\end{equation}
with
\begin{equation}
    h_1(S,C,K)=\frac{C-1}{8\sqrt{2S} C}\sqrt{1-\frac{2CS}{(C-1)K}},
\end{equation}
then the ESPRIT algorithm outputs energy estimates $\{\tilde{E}_j\}$ with distance 
\begin{equation}
    d(\{E_i\}),\{\tilde{E}_j\})  \leq \norm{H(\eta)} c_{\rm min}^{-1} K^{-1}\, h_2(S,C,K),
    \end{equation}
    with
    \begin{equation}
    h_2(S,C,K)=
    40\sqrt{2} S^2  \left(\frac{C}{C-1}\right)^{3/2}\left(1-\frac{2CS}{(C-1)K} \right)^{-1}.
\end{equation}
\label{thm:osc-esprit-gap}
\end{theorem}

\bigskip
By Eq.~\eqref{eq:upper} we have $\norm{H(\eta)}\leq K \epsilon_{\rm tot}$ and if we choose $K \sim S$, $\epsilon_{\rm tot}$ can be chosen sufficiently small, inversely polynomial with $S$, such that at least Eq.~\eqref{eq:condH-gap} holds. Then $d(\{E_i\}),\{\tilde{E}_j\})$ will be $\Theta(\norm{H(\eta)}S)$, hence decreasing like $S^{2} \epsilon_{\rm tot}$.

If we combine this Theorem with the quantum results of Theorem \ref{thm:re}, then we obtain Theorem \ref{thm:QPE-total}. These results thus form the theoretical underpinning of the ideas and numerical work in \cite{TomBarbara} in which quantum phase estimation was replaced by the repeated execution of a circuit applying controlled-$U^k$ (conditioned on  an ancilla qubit state) which gets Trotterized to the overlap test circuit in Fig.~\ref{qpecircuit}.

\begin{remark}
It is noteworthy that even when the eigenvalues $E_j$ are not well-separated but occur in `clumps', results exist \cite{li2019superresolution} which bound the performance of ESPRIT. 
\end{remark}

\subsection{Imaginary-time (decaying) signal}

Let us now discuss what information can be extracted from the imaginary-time signal in the presence of sampling and Trotter noise and compare this to the known Theorem \ref{thm:osc-esprit-gap} for the real-time signal.

In \ref{app:MPM} we discuss in detail the recovery guarantees for ESPRIT for imaginary-time signals. This analysis is an adaptation of the work done in \cite{li2019superresolution} for real-time signals, with the only true novelty being Lemma \ref{lem:U_inv_bound}. However, since no rigorous analysis for imaginary-time signals exists in the literature we go through all the steps in considerable detail.  
The analysis will again depend on the condition number of the Vandermonde matrix $V_{L=K/2}$ in Eq.~\eqref{eq:V}.

This condition number is much worse behaved, i.e. much larger, in case the $z_i$'s all lie on the real axis --which is the case for the imaginary-time signal-- but bounds on this condition number do exist \cite{bazan}. Based on the work of Gautschi~\cite{gautschi1962inverses}, we derive our own upper bounds on this condition number, 
which are asymptotically sub-optimal but have a clearer dependence on the choice of $K$ and the given $S$ than previous bounds in \cite{bazan}. We then use the gap $\Delta$ to fill in the upper bound.

In analogy to Theorem \ref{thm:osc-esprit-gap}, we then obtain the following:

\bigskip
\begin{theorem}\label{thm:final}
Let $(g+\eta)(k)$ be an imaginary-time decaying signal with $k=0,\ldots,K$, and with $g(k) = \sum_{i=1}^{S}c_iz_i^k$, $c_i > 0, \:\forall i$, $c_{\min}=\min_i c_i$, and $\eta(k)$ a small noise vector. Let $z_i = e^{-E_i}$ with $E_{i} \in [0, 2\pi)$ and given eigenvalue gap $\Delta < 1$ in Eq.~\eqref{eq:gap}, and $\{\tilde{E}_i\}$ the energy estimates of ESPRIT with $L=K/2$. Let $K+1 \geq 2S$, $K$ even and $K=TS$ for some positive integer $T$. 
If we have
\begin{equation}
    \norm{H(\eta)}\leq \frac{c_{\rm min}}{\sqrt{K}} g_1(S,\Delta),
    \label{eq:suff-cond-main}
\end{equation}
with 
\begin{align}
    g_1(S,\Delta)=\frac{1}{32 S^2}\, (e^{-2\pi}\pi \Delta)^{3(S-1)}, \end{align}
then 
\begin{equation}
    d(\{\tilde{E}_i\}, \{E_j\}) \leq  \norm{H(\eta)}\,c_{\min}^{-1} K\sqrt{K} g_2(S,\Delta),
    \label{eq:upper-d-main}
\end{equation}
with
\begin{align}
g_2(S, \Delta)= e^{2\pi} 640 \sqrt{2}  \,S^{5.5} \,  (e^{-2\pi}\pi \Delta)^{-5(S-1)}.
\end{align}
\end{theorem}

\bigskip
Since the dependence on $S$ is exponential in Eq.~\eqref{eq:upper-d-main}, one cannot make the distance $d(\{\tilde{E}_i\}, \{E_j\})$ small when the number of eigenvalues $S={\rm poly}(n)$, no matter what the gap. This is a crucial difference with the oscillatory real-time case. 
However, for $S=O(1)$, with sufficient, ${\rm poly}(n)$, effort one can make $\norm{H(\eta)}$ sufficiently small to obey Eq.~\eqref{eq:suff-cond-main} and then reduce the error on the found eigenvalues to $1/{\rm poly}(n)$. This assumes that the gap between the $O(1)$ rescaled eigenvalues present in the initial state is at least $1/{\rm poly}(n)$ (and not exponentially small in $n$).

Furthermore, given that $\norm{H(\eta)}$ should decrease at least as $\sim 1/\sqrt{K}$ through Eq.~\eqref{eq:suff-cond-main} but the upper bound in Eq.~\eqref{eq:upper-d-main} scales as $\norm{H(\eta)} K^{3/2}$, one obtains the optimal bound by choosing the {\em minimal} $K$, namely $K=2S$, so that $L=K/2=S$. In this case the Vandermonde matrix $V_{L-1}=V_{S-1}$ is square \footnote{Hence, 
strictly speaking Lemma \ref{lem:non-sq} is not much of a help.}. This expresses the intuitive fact that increasing $K$ will not help beyond a point, as for larger $K$ the signal simply dies out. This is unlike the oscillatory case of Theorem \ref{thm:osc-esprit-gap} in which the optimal $K$ is required to grow with $1/\Delta$. Here the bound does not require that $K$ grows with $1/\Delta$, so there is no `super-resolution'. We note that the upper bounds may have a sub-optimal dependence on $K$ and $S$, which is due to the proof techniques. Practically (roughly) speaking, whenever the condition number of the Vandermonde matrix $V_{L=K/2}$ grows by choosing a larger $K$, choosing that larger $K$ can be beneficial.\\


For the other decaying signal $(g_D(k))$, a rather small change from $z_i=\exp(-E_i)$ to $z_i=1-E_i/2\pi$ gives:
\bigskip
\begin{theorem}\label{thm:final-D}
Let $(g+\eta)(k)$ be a decaying signal with $k=0,\ldots,K$, and with $g(k) = \sum_{i=1}^{S}c_iz_i^k$, $c_i > 0, \:\forall i$, $c_{\min}=\min_i c_i$, and $\eta(k)$ a small noise vector. Let $z_i = 1-E_i/2\pi$ with $E_{i} \in [0, \pi]$ and given eigenvalue gap $\Delta < 1$ in Eq.~\eqref{eq:gap}, and $\{\tilde{E}_i\}$ the energy estimates of ESPRIT with $L=K/2$. Let $K+1 \geq 2S$, $K$ even and $K=TS$ for some positive integer $T$. 
If we have
\begin{equation}
    \norm{H(\eta)}\leq \frac{c_{\rm min}}{\sqrt{K}} \tilde{g}_1(S,\Delta),
    \label{eq:suff-cond}
\end{equation}
with 
\begin{align}
    \tilde{g}_1(S,\Delta)=\frac{1}{32 S^2}\,  \Delta^{3(S-1)}, \end{align}
then 
\begin{equation}
    d(\{\tilde{E}_i\}, \{E_j\}) \leq  \norm{H(\eta)}\,c_{\min}^{-1} K\sqrt{K} \tilde{g}_2(S,\Delta),
    \label{eq:upper-d}
\end{equation}
with
\begin{align}
\tilde{g}_2(S, \Delta)= 640 \sqrt{2}  \,S^{5.5} \,  \Delta^{-5(S-1)}.
\end{align}
\end{theorem}

Now to argue Theorem \ref{thm:MCD-total} from Theorem \ref{thm:final-D}, we simply choose the minimal $K=2S$, and since $S=O(1)$, it implies that the classical algorithm which estimates $g_D(k)$ for $k=0,\ldots, K(=O(1))$ within error $\epsilon$ using Lemma \ref{lem:SVT} requires ${\rm poly}(n)$ effort.

\section{Spectral estimation for a transverse-field Ising chain}
\label{simsection}
In this section, we numerically investigate the methods described thus far by applying them to an archetypal stoquastic Hamiltonian: The transverse field Ising chain. This system has been extensively studied \cite{Sachdev} and will serve as a proof-of-principle test. The system consists of qubits on a one-dimensional lattice, which interact via an Ising interaction and are exposed to an external magnetic field in the transverse direction. 
The Hamiltonian associated with this system is:
\begin{equation}
    H = -J\Big( \sum_{i}Z_{i}Z_{i+1} + g\sum_{i}X_{i} \Big),
    \label{eq:transverseisingHpaper}
\end{equation}
where $X,\:Y,\:Z$ denote the Pauli matrices, $J>0$ (for a ferromagnetic interaction) and $g\geq 0$, so that $H$ is term-wise stoquastic in the standard basis. We take the field to be pointing in the $x$-direction without loss of generality \footnote{The Hamiltonian can be transformed to $\tilde{H} = UHU^{\dagger}$ by the unitary transformation $U = \bigotimes_{i}\text{exp}\big( \frac{i\theta Z_{i}}{2} \big)$, which alters the direction of the field in the transverse plane while preserving the spectrum.}.

The system exhibits an abrupt change in the ground state of the system as a function of $g$ at $g=1$ (for $n\to \infty$). On either side of the phase transition, one has:
\begin{itemize}
\item \textbf{Strong-coupling limit} ($g\gg 1$): In this limit, the Hamiltonian is dominated by the magnetic field terms and the ground state is given by $\ket{\psi_0} \approx \ket{+}^{\otimes n}$. The $p$-particle excitations correspond to states $\ket{-}_{q_{1}}\ket{-}_{q_{2}}...\ket{-}_{q_{p}}\prod_{i\neq q_{1},q_{2},...,q_{p}}\ket{+}_{i}$, i.e., the ground state with spin flips at $p$ sites $q_{1},...,q_{p}$ along the chain. These $p$-particle excited states are $\binom{n}{p}$-fold degenerate.
\item \textbf{Weak-coupling limit} ($g\ll 1$): In this limit, the Hamiltonian is dominated by the Ising interaction terms and the (degenerate) ground state is given by either $\ket{\psi_0} \approx \ket{0}^{\otimes n}$ or $\ket{\psi_0} \approx \ket{1}^{\otimes n}$ (ferromagnetic phase). The excitations w.r.t. the ground state correspond to domain walls separating ferromagnetic regions of opposite spin.
\end{itemize}

To run the Monte Carlo scheme described in Lemma~\ref{prop:imtimesignalest}, the imaginary-time propagation operator $e^{-k H}$ must be decomposed (by means of Trotterization) in terms of the local propagation operators $e^{-a_{l}k/M\: H_{i}}$ (where $a_{l}$ and $M$ are set by the Trotterization scheme) \footnote{We note that the numerical results presented in this section are obtained using a first-order Trotter decomposition.}. 
The local propagation operators acting on a subset of two qubits on the chain are given by:
\begin{equation}
    e^{-\tilde{k} H_{i}} = 
    \begin{pmatrix}
    \frac{\sinh(\lambda \tilde{k})}{\sqrt{1+g^{2}}} + \cosh(\lambda \tilde{k}) & 0 & \frac{g\: \sinh(\lambda \tilde{k})}{\sqrt{1+g^{2}}} & 0 \\
    0 & \frac{-\sinh(\lambda \tilde{k})}{\sqrt{1+g^{2}}} + \cosh(\lambda \tilde{k}) & 0 & \frac{g\: \sinh(\lambda \tilde{k})}{\sqrt{1+g^{2}}} \\
    \frac{g\: \sinh(\lambda \tilde{k})}{\sqrt{1+g^{2}}} & 0 & \frac{-\sinh(\lambda \tilde{k})}{\sqrt{1+g^{2}}} + \cosh(\lambda \tilde{k}) & 0 \\
    0 & \frac{g\: \sinh(\lambda \tilde{k})}{\sqrt{1+g^{2}}} & 0 & \frac{\sinh(\lambda \tilde{k})}{\sqrt{1+g^{2}}} + \cosh(\lambda \tilde{k})
    \end{pmatrix},
\end{equation}
where $\lambda = J\sqrt{1+g^{2}}$ and $\tilde{k} = a_{l}k/M$. This operator is element-wise non-negative and can be efficiently brought to bock-diagonal form (with each block being irreducible). 

Since the choice of $\ket{\Phi}$ directly governs which eigenvalues can be obtained from the real-time and imaginary-time evolution signals, it is a point of particular importance. In addition, the ability of ESPRIT to extract eigenvalues from the imaginary-time and real-time signals depends very strongly on the spectral gap between the eigenvalues in the signal. We consider a state $\ket{\Phi}$ which has considerable overlap with the ground state and the ($n$-fold degenerate) first excited state in the ($g>1$)-regime. Since the gap between their associated eigenvalues increases monotonically as a function of $g$ in this regime, this allows us to present the aforementioned gap dependence numerically. We shall call the state $\ket{\Phi_{\text{optimal}}}$ since in the ($g\gg 1$)-regime it optimally overlaps with the eigenstates of interest, i.e. $|\langle +^{\otimes n}|{\psi_{p=0}}\rangle|^{2} = \sum_{q=1}^{n}|\langle +^{\otimes n}|{\psi_{p=1,q}}\rangle|^{2} = \frac{1}{2}$. This state is given by:
\begin{align}
\begin{split}
    \ket{\Phi_{\text{optimal}}} =&\: \frac{1}{\sqrt{2}}\Big( \underbrace{\prod_{i=1}^{n}\ket{+}_{i}}_{\ket{\psi_{p=0}}} + \sum_{q=1}^{n} \frac{1}{\sqrt{n}} \underbrace{\ket{-}_{q}\prod_{i\neq q}\ket{+}_{i}}_{\ket{\psi_{p=1,q}}} \Big) \\
    =&\: \frac{1}{2^{(n+1)/2}} \sum_{q=1}^{n}\Bigg(\bigg( \Big(\frac{1}{n}+\frac{1}{\sqrt{n}}\Big)\ket{0}_{q} + \Big(\frac{1}{n}-\frac{1}{\sqrt{n}}\Big)\ket{1}_{q} \bigg)\:{\sum}_{x\in \{0,1\}^{n-1}}\ket{x}\Bigg),
    \label{phioptpaper}
\end{split}
\end{align}
where $\sum_{x\in \{0,1\}^{n-1}} \ket{x}$ denotes an equal superposition of $(n-1)$-bit strings that exclude the bit in register $q$. 

We note that for $\ket{\Phi_{\text{optimal}}}$, one can efficiently obtain $\frac{\Phi(y)}{\Phi(x)}$ for a given $x,y\in\{0,1\}^{n}$ and one can efficiently sample from $\lvert \Phi(x) \rvert^{2}$: From Eq.~\eqref{phioptpaper}, one can infer a function $\Phi(x)$ ($\{0,1\}^{n} \to \mathbb{R}$) that (efficiently) gives the coefficient of the state $\ket{\Phi_{\text{optimal}}}$ associated with an $n$-bit string $x$:
    $\Phi(x) = 1/2^{(n+1)/2} \Big( \Big(\frac{1}{n}+\frac{1}{\sqrt{n}}\Big)\big(n-|x|\big) + \Big(\frac{1}{n}-\frac{1}{\sqrt{n}}\Big)|x| \Big)$,
so $\Phi(x)$ only depends on the Hamming weight $|x|$ of bit string $x$, i.e. the quantity $\frac{\Phi(y)}{\Phi(x)}$ can be efficiently determined. Furthermore, since $\Phi(x)$ only depends on $n$ and $|x|$, the distribution $\lvert\Phi(x)\rvert^{2}$ also depends solely on these quantities. This implies that one can indeed efficiently sample from this distribution: First, one draws a Hamming weight $|x|$ from the distribution $\lvert\Phi(x)\rvert^{2} = \lvert\Phi(|x|)\rvert^{2}$. Then, given $|x|$, one constructs at random an $n$-bit string with this Hamming weight. This latter step can be efficiently implemented by starting from some $n$-bit string with Hamming weight $|x|$ (such as $\{1\}^{|x|}\{0\}^{n-|x|}$) and then applying a random permutation.

\subsection{Numerical method and results}
We briefly discuss the details of the numerical analysis that is used to obtain the results presented in this section. We use the Monte Carlo and quantum algorithms (where the latter is inefficiently implemented on a classical computer), which are presented in Section \ref{sec:QPE} and summarized in Theorems \ref{thm:im} and \ref{thm:re}, to obtain resp. the imaginary-time and real-time evolution signals for the transverse-field Ising chain. We note that here we estimate the imaginary-time evolution signal using the empirical mean estimator, instead of the (asymptotically superior) median-of-means estimator. Having obtained these signals, we obtain estimates of the eigenvalues using the filtered ESPRIT method: This method corresponds to Algorithm \ref{alg:two} in combination with an additional filtering step. This additional step is required since in principle the number of components in the signal $S$ is not known a priori in the current setting. Therefore, we construct the matrix $\tilde{U}_{S}$ (in Algorithm \ref{alg:two}) by taking the first $S$ columns of $\tilde{U}$, where $S$ is now the number of singular values in the SVD of the Hankel matrix $H(y)$ that exceed ${\sf TF}\:\sigma_{\text{max}}$. ${\sf TF}$ denotes what we call a truncation factor, and $\sigma_{\text{max}}$ denotes the largest singular value of $H(y)$. In this way, the number of components in the signal emerges from the analysis of its Hankel matrix, rather than being a quantity that is known beforehand. By implementing the remainder of Algorithm \ref{alg:two} as usual, we obtain estimates of the $z_{j}$'s. From these estimates of the $z_{j}$'s, we obtain the \textit{spectral estimates} $\tilde{E}_{j}$ for the quantum algorithm and for the Monte Carlo algorithm.

Note that this approach of including a filtering step -- which often resembles more closely the practically encountered scenario when running the algorithms from Lemmas \ref{prop:retimesignalest} and \ref{prop:imtimesignalest} -- differs from that considered in Theorems \ref{thm:osc-esprit-gap} and \ref{thm:final}, where the number of components $S$ in signals is known beforehand. Here, $S$ is a quantity emerging in the analysis and it can even generally occur that components of the signal with very small coefficients -- corresponding to eigenstates with very small overlap with $\ket{\Phi}$ -- are filtered out.

In the results presented in this section, note that the real-time and imaginary-time increments have been chosen such that all $E_{j}$ that are present in the signals lie in $[0,2\pi)$. This does \textit{not} mean that the whole spectrum of the Hamiltonian lies in $[0,2\pi)$, as the majority of its eigenvalues will not be present in the signals.

We note that for the quantum algorithm, the parameters $\{z_{j}\}$ have unit norm. However, due to finite sampling, one determines a noisy version of the signal $g_R(k)$, resulting in estimated eigenvalues of the Trotterized unitary having norms that slightly deviate from unity. To ensure that the estimates $\tilde{E}_{j}$ are real-valued, we take them to be the real parts of $i \log(\tilde{z}_{j})$.

The code that is used to obtain the numerical results presented in this work can be found at \cite{GithubMaarten}.

\bigskip
In Figure \ref{excitedstateMCvQPEpaper}, the Monte Carlo signals $\bra{\Phi} e^{-k H}\ket{\Phi}$ and the real and imaginary parts of the quantum algorithm signals $\bra{\Phi}e^{-i k H}\ket{\Phi}$ for $\ket{\Phi} = \ket{+}^{\otimes n}$ and $\ket{\Phi_{\text{optimal}}}$ are depicted. 

The upper three figures correspond to $\ket{\Phi} = \ket{+}^{\otimes n}$. For this choice of $\ket{\Phi}$, the signals are clearly dominated by a single eigenvalue (the ground state eigenvalue): The Monte Carlo signal decays with a single decay rate and the quantum algorithm signals oscillate with a single frequency. For the quantum algorithm signals, there are also higher-frequency components visible (due to $\ket{+}^{\otimes n}$ not having overlap with \textit{only} the ground state).

For the lower three figures, we take $\ket{\Phi} = \ket{\Phi_{\text{optimal}}}$. For this choice of $\ket{\Phi}$, there are two eigenvalues present in the signals (the ground state and first excited state eigenvalues). For the Monte Carlo signal, the excited state eigenvalue can be seen to die out within a few units of time, after which only the ground state component is left. The quantum algorithm signals can be seen to be composed of a high-frequency (excited-state) component superposed on the ground-state component, where the excited-state component now obviously does \textit{not} die out.

\begin{figure}[t]
    \centering
    \includegraphics[width=0.9\linewidth]{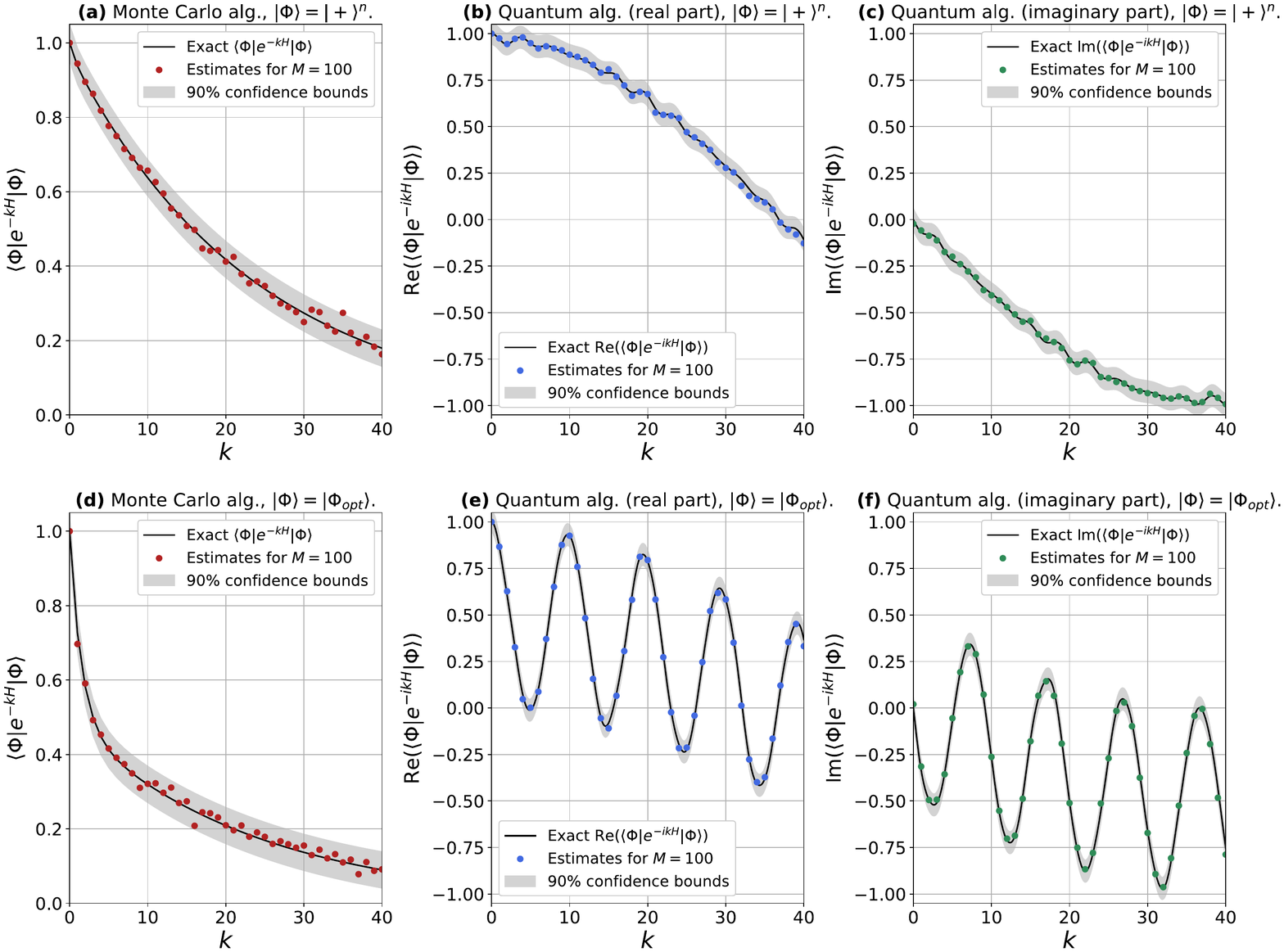}
    \caption{The evolution of the states $\ket{+}^{\otimes n}$ (in \textbf{(a)}, \textbf{(b)} and \textbf{(c)}, for which the ground state is the dominant component in the signal) and $\ket{\Phi_{\text{optimal}}}$ (in \textbf{(d)}, \textbf{(e)} and \textbf{(f)}, for which the ground state and first excited state are the dominant components in the signal) for $n=7$ and $g=4$ in imaginary time (in \textbf{(a)} and \textbf{(d)}) and in real time (in \textbf{(b)}, \textbf{(c)}, \textbf{(e)} and \textbf{(f)}). The signals in \textbf{(a)} and \textbf{(d)} are obtained through the Monte Carlo scheme of Theorem \ref{thm:im}. The signals in \textbf{(b)}, \textbf{(c)}, \textbf{(e)} and \textbf{(f)} are obtained through the quantum algorithm of Theorem \ref{thm:re} (which is inefficiently implemented on a classical computer). The Trotter variable is taken to be $M = 100$ and $|\Sigma|$ is set to be $4200$. }
    \label{excitedstateMCvQPEpaper}
\end{figure}

We now consider the spectral estimates that are obtained by applying ESPRIT to the evolution signals that are produced by the quantum algorithm (from Theorem \ref{thm:re}) and Monte Carlo algorithm (from Theorem \ref{thm:im}). In particular, we determine both time evolution signals at a given total number of measurement points in real/imaginary time. We then determine the spectral estimates from both signals for increasing $K$, by including step-by-step more of the total number of measurement points in the analysis \footnote{For $K=2$; $k = 0,1,2$. For $K=4$; $k = 0,1,2,3,4$. Etc.}. The truncation factor ${\sf TF}$ is taken to be equal to $0.02$ throughout.

The top two plots in Figure \ref{fig:spectralest} depict, for a given $|\Sigma|$, the eigenvalue estimates as a function of $g$ and for several values of $K$. For both the quantum algorithm and Monte Carlo algorithm estimates, it is clear that a smaller spectral gap indeed requires a larger $K$ for the eigenvalues to be obtained accurately. 
Furthermore, for a given $|\Sigma|$ and $K$, it is clear that the error of the estimate for the excited-state eigenvalue obtained from the imaginary-time signal is larger than that obtained from the real-time signal. We conclude furthermore that, in line with Theorems \ref{thm:osc-esprit-gap} and \ref{thm:final}, increasing $K$ beyond a certain threshold does not necessarily reduce the error of the eigenvalue estimates.

\begin{figure}[t]
    \centering
    \includegraphics[width=0.83\linewidth]{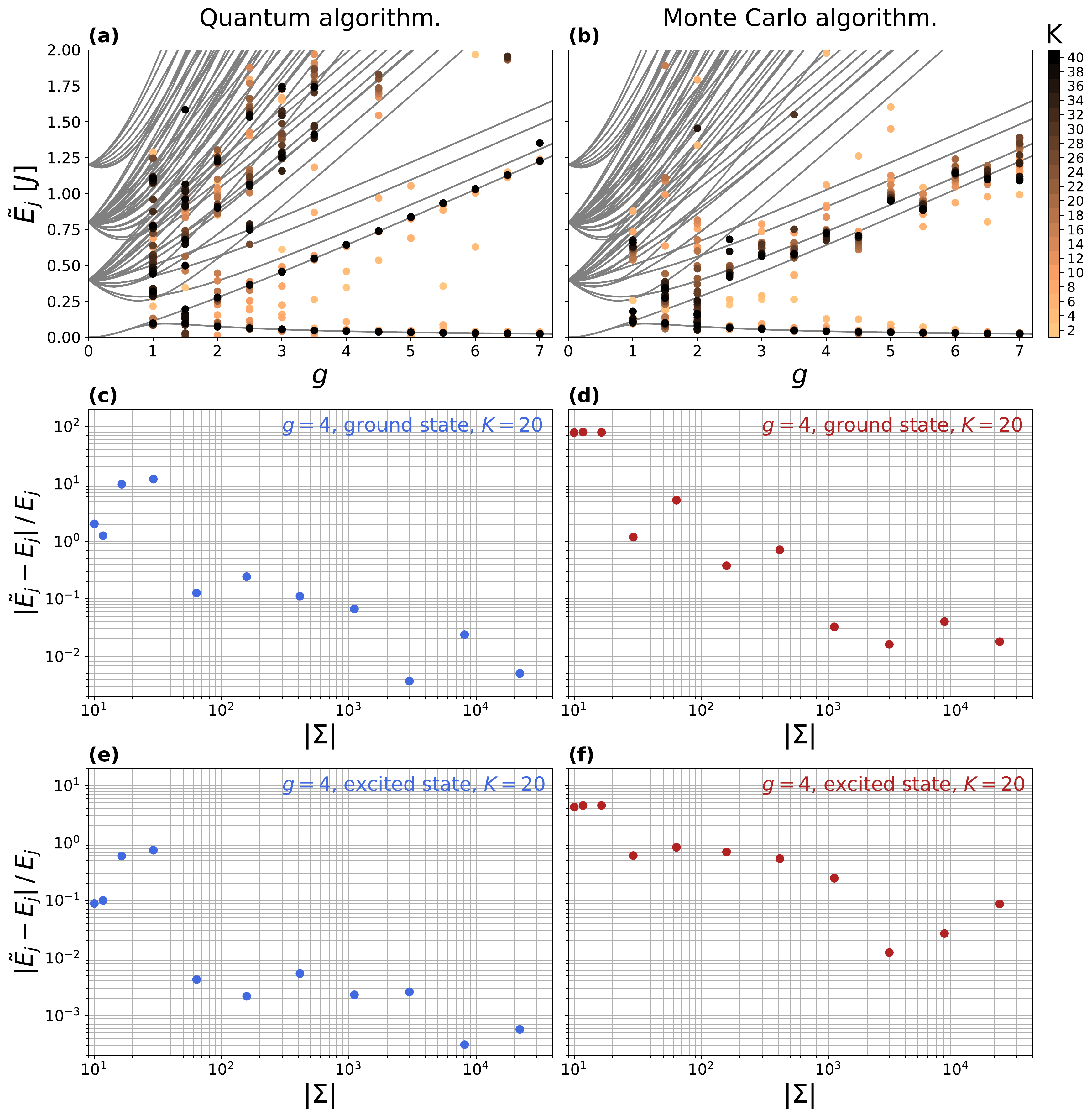}
    \caption{Spectral estimates of the ferromagnetic Ising chain in a transverse field (for $n=7$) obtained through analysis of the evolution of $\ket{\Phi_{\text{optimal}}}$. Plots \textbf{(a)} and \textbf{(b)} depict the spectral estimates (together with the true spectrum) obtained through the quantum algorithm and the Monte Carlo algorithm for $|\Sigma| = 4200$ and $M=100$ for several values of $K$. Plots \textbf{(c)},\textbf{(d)} and plots \textbf{(e)},\textbf{(f)} depict the relative error of the spectral estimates -- i.e. $|\tilde{E}_{j} - E_{j}|/E_{j}$ -- for the resp. ground state and excited state eigenvalues at $g=4$, for $M=400$ and as a function of $|\Sigma|$. The truncation factor is taken to be ${\sf TF} = 0.02$ throughout. The scaling of the error of the ground-state eigenvalue estimates is similar for both methods, while the error for excited-state eigenvalue is larger for the MC algorithm than for the quantum algorithm. The excited-state eigenvalue estimates also converge more quickly as a function of $K$ for the quantum algorithm. }
    \label{fig:spectralest}
\end{figure}

It is apparent that as one approaches the $g=1$ point, more higher-lying eigenvalues emerge from the ESPRIT analysis. This is especially true for the quantum algorithm (note that for the Monte Carlo signal, the larger the eigenvalues are, the quicker the associated components in the signal die out). The appearance of these higher-lying eigenvalues can be attributed to the fact that (for finite $n$) the state $\ket{\Phi_{\text{optimal}}}$ starts to have significant overlap with states other than the two lowest-energy eigenstates in this regime. 

The middle two and bottom two plots in Figure \ref{fig:spectralest} depict the relative error of the spectral estimates -- i.e. $|\tilde{E}_{j} - E_{j}|/E_{j}$ -- for resp. the ground-state eigenvalue and excited-state eigenvalue (at fixed $g=4$). We consider a range of values for $|\Sigma|$. For the ground-state eigenvalue, the scaling of the relative errors as a function of $|\Sigma|$ is similar for the quantum algorithm and the Monte Carlo algorithm. Clearly, the relative errors of the excited-state eigenvalue estimates for the quantum algorithm are smaller than those for the Monte Carlo algorithm.

We have also implemented the matrix pencil method in \cite{Sarkar3, francesco} to estimate the eigenvalues from the real-time and imaginary-time signals. The only significant difference that was found between the estimates obtained through the ESPRIT method and through this matrix pencil method is that -- in the ($K<2S$)-regime -- the matrix pencil method outputs estimates which resemble an average of the eigenvalues in the signal (as can be seen in Figure \ref{MPMtest} in a noiseless setting), while this is not the case generally for the ESPRIT method.


\section{Discussion}
\label{sec:con}

We have considered the problem of obtaining (some) eigenvalues of local stoquastic -- i.e. sign-problem-free -- Hamiltonians and general local Hamiltonians $H$ by means of tracking the evolution of the system state, differentiating between the evolution of the system state in real time and imaginary time. In both cases, we examine the use of the matrix pencil ESPRIT method in extracting eigenvalues of $H$ from the state evolution signal. The real-time (oscillating) evolution signal is obtained through running quantum circuits, while the imaginary-time (decaying) signal for local stoquastic Hamiltonians is obtained through a Monte Carlo scheme (developed in this work) that is implemented in a computationally tractable manner classically. Another type of decaying evolution signal -- from which the ESPRIT method can extract eigenvalues of $H$ -- is obtained through a classical method for general local Hamiltonians that is similar in spirit to `dequantization'.

We have invoked some known performance bounds of the ESPRIT method for the real-time signal and applied and extended bounds for the imaginary-time signal. Our bounds suggest that the ESPRIT method (or matrix pencil methods more generally) performs -- not surprisingly -- worse in extracting (multiple) eigenvalues from an imaginary-time decaying (MC algorithm) signal than from a real-time oscillating  (quantum algorithm) signal in the presence of noise. However, we show that if the input state contains $S=O(1)$ eigenstates and the spectral gap is at least $1/{\rm poly}(n)$, and the right access to the input state is available, the associated eigenvalues can be resolved efficiently (with ${\rm poly}(n)$ classical effort) for local stoquastic as well as for general local Hamiltonians. Even though for $S=O(1)$, the classical effort for stoquastic as well as general Hamiltonians is ${\rm poly}(n)$, the `brute-force' algorithm for general Hamiltonians (in Lemma \ref{lem:SVT}) incurs an exponential cost in $k$ in estimating the signal $g_D(k)$, while for stoquastic Hamiltonians the cost is polynomial in $k$. Despite this difference in cost, the error bounds for the eigenvalue estimates obtained here through analysis of the ESPRIT method applied to a decaying signal ($g_D(k)$ or $g_I(k)$) suggests that letting $k$ grow as some function of $n$ will generally not help.

Even though our results show that for these Hamiltonians, for an input state supported on $S=O(1)$ eigenvalues (separated by an at least $1/{\rm poly}(n)$ gap), these eigenvalues can be estimated with ${\rm poly}(n)$ classical effort, it remains to be better understood how {\em practical} this MC method for stoquastic Hamiltonians or the `dequantization' method in Lemma \ref{lem:SVT} are. The upper bounds for the errors on the eigenvalue estimates in Theorem \ref{thm:final} grow rather fast with $S$ (and the computational effort grows fast with $k$ in Lemma \ref{lem:SVT} for general local Hamiltonians), and it is not clear how much one can improve, say, the ESPRIT bounds.

Indeed, it would be interesting to show that the current bounds of ESPRIT for the imaginary-time decaying signal cannot be improved upon. There are definitely known negative results on the condition number of Vandermonde matrices \cite{pan:bad}, but there might be signal extraction algorithms that have better practical performance on decaying signals, or have looser requirements (such as the requirement that all data is evenly spaced). However, we suspect that the difficulty gap we observe between real-time and imaginary-time signal is universal. One possible way to argue this is through the Cramer-Rao bound (which has been analysed for real-time signals~\cite{stoica1989music} but not for imaginary-time signals), which is a question we leave for further research.

In terms of numerical results, we find that: For a given spectral gap and sample size, the ability to distinguish between two eigenvalues indeed depends on the number of measurement points $K$ at which the real-time and imaginary-time evolution signals are evaluated. The MC algorithm for stoquastic Hamiltonians and the quantum algorithm (in combination with the ESPRIT method) lead to a similar scaling of the relative error of the ground-state eigenvalue as a function of the sample size. However, for an excited-state eigenvalue, the quantum algorithm leads to significantly smaller relative errors than the MC algorithm. More extensive numerical studies, also of models other than the transverse-field Ising chain, may shed further light on whether the Monte Carlo + ESPRIT method is useful in practice. For frustrated stoquastic Hamiltonians, even the smallest eigenvalue may lead to a fast decaying signal, requiring small sampling error and Trotter error in practice.

As for other directions of further research, one can ask whether a hybrid approach in which imaginary-time data from an error-free Monte Carlo algorithm can strengthen the use of real-time data from a quantum algorithm obtained from a noisy quantum circuit. This approach requires combining the data where the poles/nodes $z_j=e^{-i E_j}$ on the unit circle each have a partner pole $z_j'=e^{- E_j}$ (or $z_j'=I-E_j/2\pi$) on the real axis. If the effect of noise can be modeled $z_j=e^{-i E_j} \rightarrow e^{i E_j-\gamma}$ \cite{TomBarbara}, then the imaginary-time data may help in extracting the values for $E_j$. It may also be of interest to consider the case of sampling $k$ for both the quantum circuit and Monte Carlo method at random (instead of picking $k=0,1,\ldots, K$). Another direction of further research is the following. Suppose the input state has overlap with S (here not necessarily $O(1)$) eigenstates of the Hamiltonian, one could asses how well the ESPRIT methods succeeds in extracting e.g. the ground-state eigenvalue by filtering out all other components in the real-time or imaginary-time evolution signals.

\section*{Acknowledgements}
This work is supported by QuTech NWO funding 2020-2024 – Part I “Fundamental Research”, project number 601.QT.001-1, financed by the Dutch Research Council (NWO). JH is supported by the Quantum Software Consortium (NWO Gravitation Grant, project number 024.003.037).
MS and BMT developed the MC method based on unpublished results of Sergey Bravyi, MS implemented the numerics on the transverse field Ising model, JH performed the analysis of the ESPRIT algorithm for decaying signals, BMT supervised the whole project and all authors contributed to the writing. We thank Ingo Roth for pointing out the use of median-of-means estimators for observables whose higher-order moments cannot be upper bounded by a constant. We thank Sergey Bravyi for pointing out \cite{dequant:GG}.

\appendix

\renewcommand{\thetheorem}{\Alph{section}.\arabic{theorem}}

\section{Trotterization}
\label{sec:trotter}

Suppose $H = \sum_{i=1}^{N} H_{i}$ (where $N = \mathcal{O}(\text{poly}(n))$) represents a $k$-local Hamiltonian of a quantum system. $\{H_{i}\}_{i=1}^{N}$ is generally a set of non-commuting terms but can be divided into subsets, such that within each subset all terms commute. For a given set $\{H_{i}\}_{i=1}^{N}$, we denote the minimum possible number of these subsets by $\Gamma$. This number of subsets is at most $N$ and equals $1$ in the trivial case where all $H_{i}$'s commute with each other. The Hamiltonian $H$ can thus be decomposed as $H = \sum_{\gamma = 1}^{\Gamma} H_{\gamma}$, where all $H_{\gamma}$ do not commute with each other, but the terms of which each individual $H_{\gamma}$ is composed do commute. Choosing a decomposition into the minimum number of subsets brings about an additional advantage of parallelizability when implementing the evolution of the systems in imaginary or real time.

The following Lemma (adaptation from \cite{childs+:trotter}) upper bounds the errors of implementing imaginary-time and real-time state evolution through a first-order Trotter decomposition.

\begin{lemma}
\textbf{First-Order Trotter Decomposition.} 
Given a $k$-local Hamiltonian $H = \sum_{i}^{N}H_{i}$. Furthermore, suppose the set $\{H_{i}\}_{i=1}^{N}$ can be divided into a minimum of $\Gamma$ subsets $\{H_{\gamma}\}_{\gamma = 1}^{\Gamma}$, such that within each individual subset all $H_{i}$'s commute. Then the quantities $\bigl\lvert \bra{\Phi}e^{-itH}\ket{\Phi} - \bra{\Phi}\big(\prod_{\gamma}e^{-it H_{\gamma}/M}\big)^{M} \ket{\Phi} \bigr\rvert$ and $\bigl\lvert \bra{\Phi}e^{-\tau H}\ket{\Phi} - \bra{\Phi}\big(\prod_{\gamma}e^{-\tau H_{\gamma}/M}\big)^{M} \ket{\Phi} \bigr\rvert$ (where $\ket{\Phi}$ is a normalized state and $t,\tau \in \mathbb{R}_{+}$) are bounded as follows:
\begin{subequations}
\begin{equation}
    \Bigl\lvert \bra{\Phi}e^{-itH}\ket{\Phi} - \bra{\Phi}\big(\prod_{\gamma}e^{-it H_{\gamma}/M}\big)^{M} \ket{\Phi} \Bigr\rvert \leq \sum_{\gamma'=1}^{\Gamma - 1}\sum_{\gamma > \gamma'}\norm{\:[H_{\gamma'},H_{\gamma}]\:}\:\frac{t^{2}}{2M},
\end{equation}
\begin{equation}
    \Bigl\lvert \bra{\Phi}e^{-\tau H}\ket{\Phi} - \bra{\Phi}\big(\prod_{\gamma}e^{-\tau H_{\gamma}/M}\big)^{M} \ket{\Phi} \Bigr\rvert \leq 3e^{2}\sum_{\gamma'=1}^{\Gamma - 1}\sum_{\gamma > \gamma'}\norm{\:[H_{\gamma'},H_{\gamma}]\:}\:\frac{\tau^{2}}{2M},
\end{equation}
\end{subequations}
where the second inequality holds provided that $\bigl\lvert\bigl\lvert e^{-\tau H/M} \bigr\rvert\bigr\rvert \leq 1$, $\bigl\lvert\bigl\lvert e^{-\tau H_{\gamma}/M} \bigr\rvert\bigr\rvert \leq 1$ ($\: \forall \gamma$) and $\frac{\tau\big( \sum_{\gamma} \norm{H_{\gamma}} \big)}{M} \leq 1$, and $M$ denotes the Trotter variable.
\label{lemma:Trotterpaper}
\end{lemma}

\begin{figure}[t]
    \centering
    \includegraphics[width=0.7\linewidth]{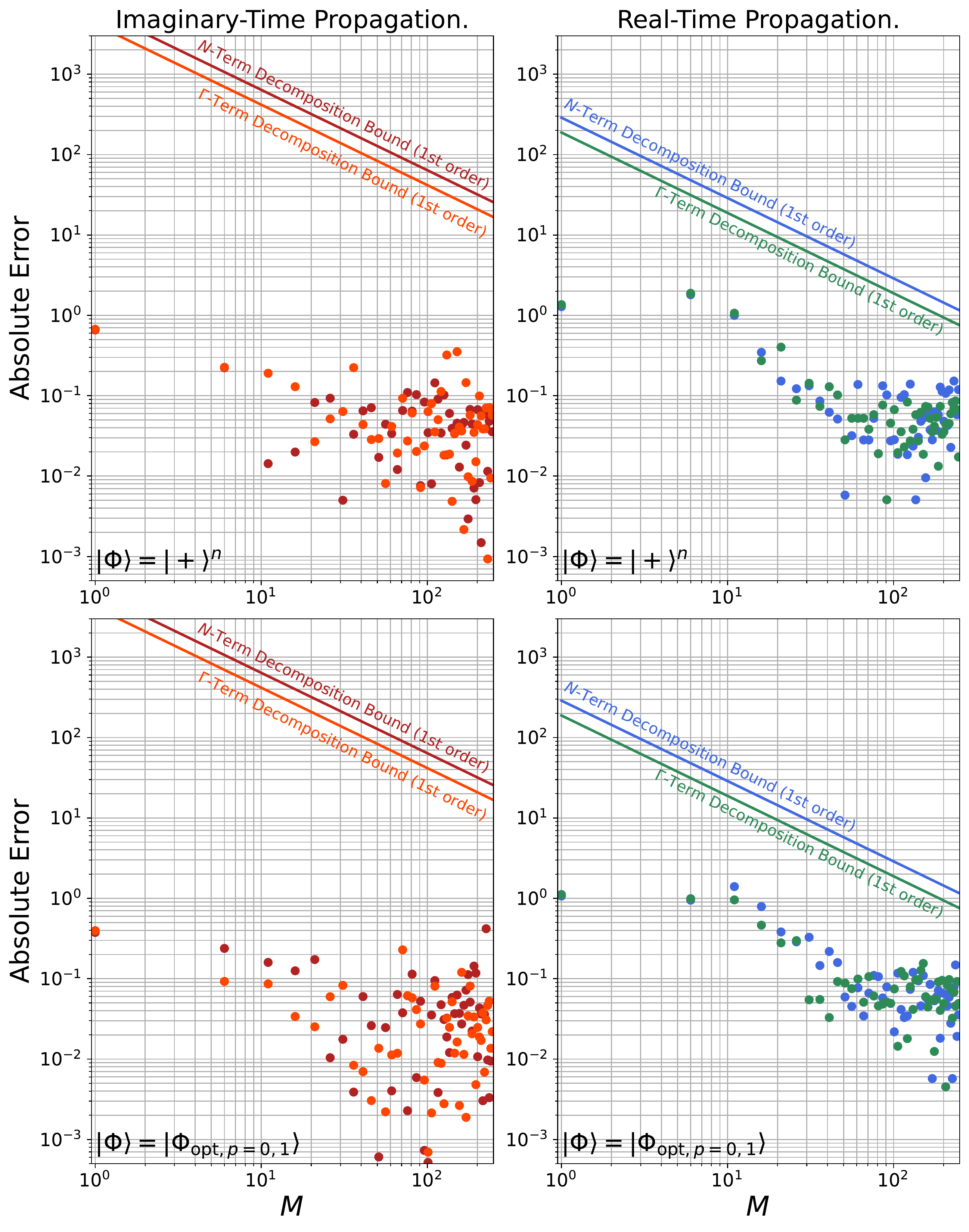}
    \caption{Absolute Trotter error (imposed on the signal estimate) as a function of the Trotter variable $M$ for the imaginary-time and real-time signals. The noisy ($|\Sigma| = 200$) and Trotterized versions of $\bra{\Phi}e^{-\tau H}\ket{\Phi}$ and $\bra{\Phi}e^{-it H}\ket{\Phi}$ for a ferromagnetic Ising chain in a transverse field (for $g=4$ and $n=8$) are evaluated at $\tau = t = 3$ and several values of $M$. The Trotterization schemes are first-order $N$-term and first-order $\Gamma$($=2$)-term schemes. The associated error bounds are included in matching colors.}
    \label{TROTTERCOMP}
\end{figure}

To obtain a better scaling of the errors as a function of the Trotter variable $M$, one can employ higher-order Trotter decompositions. We denote the $p$th-order approximants of $e^{-itH/M}$ and $e^{-\tau H/M}$ by $\mathcal{T}_{M}(p,t)$ and $\mathcal{T}_{M}(p,\tau)$, respectively. We denote $\bigl\lvert \bra{\Phi}e^{-itH}\ket{\Phi} - \bra{\Phi}\mathcal{T}_{M}(p,t)^{M}\ket{\Phi} \bigr\rvert$ and $\bigl\lvert \bra{\Phi}e^{-\tau H}\ket{\Phi} - \bra{\Phi}\mathcal{T}_{M}(p,\tau)^{M}\ket{\Phi} \bigr\rvert$ by $\epsilon_{\text{trot}}$. In \cite{childs+:trotter}, it was shown that, for general $p$, $\epsilon_{\text{trot}}$ is upper bounded as follows:
\begin{subequations}
\begin{align}
    \epsilon_{\text{trot}} \leq &\: \mathcal{O}\Big( \alpha\:t^{p+1}/M^{p} \Big),\quad \text{for real-time evolution},\\
    \epsilon_{\text{trot}} \leq &\: \mathcal{O}\Big( \alpha\:\tau^{p+1}/M^{p} \Big),\quad \text{for imaginary-time evolution},
    \label{eq:trotimag}
\end{align}
\label{highptrotbounds}
\end{subequations}
where $\alpha = \sum_{\gamma_{1},\gamma_{2},...,\gamma_{p+1}=1}^{\Gamma}\norm{\:[H_{\gamma_{p+1}},...,[H_{\gamma_{2}},H_{\gamma_{1}}]...]\:}$ ($\alpha^{1/p}$ is typically $\text{poly}(n)$) and Eq.~\eqref{eq:trotimag} holds provided that $4\tau \Upsilon \big( \sum_{\gamma} \norm{H_{\gamma}} \big)/M\leq 1$ (where $\Upsilon$ corresponds to the number of \textit{stages} of the Trotter decomposition and typically scales exponentially in $p$) \footnote{In the remainder of this discussion it is assumed that this condition is satisfied.}. In \cite{Suzuki}, a widely used scheme is discussed for constructing $p$th-order approximants.

It is important to consider the total number of $k$-local propagation operators $L$ required to simulate $e^{-itH}$ and $e^{-\tau H}$ (for a given order $p$ and Trotter variable $M$). For the scheme in \cite{Suzuki}, the number of these $k$-local propagation operators required to be implemented for the simulation of $e^{-itH}$ and $e^{-\tau H}$ for $p>1$ is $L=2MN\:5^{\frac{p}{2}-1}$ (and for $p=1$ is $MN$). If one wishes to obtain a given $\epsilon_{\text{trot}}$, the number of $k$-local propagation operators into which the evolutions are decomposed scales as $L=\text{poly}(n)\:\mathcal{O}\big(\Upsilon\:t^{1+1/p}\epsilon_{\text{trot}}^{-1/p}\big)$ (for real time) and $L=\text{poly}(n)\:\mathcal{O}\big(\Upsilon\:\tau^{1+1/p}\epsilon_{\text{trot}}^{-1/p}\big)$ (for imaginary time). We thus conclude that for large $p$ (i.e. high-order decompositions), $L$ scales approximately linearly in the evolution time of the system under consideration (for real-time and imaginary-time evolution).

In Figure \ref{TROTTERCOMP}, we have depicted the absolute error of noisy MC (imaginary-time) and QPE (real-time) signals at fixed $\tau = t$ as a function of $M$, obtained through first-order $N$-term and $\Gamma$($=2$)-term Trotterization schemes. We have included the first-order Trotter error bounds. We note that the apparent drastic increase in noise magnitude as a function of $M$ is primarily due to the fact that the absolute error decreases as a function of $M$ and is plotted on a logarithmic scale.


\section{Extension to non-Hermitian propagation operators}\label{AppC}
In this Appendix we prove the following Lemma, extending Lemma \ref{prop:imtimesignalest}:

\begin{lemma}
Let $\mathcal{F}  \equiv \bra{\Phi} G_{1}G_{2}\: ...\: G_{L} \ket{\Phi}$, 
where:
\begin{enumerate}
    \item $\ket{\Phi} = \sum_{x=1}^{2^{n}}\Phi(x)\ket{x}$ is a normalized state of $n$ qubits where $\Phi(x) \in \mathbb{C}$ ($\forall x$) and $\sum_{x}\bigl\lvert \Phi(x) \bigr\rvert^{2} = 1$. We assume that (1) $\frac{\Phi(y)}{\Phi(x)}$ can be efficiently (${\rm poly}(n)$) calculated for a {\em given} $x$ and $y$ and (2) we can efficiently draw samples from the probability distribution $P(x) = \bigl\lvert \Phi(x) \bigr\rvert^{2}$.
    \item Each $G_{l}$ is a $k$-local (possibly non-Hermitian) element-wise nonnegative matrix with singular values in $(0,1]$.
\end{enumerate}
$\mathcal{F}$ can be estimated within error $\epsilon$ with probability at least $1-\delta$ with a classical MC algorithm with runtime $\text{poly}(n)\times \Theta(\epsilon^{-2}\delta^{-1})\times \Theta(L)$. 
\label{prop:imtimesignalest-ext}
\end{lemma}

\smallbreak
\begin{proof}
In addition to the $n$-qubit register, we exploit a single ancillary qubit. The matrices $G_{l}$ are still element-wise non-negative. The state $\ket{a}$ denotes the state of the single ancillary qubit. By making use of the single ancillary qubit, the propagation operators can be symmetrized as follows: 
\begin{equation}
    F_{l} \equiv \begin{cases}
               G_{l}\otimes \ketbra{0}{1} + G_{l}^{\dagger}\otimes \ketbra{1}{0}\:,\quad \text{if $l$ is odd}\\
               G_{l}\otimes \ketbra{1}{0} + G_{l}^{\dagger}\otimes \ketbra{0}{1}\:,\quad \text{if $l$ is even.}
            \end{cases}
    \label{eq:definitionFpaper}
\end{equation}
In this form, $F_{l}$ (the `new' propagation operator) is element-wise non-negative, $k+1$-local and Hermitian and hence one can apply Lemma \ref{prop:imtimesignalest} to $\bra{\Phi} F_1 F_2 \ldots F_L \ket{\Phi}$, provided that its eigenvalues lie in $(0,1]$. The eigenvalues $\lambda$ of $F_{l}$ (for $l$ odd) can be found by solving:
\begin{equation}
    \det \begin{pmatrix}
    -\lambda \mathbb{1} & G_{l} \\
    G_{l}^{\dagger} & -\lambda \mathbb{1}
    \end{pmatrix} = \det\Big(\lambda^{2}\mathbb{1}- G_{l}G_{l}^{\dagger}\Big) = \det\Big( G_{l}G_{l}^{\dagger} - \lambda^{2}\mathbb{1}\Big) = 0,
\end{equation}
where we have used that $G_{l}^{(\dagger)}$ commutes with $\mathbb{1}$ and that $G_{l}^{(\dagger)}$ is of even dimensionality. The eigenvalues of the Hermitian and positive semi-definite matrix $G_{l}G_{l}^{\dagger}$ are thus equal to $\lambda^{2}$. Since the singular values of $G_{l}$ are equal to the square root of the eigenvalues of $G_{l}G_{l}^{\dagger}$, the eigenvalues of $F_{l}$ will lie in $(0,1]$ if the singular values of $G_{l}$ lie in $(0,1]$. This can be similarly shown for $l$ even and this statement thus holds for all $l$.

What is left to prove is that estimating the signal for the string of $F_{l}$'s is equivalent to estimating the signal for the string of $G_{l}$'s. Specifically, we want to prove the following identity: $G_{1}G_{2}...G_{L} = \bra{0}F_{1}F_{2}...F_{L} \ket{L\:\text{mod}\:2}$, for $L\in \mathbb{Z}_{+}$. This is done below by means of induction.
\begin{itemize}
    \item For $L=1$: 
    \begin{equation}
    \begin{split}
    \bra{0}F_{1} \ket{1} &=  \bra{0}\Big( G_{1}\otimes \ketbra{0}{1} + G_{1}^{\dagger}\otimes \ketbra{1}{0} \Big)\ket{1} \\
    &= G_{1} \braket{0}{0}\braket{1}{1}+G_{1}^{\dagger} \braket{0}{1}\braket{0}{1} \\
    &= G_{1},
    \end{split}
    \end{equation}
    \item Assuming  $G_{1}G_{2}...G_{L} =  \bra{0}F_{1}F_{2}...F_{L} \ket{L\:\text{mod}\:2}$ holds for $L$, it holds for $L+1$ as well:
    \noindent
    Making use of the definition in Eq.~\eqref{eq:definitionFpaper}, we write $F_{L+1}$ as follows:
    \begin{equation}
        F_{L+1} = G_{L+1}\otimes \ketbra{L\:\text{mod}\:2}{L+1\:\text{mod}\:2} + G_{L+1}^{\dagger}\otimes \ketbra{L+1\:\text{mod}\:2}{L\:\text{mod}\:2}.
    \end{equation}
    The quantity of interest -- in the case of the length of the operator string being $L+1$ -- can now be rewritten as follows: 
    \begin{equation}
            \bra{0}F_{1}F_{2}...F_{L}F_{L+1}\ket{L+1\:\text{mod}\:2} =  \bra{0}F_{1}F_{2}...F_{L} \ket{L\:\text{mod}\:2}G_{L+1} =
             G_{1}G_{2}...G_{L}G_{L+1},
    \end{equation}
    which finishes the proof.
\end{itemize}
\end{proof}

\section{Median-of-means estimator}
\label{app:MOM}
The MC scheme described in Section \ref{sec:QPE} produces a set of $|\Sigma|$ samples $\{\boldsymbol{x}\}$ which are distributed according to $\Pi(\boldsymbol{x})$. For each sample, $\text{Re}(\mathcal{R}(\boldsymbol{x}))$ can be evaluated and subsequently an estimate of $\mathcal{F}$ can be obtained. Only the first and second moments of the random variable $\text{Re}(\mathcal{R}(\boldsymbol{x}))$ can be upper bounded in general. Therefore, if one would use the empirical mean $\text{Re}(\tilde{\mathcal{F}}) = \frac{1}{|\Sigma|}\sum_{\boldsymbol{x}\in\Sigma}\text{Re}(\mathcal{R}(\boldsymbol{x}))$ as a mean estimator for $\mathcal{F}$, then the best achievable scaling of $|\Sigma|$ such that
\begin{equation}
    \text{Pr}\Big( \bigl\lvert \text{Re}(\tilde{\mathcal{F}}) - \mathcal{F} \bigr\rvert \leq \epsilon \Big) \geq 1-\delta,
\end{equation}
is $|\Sigma| = \Theta(\epsilon^{-2}\delta^{-1})$ (by means of Chebyshev's inequality).

Taking the median-of-means estimator \cite{MOM} as estimator (instead of the empirical mean), one can obtain a more convenient scaling of $|\Sigma|$ w.r.t. $\delta$ (despite the fact that only the first two moments of $\text{Re}(\mathcal{R}(\boldsymbol{x}))$ can be upper bounded). The median-of-means estimator can be constructed as follows: Partition the set of MC samples $\Sigma$ into $q$ groups $s_1,\ldots,s_q$ of size approximately $|\Sigma|/q$. One then computes the empirical mean of $\text{Re}(\mathcal{R}(\boldsymbol{x}))$ over the samples in each group separately (giving $q$ unbiased estimators of $\mathcal{F}$) and takes the median of these empirical means. We denote the empirical mean for each group by $f_{j} = \frac{1}{|s_{j}|}\sum_{\boldsymbol{x}\in s_{j}}\text{Re}(\mathcal{R}(\boldsymbol{x}))$ (for $j \in \{1,\ldots,q\}$) and denote the median of these empirical means by $\hat{\mathcal{F}} = \mathrm{M}(f_{1},\ldots,f_{q})$. The estimator $\hat{\mathcal{F}}$ is the median-of-means estimator.

We define the median of $q$ real numbers $a_{1},\ldots,a_{q}$ as $\mathrm{M}(a_{1},\ldots,a_{q}) = a_i$ with $a_i$ such that
\begin{equation}
    |\{j:a_{j}\leq a_i\}|\geq q/2 \quad \wedge \quad |\{j:a_{j}\geq a_i\}|\geq q/2,
\end{equation}
where we take the smallest $i$ if multiple $i$s obey this condition.

$\{\text{Re}(\mathcal{R}(\boldsymbol{x}))\}$ are i.i.d. random variables with mean $\mathcal{F}$ and variance $\text{Var}\big(\text{Re}(\mathcal{R}(\boldsymbol{x}))\big)\leq 1$. Let $q$ and $|\Sigma|/q$ be positive integers, then
\begin{equation}
    \text{Pr}\Big( \bigl\lvert \hat{\mathcal{F}} - \mathcal{F} \bigr\rvert \leq \sqrt{4q/|\Sigma|} \Big) \geq 1-e^{-q/8}.
    \label{eq:dev}
\end{equation}
So for $q=8\:\text{log}(\delta^{-1})$ and $|\Sigma| = 4\:q\epsilon^{-2} = 32\:\text{log}(\delta^{-1})\epsilon^{-2}$, we have:
\begin{equation}
    \text{Pr}\Big( \bigl\lvert \hat{\mathcal{F}} - \mathcal{F} \bigr\rvert \leq \epsilon \Big) \geq 1-\delta.
\label{confidencebound}
\end{equation}
Note that the estimator $\hat{\mathcal{F}} = \mathrm{M}(f_{1},\ldots,f_{q})$ depends explicitly on the confidence since $q$ scales with $\delta$. Given that indeed $q = \Theta\big(\text{log}(\delta^{-1})\big)$, the number of samples required to obtain Eq.~\eqref{confidencebound} is $|\Sigma|=\Theta(\text{log}(\delta^{-1})\epsilon^{-2})$ (which is an exponentially better scaling w.r.t. $\delta$ compared to that for the empirical mean estimator).

To see why Eq.~\eqref{eq:dev} is true
, see \cite{MOM}, note that one can apply Chebyshev's inequality to each of the empirical means $f_{j}$: with probability at least $3/4$, we have $\bigl\lvert f_{j} - \mathcal{F} \bigr\rvert \leq \sqrt{4q/|\Sigma|}$. If $\bigl\lvert \hat{\mathcal{F}}-\mathcal{F} \bigr\rvert \geq \sqrt{4q/|\Sigma|}$, then, by definition of $\hat{\mathcal{F}}$, at least $q/2$ of the empirical means $f_{j}$ satisfy $\bigl\lvert f_{j} - \mathcal{F} \bigr\rvert \geq \sqrt{4q/|\Sigma|}$. Hence the probability that $\bigl\lvert \hat{\mathcal{F}}-\mathcal{F} \bigr\rvert \geq \sqrt{4q/|\Sigma|}$ is upper bounded by the probability that a binomially distributed random variable with $q$ draws and success probability $1/4$ exceeds $q/2$:
\begin{equation}
    \text{Pr}\Big( \bigl\lvert \hat{\mathcal{F}}-\mathcal{F} \bigr\rvert \geq \sqrt{4q/|\Sigma|} \Big) \leq \text{Pr}\Big( \text{Bin}(q,1/4)\geq q/2 \Big) = \text{Pr}\Big( \text{Bin}(q,1/4) - \mathbb{E}\big( \text{Bin}(q,1/4) \big) \geq q/4 \Big) \\ \leq e^{-q/8},
\end{equation}
where we have used $\mathbb{E}\big( \text{Bin}(q,1/4) \big) = q/4$ and Hoeffding's inequality.

\section{Performance of ESPRIT on the imaginary-time (decaying) signal}
\label{app:MPM}

In this section we prove a series of Lemmas that characterize the behaviour of the ESPRIT algorithm (Algorithm \ref{alg:two}) on an imaginary-time signal obtained with finite error. They are direct generalisations of the work done in \cite{li2019superresolution}, which leads up to Theorem \ref{thm:osc-esprit-gap} for oscillatory signals, to signals composed of real exponential decays. We will see that the guarantees on the algorithm will be substantially weaker in this case. The end goal of this section is Theorem \ref{thm:final} in the main text. 

The argument decomposes roughly into two halves. In the first half we argue that the behaviour of ESPRIT is controlled by the smallest non-zero singular value of the Vandermonde matrix $V_{L}$. In the second half we argue that that this smallest nonzero singular value can be controlled in terms of a gap condition on the energy eigenvalues of the imaginary-time signal.\\

We start by proving a short result on the smallest nonzero singular values of products of matrices. 

\begin{lemma}\label{lem:small_sing}
Let the smallest \emph{nonzero} singular value of a matrix $X$ be $\sigma_{\rm min}(X)$. For any matrix, $X$ we have $\sigma_{\rm min}(X):= \norm{X^+}^{-1}$, where $X^+$ is the Moore-Penrose pseudo-inverse of $X$, i.e. through the SVD, we have $\sigma_{\rm min}^{-1}(X) = \norm{X^+}$, where $\norm{X}$ is the operator norm (the largest singular value). Let $A,B$ be (non-square) matrices such that $(AB)^+ = B^+A^+$. Then we have that
\begin{equation}
    \sigma_{\rm min}(AB) \geq \sigma_{\rm min}(A)\sigma_{\rm min}(B).
\end{equation}
\end{lemma}
\begin{proof}
By sub-multiplicativity of the operator norm, we have that
\begin{equation}
    \sigma_{\rm min}(AB) = \big(\norm{(AB)^+}\big)^{-1} = \big(\norm{B^+A^+}\big)^{-1}\geq  \big(\norm{B^+}\norm{A^+}\big)^{-1} = \sigma_{\rm min}(B)\sigma_{\rm min}(A).
\end{equation}
\end{proof}
We note that the product property on the Moore-Penrose pseudo-inverse does not hold for all matrices (unlike for the regular inverse). We will make use of the following sufficient condition:
\begin{lemma}[\cite{greville1966note}]
Let $A,B$ be matrices and let $A$ have full column rank, and $B$ have full row rank. Then we have $(AB)^+ = B^+A^+$.
\end{lemma}

Next, we argue that a small perturbation in the imaginary-time signal does not impact the space spanned by the the first $S$ left singular vectors of the Hankel matrix $H(g)$ too strongly, see the ESPRIT Algorithm \ref{alg:two}.  It is a compressed version of Lemmas $4$ and $5$ in \cite{li2019superresolution} (which are formulated for real-time signals only, but hold more generally).
To state this Lemma we need to consider a freedom of choice in $U_S$ and $\tilde{U}_S$ with $\tilde{U}_S$ as defined in the ESPRIT Algorithm \ref{alg:two} and $U_S$ its noise-free version. It is possible that $U_S$ and $\tilde{U}_S$ are far apart as operators, even if the spaces they span are close together. 

We solve this by not considering $U_S$ proper, but rather a rotated version of $U_S$. As we will see, this rotation will not impact the actual output of ESPRIT which are the eigenvalues of the signal matrix $\tilde{\Psi}$. The rotated version of $U_S$ is given through the $S\times S$ unitary operator $(O_2O_1)^\dagger$, which is defined via the singular value decomposition of $U_S^\dagger \tilde{U}_S$, i.e.  
\begin{align}
    O_1 U_S^\dagger \tilde{U}_S O_2 = D,
    \label{eq:defD}
\end{align}
with $I_S\geq D\geq 0$ and $D$ diagonal. The  diagonal elements of the matrix $D$ are cosines of the so-called \emph{canonical angles}.  We note that this internal rotation is performed implicitly in \cite{li2019superresolution}, whereas we have chosen to make it explicit at all times. 
\begin{lemma}\label{lem:U_pert}
Let $(g+\eta)(k)$ be an imaginary-time signal with $g(k) = \sum_{i=1}^{S}c_iz_i^k$ and $\eta(k)$ a small noise vector.  Consider the associated Hankel matrices $H(g)$ and $H(g+\eta)$, with singular value decompositions $H(g) = U\Sigma W$ and $H(g+\eta) = \tilde{U}\tilde{\Sigma} \tilde{W}$, and label the matrix of the first $S$ columns of $U$ (resp. $\tilde{U}$) as $U_S$ (resp. $\tilde{U}_S$). Finally, let $O_1 U_S^\dagger \tilde{U}_S O_2 = D$ with $I_S\geq D\geq 0$ be the singular value decomposition of $U_S^\dagger \tilde{U}_S$. 
If 
\begin{equation}
    \norm{H(\eta)}\leq \sigma_{\rm min}(H(g))/2,
\end{equation}
then
\begin{equation}
    \norm{U_S(O_2 O_1)^\dagger-\tilde{U}_S}\leq \frac{2\sqrt{2S}\norm{H(\eta)}}{\sigma_{\rm min}(H(g))}.
\end{equation}
\end{lemma}
\begin{proof}
First, we can observe that indeed $I_S \geq D$ as $\norm{O_1 U_S^{\dagger} \tilde{U}_S O_2} \leq \norm{O_1}\, \norm{U_S^{\dagger}}\, \norm{\tilde{U}_S}\,\norm{O_2} \leq 1$.

The proof follows from Wedin's  $\sin\Theta$ theorem for perturbations of singular subspaces as well as Weyl's perturbation theorem for singular values, see e.g. \cite{bhatia2013matrix,stewart1991perturbation}. From this latter theorem we know that $|\sigma_i(H(g+\eta))-\sigma_i(H(g))|\leq \norm{H(\eta)}\leq \sigma_{\rm min}(H(g))/2$ where $\sigma_i$ is the $i$th singular value (in order and some singular values can be zero). Let $\sigma_{\rm min}(H(g+\eta))>0$ be the $k$th singular value, and thus 
\begin{align}
    \sigma_{\rm min}(H(g+\eta))\geq  \sigma_k(H(g))-\sigma_{\rm min}(H(g))/2 \geq \sigma_{\rm min}(H(g))/2,
\end{align}
where the last inequality holds as $\sigma_k(H(g))> 0$ and hence is at least $\sigma_{\rm min}(H(g))$.
Hence we can use Wedin's theorem on singular values (Theorem 3.4 in \cite{li1998relative}, setting $\delta=\alpha = \sigma_{\rm min}(H(g))/2$) to conclude that
\begin{equation}
   \norm{(U_S^{\perp})^\dagger\tilde{U}_S} \leq \frac{2\norm{H(\eta)}}{\sigma_{\rm min}(H(g))}, 
\end{equation}
where $U_S^\perp$ is the matrix formed from the $L+1-S$ other (besides $U_S$) columns of the noiseless $U$. 
To connect this to $U_S(O_2 O_1)^\dagger-\tilde{U}_S$ we can make the following long calculation: 
\begin{align}
    \norm{U_{S}(O_2 O_1)^\dagger-\tilde{U}_{S}}&\leq \norm{U_{S}(O_2 O_1)^\dagger-\tilde{U}_{S}}_F \notag \\
    &= \left[ \tr(U_S(O_2 O_1)^\dagger(O_2 O_1)U_S^{\dagger}) + \tr(\tilde{U}_S\tilde{U}_S^\dagger) - \tr\big((U_S(O_2 O_1)^\dagger\tilde{U}_S^\dagger+ \tilde{U}_S (O_2 O_1)U_S^\dagger)\big)\right]^{1/2} \notag \\
    &=\left[ \tr(U_SU_S^{\dagger}) + \tr(\tilde{U}_S\tilde{U}_S^\dagger) - 2\tr(D)\right]^{1/2} \notag \\
    &\leq\left[ 2\tr(\tilde{U}_S \tilde{U}_S^{\dagger})  - 2\tr\big(DD^\dagger\big)\right]^{1/2} \notag \\
    &=\sqrt{2}\left[ \tr(\tilde{U}_S \tilde{U}_S^{\dagger}) -\tr\big((O_1 U_S^\dagger\tilde{U}_S O_2)(O_1 U_S^\dagger\tilde{U}_S O_2)^\dagger\big)\right]^{1/2} \notag \\
    &= \sqrt{2}\left[ \tr(\tilde{U}_S \tilde{U}_S^{\dagger}) - \tr\big(U_S^\dagger\tilde{U}_S\tilde{U}_S^\dagger U_S \big)\right]^{1/2} \notag \\
    &= \sqrt{2}\left[ \tr(\tilde{U}_S \tilde{U}_S^{\dagger})  - \tr\big(U_S U_S^\dagger\tilde{U}_S\tilde{U}_S^\dagger \big)\right]^{1/2} \notag \\
    & = \sqrt{2}\left[ \tr\big(\tilde{U}_S \tilde{U}_S^{\dagger}) - \tr(\tilde{U}_S\tilde{U}_S^\dagger)  + \tr\big(U_S^{\perp} (U_S^{\perp})^\dagger\tilde{U}_S\tilde{U}_S^\dagger \big)\right]^{1/2} \notag \\
    & = \sqrt{2}\norm{(U_S^{\perp})^\dagger\tilde{U}_S}_F
    \leq \sqrt{2S} \norm{(U_S^{\perp})^\dagger\tilde{U}_S}
\end{align}
In the second inequality we used that $\tr(U_SU_S^{\dagger})=\tr(\tilde{U}_S\tilde{U}_S^\dagger)=S$ since $U_S$ as well as $\tilde{U}_S$ consist of $S$ orthonormal columns, and $D\geq D^2 = DD^\dagger$, since $I\geq D\geq 0$. In addition, at the end we use that $U U^{\dagger}=U_S U_S^{\dagger}+U_S^{\perp}(U_S^{\perp})^{\dagger}=I$ as $U$ is unitary.
\end{proof}

The next step is to bound the deviation of the ESPRIT signal matrix $\tilde{\Psi} = \tilde{U}_0^+ \tilde{U}_1$ from the rotated version of its noiseless variant $(O_2 O_1)\Psi(O_2 O_1)^\dagger = (O_2 O_1){U}_0^+ {U}_1(O_2 O_1)^\dagger$ in terms of $\norm{U_S(O_2 O_1)^\dagger-\tilde{U}_S}$. 
Recall that $U_0$ (resp. $U_1$) are constructed by removing respectively the first or last row from the matrix $U_S$.
 Note also that only the eigenvalues of the signal matrix $\Psi$ matter in the ESPRIT Algorithm \ref{alg:two} and the additional unitary rotations $O_2 O_1$ do not alter these eigenvalues.
We first establish some intermediate result:

\begin{lemma}\label{lem:inverse_pert}
Let $A,B$ be matrices such that $\mathrm{rank}(A) = \mathrm{rank}(B)$. If $\norm{A-B} \leq \sigma_{\rm min}(A)/2$ then 
\begin{equation}
    \norm{A^+ - B^+} \leq \frac{1+ \sqrt{5}}{2} \norm{A - B}\norm{A^+}^2 = \frac{1+ \sqrt{5}}{2} \frac{\norm{A - B}}{\sigma_{\rm min}^2(A)}.
    \label{eq:sqrt5}
\end{equation}
\end{lemma}
\begin{proof}
From Theorem 4.1 in \cite{wedin1973perturbation} we get that
\begin{equation}
    \norm{A^+ - B^+} \leq \frac{1+ \sqrt{5}}{2} \norm{A - B}\norm{A^+}\norm{B^+}.
\end{equation}
Furthermore, since
$\norm{A-B} \leq \frac{\sigma_{\rm min}(A)}{2}< \frac{1}{\norm{A^+}}$, we have by Lemma 3.1 in \cite{wedin1973perturbation} that 
\begin{equation}
    \norm{B^+} \leq \frac{\norm{A^+}}{1- \norm{A^+}\norm{A-B}}\leq \norm{A^+},
\end{equation}
leading to the first inequality in Eq.~\eqref{eq:sqrt5} and the Lemma follows.
\end{proof}
The next Lemma establishes that if a (non-square) matrix is full rank, a sufficiently small perturbation does not \emph{decrease} the rank (and hence rank is preserved). Note that full-rankness is really required, as an arbitrarily small perturbation can always \emph{increase} the rank.
\begin{lemma}\label{lem:equal_rank}
Let $A$ be an $m\times n$ $(m\leq n)$ matrix of rank $m$ and let $B$ an $m\times n$ matrix s.t. $\norm{A-B}\leq \sigma_{\rm min}(A)/2$. Then $\mathrm{rank}(A) = \mathrm{rank}(B)$. 
\end{lemma}
\begin{proof}
By construction, we have $\mathrm{rank}(A) \geq \mathrm{rank}(B)$. 
Moreover we have that the smallest singular value of $B$ is at least $\sigma_{\rm min}(A) - \norm{A-B}$, by Weyl's singular value perturbation theorem \cite{stewart1991perturbation}. Thus by construction the smallest singular value of $B$ is at least $\sigma_{\rm min}(A)/2$ which is strictly larger than 0 as $A$ is full rank and thus $B$ is also full rank, and thus $\mathrm{rank}(A) = \mathrm{rank}(B)$. 
\end{proof}

Finally, we will require a bound on the smallest nonzero singular value of $U_0$.
This is the only Lemma where we deviate significantly from the work done in \cite{li2019superresolution}, where the corresponding result, Lemma 3 in \cite{li2019superresolution}, makes explicit use of the fact that in their scenario all poles $z_j$ lie on the unit circle (what we call the real-time, oscillatory signal).
The bound we present here is simpler and more general and thus applies to both imaginary (decaying) as well as real-time (oscillatory) signals, but leads to a suboptimal dependence on the condition number of the Vandermonde matrix $V_L$ defined in Eq.~\eqref{eq:V} (in particular $\sigma_{\rm min}(V_L)$). However, it is sufficient for our purpose.
The Lemma will use some essential properties of how the ESPRIT Algorithm \ref{alg:two} works which we review first. Key to the functioning of ESPRIT is the fact that $H(g)$ has two decompositions
\begin{equation}
    H(g) = U\Sigma W = V_{L}  C V_{K-L}^{T}, 
\end{equation}
where $V_L$ is the $(L+1)\times S$ Vandermonde matrix defined in Eq.~\eqref{eq:V} and the coefficient matrix $C$ is given in Eq.~\eqref{def:matrixC}. When $S\leq L \leq K+1-S$ (requiring $K+1 \geq 2S$), $V_{L}$ and $V_{K-L}$ are full rank.
Then $V_{L}$ and $U_{S}$ have an image of the same dimension, which means there exists an invertible matrix $A$ such that
\begin{equation}
    U_S = V_{L}A,
    \label{eq:defA}
\end{equation}
and thus 
\begin{equation}
    U_0=V_{L-1}A, U_1=V_{L-1} Z A.
    \end{equation} 
with $Z={\rm diag}(z_1,\ldots, z_S)$.
This implies that 
\begin{align}
    \Psi=U_0^+ U_1=A^{-1} V_{L-1}^+ V_{L-1} Z A=A^{-1} Z A.
    \label{eq:signal-c}
\end{align}
and hence the eigenvalues of $\Psi$ are the poles $z_i$.

\begin{lemma}\label{lem:small_sing_U_0}
Let $U_0$ be the $L\times S $ matrix obtained from $U_S$ by removing the last row. If the associated Vandermonde matrix $V_{L-1}$ is of (full) rank $S$  then so is $U_0$, and moreover 
\begin{equation}
    \sigma_{\rm min}(U_0)\geq \frac{\sigma_{\rm min}(V_{L-1})}{\norm{V_L}}.
\end{equation}
\end{lemma}
\begin{proof}
We have
\begin{equation}
    I_S  = U_S^{\dagger}U_S = (V_{L}A)^\dagger V_{L}A = A^\dagger V_L ^\dagger V_L A,
\end{equation}
which means that the singular values of $A$ are precisely inverse to those of $V_L$, or equivalently that $A^+$ has the same singular spectrum as $V_L$.
Moreover, by assumption $V_{L-1}$ has full column rank, and $A$ is invertible so
\begin{equation}
    \sigma_{\rm min}^{-1}(U_0)=\norm{U_0^+} = \norm{( V_{L-1}A)^+} = \norm{A^{-1} V_{L-1}^+}\leq \norm{A^{-1}}\norm{ V_{L-1}^+}= \frac{\norm{V_L}}{\sigma_{\rm min}(V_{L-1})},
\end{equation}
which is the inverse of the Lemma statement. 
\end{proof}
With Lemmas \ref{lem:inverse_pert}, \ref{lem:equal_rank} and \ref{lem:small_sing_U_0} in hand, we can give a perturbation bound for the signal matrix $\tilde{\Psi}$. We will show that if $\tilde{U}_0$ does not deviate strongly from the rotated version of $U_0$, namely $U_0(O_2O_1)^\dagger$, then the noisy signal matrix is also close to the ideal (rotated) version. 
\begin{lemma}\label{lem:U_inv_bound}
Let $\Psi: = U_0^+U_1,\tilde{\Psi} = \tilde{U}_0^+ \tilde{U}_1$ be the ideal and perturbed version of the signal matrix, respectively. Now assume that $\norm{ U_0(O_2O_1)^\dagger-\tilde{U}_0}\leq \sigma_{\rm min}(U_0)/2$, where $O_2O_1$ is defined through the singular value decomposition of $U_S^\dagger \tilde{U}_S$ (i.e. $O_1 U_S^\dagger\tilde{U}_S O_2 = D$). With this assumption we have
\begin{equation}
    \norm{(O_2 O_1)\Psi(O_2 O_1)^\dagger-\tilde{\Psi}}\leq 5\norm{U_S(O_2 O_1)^\dagger-\tilde{U}_S} \frac{\norm{V_{L}}^2}{\sigma_{\rm min}^2(V_{L-1})}.
\end{equation}

\end{lemma}
\begin{proof}
Following \cite{li2019superresolution} we have
\begin{align}
    \norm{(O_2 O_1)\Psi(O_2 O_1)^\dagger-\tilde{\Psi}} \leq \norm{(O_2 O_1)U_0^+}\norm{U_1(O_2 O_1)^\dagger-\tilde{U}_1} +\norm{(O_2 O_1)U_0^+-\tilde{U}_0^+}\norm{\tilde{U_1}}.
\end{align}
We have $\norm{\tilde{U}_1}\leq \norm{\tilde{U}_S} = 1$, since $\tilde{U}_S^{\dagger} \tilde{U}_S=I_S$ and removing a row vector decreases the operator norm.
Similarly we have $\norm{U_1(O_2 O_1)^\dagger-\tilde{U}_1}\leq \norm{U_S(O_2 O_1)^\dagger-\tilde{U}_S}$.
Now note that by our initial assumption and Lemma \ref{lem:equal_rank} (with $A=U_0(O_2 O_1)^{\dagger}$ and $B=\tilde{U}_0$) we have $\mathrm{Rank}(U_0(O_2 O_1)^\dagger)= \mathrm{Rank}(\tilde{U}_0)$.
This means that we can use Lemma \ref{lem:inverse_pert} to conclude that
\begin{equation}
    \norm{(O_2 O_1)U_0^+-\tilde{U}_0^+}\leq \frac{1+ \sqrt{5}}{2}\frac{\norm{U_0(O_2 O_1)^\dagger-\tilde{U}_0}}{\sigma_{\rm min}^2(U_0)}\leq \frac{1+ \sqrt{5}}{2}\frac{\norm{U_S(O_2 O_1)^\dagger-\tilde{U}_S}}{\sigma_{\rm min}^2(U_0)}.
\end{equation}
Hence we get
\begin{align}
    \norm{(O_2 O_1)\Psi(O_2 O_1)^\dagger-\tilde{\Psi}}&\leq \frac{\norm{U_S(O_2 O_1)^\dagger-\tilde{U}_S}}{\sigma_{\rm min}(U_0)} + \frac{1+ \sqrt{5}}{2}\frac{\norm{U_S(O_2 O_1)^\dagger-\tilde{U}_S}}{\sigma_{\rm min}^2(U_0)} \notag \\
    &\leq (\sigma_{\rm min}(U_0) +(1 + \sqrt{5})/2 )\frac{\norm{U_S(O_2 O_1)^\dagger-\tilde{U}_S}}{\sigma_{\rm min}^2(U_0)} \notag \\
    &\leq  \left(\frac{3}{2}+ \sqrt{5} \right)\frac{\norm{U_S(O_2 O_1)^\dagger-\tilde{U}_S}}{\sigma_{\rm min}^2(U_0)},
\end{align}
where we used $\sigma_{\rm min}(U_0)\leq \norm{U_0}\leq 1$.
Plugging in the lower bound on $\sigma_{\rm min}(U_0)$ (Lemma \ref{lem:small_sing_U_0}) and noting that $\frac{3}{2}+ \sqrt{5}\leq 5$ we obtain the Lemma statement.
\end{proof}

Combining all of this we get the following general theorem. From now on we restrict ourselves to the case $L=K/2$:

\begin{theorem}
Let $(g+\eta)(k)$ be the signal with $g(k) = \sum_{i=1}^{S}c_i z_i^k$ and $\eta(k)$ a small noise vector to which we apply the ESPRIT algorithm \ref{alg:two}. Consider then the associated Hankel matrices $H(g)$ and $H(g+\eta)$, with singular value decompositions $H(g) = U\Sigma W$ and $H(g+\eta) = \tilde{U}\tilde{\Sigma} \tilde{W}$, and label the matrix of the first $S$ columns of $U$ (resp. $\tilde{U}$) as $U_S$ (resp. as $\tilde{U}_S$). Denote by $U_0$ (resp. $U_1$) the submatrix of $U_S$ with the last row (resp. first row) removed and define the signal matrix $\Psi = U_0^+ U_1$ (similarly for $\tilde{\Psi}$). Let $L=K/2$ and $K+1\geq 2S$. Now, if 
\begin{equation}
    4 \norm{H(\eta)}\leq  c_{\rm min}\sigma_{\rm min}^2(V_{K/2})\sigma_{\rm min}(V_{K/2-1})\norm{V_{K/2}}^{-1},
\end{equation}
 then
\begin{equation}
    \norm{(O_2 O_1)\Psi(O_2 O_1)^\dagger-\tilde{\Psi}} \leq \frac{10\sqrt{2S}\norm{H(\eta)}\norm{V_{K/2}}^2}{\sigma_{\rm min}^4(V_{K/2-1})}c_{\rm min}^{-1} .
\end{equation}
where $O_2O_1$ is defined through the singular value decomposition of $U_S^\dagger \tilde{U}_S$, Eq.~\eqref{eq:defD}. 
\label{thm:signalbound}
\end{theorem}

\begin{proof}
We start from the requirement in Lemma \ref{lem:U_inv_bound} that 
$\norm{U_0(O_2O_1)^\dagger - \tilde{U}_0 } \leq \norm{U_S(O_2O_1)^\dagger - \tilde{U}_S }\leq \sigma_{\rm min}(U_0)/2$. By Lemma \ref{lem:small_sing_U_0}, this is certainly satisfied if $\norm{U_S(O_2 O_1)^\dagger - \tilde{U}_S}\leq \sigma_{\rm min}(V_{K/2-1})\norm{V_{K/2}}^{-1}$/2. Moreover, from Lemma \ref{lem:U_pert} we know that $\norm{U_S(O_2 O_1)^\dagger - \tilde{U}_S} \leq 2\sqrt{2S}\norm{H(\eta)} \sigma_{\rm min}^{-1}(H(g))$ so now let's upperbound $\sigma_{\rm min}^{-1}(H(g))$ in terms of the smallest singular value of $V_{K/2}$. We have
\begin{equation}
    \sigma_{\rm min}^{-1}(H(g)) = \norm{H(g)^+} =\norm{(V_{K/2}^T)^+C^{-1}V_{K/2}^+} \leq \sigma_{\rm min}^{-2}(V_{K/2}) c_{\rm min}^{-1},
\end{equation}
where we used $c_{\rm min}=\min_i c_i$. This means that the condition 
\begin{equation}
    4\norm{H(\eta)}\leq  c_{\rm min}\sigma_{\rm min}^2(V_{K/2}) \sigma_{\rm min}(V_{K/2-1})\norm{V_{K/2}}^{-1},
\end{equation}
allows us to use Lemma \ref{lem:U_inv_bound}
\begin{equation}
    \norm{(O_2 O_1)\Psi(O_2 O_1)^\dagger-\tilde{\Psi}}\leq 10\sqrt{2S}\frac{\norm{H(\eta)}\norm{V_{K/2}}^2}{\sigma_{\rm min}^2(V_{K/2})\sigma_{\rm min}^2(V_{K/2-1})} c_{\rm min}^{-1} \leq \frac{10\sqrt{2S}\norm{H(\eta)}\norm{V_{K/2}}^2}{\sigma_{\rm min}^4(V_{K/2-1})}c_{\rm min}^{-1},
\end{equation}
where we also used the general fact about Vandermonde matrices \cite[theorem 1]{bazan} that $\sigma_{\rm min}(V_{K/2})  \geq \sigma_{\rm min}(V_{K/2-1})$  (i.e. the smallest non-zero singular value of $V_L$ grows monotonically with $L$).
\end{proof}

 We wish to translate the bound in Theorem \ref{thm:signalbound} to a theorem on the distance between the inferred eigenvalues $z_i$ and $\tilde{z}_i$. 
 The argument is as follows. We know from the Bauer-Fike theorem and the fact that $(O_2 O_1)\Psi(O_2 O_1)^\dagger$ is diagonalizable (see \cite[Exercise VIII.3.2]{bhatia2013matrix}) that the operator norm bound on $(O_2 O_1)\Psi(O_2 O_1)^\dagger$ implies a matching bound distance on its eigenvalues $z_i$, using Eq.~\eqref{eq:signal-c}. That is, we have
\begin{equation}
   d( \{z_i\},\{\tilde{z}_j\}):=  \max_{\pi\in \text{Perm}_S}\min_i|z_{\pi(i)}- \tilde{z}_i| \leq (2S-1)\kappa(A (O_2 O_1)^{\dagger})\norm{(O_2 O_1)\Psi(O_2 O_1)^\dagger-\tilde{\Psi}},
\end{equation}
where $\kappa(A (O_2 O_1)^{\dagger})=\kappa(A):= \norm{A}\norm{A^{-1}}$ is the condition number of the invertible matrix $A$ in 
Eq.~\eqref{eq:defA}. We have $A=V_{K/2}^+ U_S$ and since $U_S$ is an isometry we know that
\begin{equation}
    \kappa(A) = \norm{A}\norm{A^+} = \norm{V_{K/2}}\norm{V_{K/2}^+} = \frac{\norm{V_{K/2}}}{\sigma_{\rm min}(V_{K/2})},
\end{equation}
and hence we get a bound on the matching distance of the eigenvalues in terms of known quantities, as expressed in the following Theorem:

\begin{theorem}\label{thm:general-intermed}
Let $y(k) = (g+\eta)(k)$ ($k = 0,\ldots,K$) be the signal with $g(k) = \sum_{i=1}^{S}c_i z_i$, let $\eta(k)$ a small noise vector, and $K+1 \geq 2S$ ($L=K/2$). Let $\tilde{z}_i$ be the output of the ESPRIT algorithm. Then under the noise condition:
\begin{equation}
    4\norm{H(\eta)}\leq c_{\rm min}\sigma_{\rm min}^2(V_{K/2})\sigma_{\rm min}(V_{K/2-1})\norm{V_{K/2}}^{-1},
    \label{eq:noise-cond}
\end{equation}
we have
\begin{equation}
    d(\{z_i\}, \{\tilde z_j\})  \leq (2S-1)\frac{10\sqrt{2S}\norm{H(\eta)} \norm{V_{K/2}}^3}{\sigma_{\rm min}^5(V_{K/2-1})}c_{\rm min}^{-1}.
    \label{eq:upper-gen}
\end{equation}
\end{theorem}

This theorem thus holds for both the decaying as well as oscillatory signal. \\

In order to use the Theorem, one has to lower bound $\sigma_{\rm min}(V_{L-1})$ (and more trivially, upper bound $\norm{V_{L}}$ as well) for $L=K/2$. For complex poles $z_i$ on the unit circle one can get very good bounds, assuming a gap, see Eq.~\eqref{eq:gap-del}. 

To lower bound $\sigma_{\rm min}(V_{L-1})$ for a purely decaying signal, we start with the following characterization of \emph{square} Vandermonde matrices with real poles due to Gautschi:
\begin{theorem}[Theorem 1 in \cite{gautschi1962inverses}]
Let $V_{S-1}$ be a square $S \times S$ Vandermonde matrix with $S$ (unequal) real positive poles $z_1,\ldots,z_S$. Then $\infty$ norm of $V_{S-1}^{-1}$ is
\begin{equation}
    \norm{V_{S-1}^{-1}}_{\infty}:= \max_{i}\sum_{j}\left|\big(V_{S-1}^{-1}\big)_{ij}\right| = \max_{i\in \{1,\ldots,S\}} \prod_{j=1, j\neq i}^{S}\frac{1+z_i}{|z_j-z_i|}.
    \label{eq:gaut}
\end{equation}
\end{theorem}
Based on this Theorem we can work out a very similar statement for non-square Vandermonde matrices $V_{L-1}$ where $L$ is a multiple of $S$. Note that this Lemma does not depend on any gap. 
\begin{lemma}\label{lem:non-sq}
Let $V_{ST-1}$ be an $ST\times S$ Vandermonde matrix (where $T$ is a positive integer) with $S$ (unequal) real positive poles $z_1,\ldots,z_S\leq 1$. Then we have
\begin{equation}
    \norm{V_{ST-1}^{+}}_{\infty} \leq 2 \norm{V_{S-1}^{-1}}_{\infty}.
\end{equation}
\end{lemma}
\begin{proof}
Note that 
\begin{equation}
    V_{ST-1} =  \begin{pmatrix} V_{S-1}^T & Z^S V_{S-1}^T & Z^{2S} V_{S-1}^T& \cdots & Z^{(T-1)S}V_{S-1}^T\end{pmatrix}^T,
\end{equation}
with $Z = \mathrm{diag}(z_1,\ldots, z_S)$, using $(V_{S-1})_{ij}=z_j^{i-1}$ and $(Z^S V_{S-1}^T)_{ij}=z_i^{s+j-1}$.

The pseudo-inverse $V_{ST-1}^+$ of size $S \times ST$ can be directly calculated as
\begin{equation}
    V_{ST-1}^+ =   (I- Z^{2S})(I-Z^{2ST})^{-1}  \begin{pmatrix} V_{S-1}^{-1} & Z^S V_{S-1}^{-1} & Z^{2S}V_{S-1}^{-1} & \cdots & Z^{(T-1)S}V_{S-1}^{-1} \end{pmatrix} ,
\end{equation}
using a geometric series.
Hence $\norm{V_{ST-1}^{+}}_{\infty}$ can be calculated to be
\begin{align}
    \max_{i\in \{1,\ldots,S\}} \frac{1- z_i^{2S}}{1-z_i^{2ST}}\sum_{p=0}^{T-1} z_i^{pS} \prod_{j=1, j\neq i}^{S}\frac{1+z_i}{|z_j-z_i|}\leq \max_i \left[
    \frac{1- z_i^{2S}}{1-z_i^{2ST}}\frac{1-z_i^{ST}}{1-z_i^S}\right]
    \norm{V_{S-1}^{-1}}_{\infty}= \notag \\ \max_i \left[\frac{1+z_i^S}{1+z_i^{ST}}\right] \norm{V_{S-1}^{-1}}_{\infty}\leq 2 \norm{V_{S-1}^{-1}}_{\infty}.
    \end{align}
which gives the Lemma statement. 
\end{proof}

To apply this Lemma, we use that $\sigma_{\rm min}^{-1}(V_{ST-1}) = \norm{V_{ST-1}^+}\leq  \sqrt{S} \norm{V_{ST-1}^{+}}_{\infty}$ so that 
\begin{align}
    \sigma_{\rm min}(V_{ST-1})\geq (2 \sqrt{S})^{-1} \norm{V_{S-1}^{-1}}_{\infty}^{-1},
\end{align}
for any $T$. Unlike the lower bound for the real-time signal which explicitly uses the gap $\Delta$ in Eq.~\eqref{eq:gap-del}, this bound does not improve with $T$.
We will now use the gap to upper bound $\norm{V_{S-1}^{-1}}_{\infty}$ given in Eq.~\eqref{eq:gaut}, thus lower bounding $\sigma_{\rm min}(V_{ST-1})$ for any $T$. This is done in the proof of our final Theorem \ref{thm:final} (and its slight adaptation Theorem \ref{thm:final-D}) restated here:

\begin{tada}
Let $(g+\eta)(k)$ be an imaginary-time decaying signal (of length $K$) with $g(k) = \sum_{i=1}^{S}c_iz_i^k$, $c_i > 0\: \forall i$, $c_{\min}=\min_i c_i$ and $\eta(k)$ a small noise vector. Let $z_i = e^{-E_i}$ with $E_{i} \in [0, 2\pi)$ and given eigenvalue gap $\Delta < 1$ in Eq.~\eqref{eq:gap}, and $\{\tilde{E}_i\}$ the energy estimates of ESPRIT with $L=K/2$. Let $K+1 \geq 2S$, $K$ even and $K=TS$ for some positive integer $T$. 
If we have
\begin{equation}
    \norm{H(\eta)}\leq \frac{c_{\rm min}}{\sqrt{K}} g_1(S,\Delta),
    \label{eq:suff-cond-app}
\end{equation}
with 
\begin{align}
    g_1(S,\Delta)=\frac{1}{32 S^2}\, (e^{-2\pi}\pi \Delta)^{3(S-1)}, 
\end{align}
then 
\begin{equation}
    d(\{\tilde{E}_i\}, \{E_j\}) \leq  \norm{H(\eta)}\,c_{\min}^{-1} K\sqrt{K} g_2(S,\Delta),
    \label{eq:upper-app}
\end{equation}
with
\begin{align}
g_2(S, \Delta)= e^{2\pi} \sqrt{2} 640  \,S^{5.5} \,  (e^{-2\pi}\pi \Delta)^{-5(S-1)}.
\end{align}
\end{tada}

\begin{proof}
First, we can lift the bound on the eigenvalue distance $d(\{z_i\}, \{\tilde z_j\})$ to one on energies $d(\{E_i\}, \{\tilde E_j\})$ defined in Eq.~\eqref{eq:matching_error}, with $\tilde{E}_i := -\log(\tilde z_i)$ and $E_i \in [0, 2\pi)$ by noting that 
\begin{align}
    \frac{1}{2\pi}|E_{\pi(i)}- \tilde{E}_i| = \frac{1}{2\pi}|\log(z_{\pi(i)})- \log(\tilde{z}_i)|
    =\frac{1}{2\pi} |\log\big(1 - (\tilde{z}_i -z_{\pi(i)})/\tilde{z}_i\big)| \notag \\
    \leq |(\tilde{z}_i -z_{\pi(i)})/\tilde{z}_i\big)| 
    \leq e^{2\pi}|(\tilde{z}_i -z_{\pi(i)})|,
    \label{eq:convEz}
\end{align}
using that $z_i \in (e^{-2\pi},1]$ and $\tilde{z}_i \in (e^{-2\pi},1]$. In particular, for the first inequality, let $x=(\tilde{z}_i -z_{\pi(i)})/\tilde{z}_i$. If $x < 0$, $|\log (1-x)|\leq |x|$. If $x> 0$, since $z_i\in (e^{-2\pi},1]$, we have $x \leq 1-e^{-2\pi}$, so that $|\log(1-x)|\leq 2\pi |x|$. 

Second, let us now use the gap condition $|E_i - E_j|\geq 2\pi \Delta$ in Eq.~\eqref{eq:gap}. This leads to a gap condition on the $z_i$ themselves through (assuming w.l.o.g. that $z_i\geq z_j$):
\begin{align}
2\pi \Delta \leq|E_i - E_j| =|\log(z_i/z_j)|,
\end{align}
and thus
\begin{equation}
    e^{2\pi \Delta} \leq \frac{z_i}{z_j} = \frac{z_i-z_j}{z_j} +1,
\end{equation}
which gives
\begin{equation}
|z_i-z_j| = z_i-z_j \geq z_j (e^{2\pi \Delta}-1) \geq e^{-2\pi} 2 \pi \Delta.
\end{equation}
This implies through Eq.~\eqref{eq:gaut} that
\begin{align}
\norm{V_{S-1}^{-1}}_{\infty}\leq (\pi \Delta)^{-(S-1)} e^{2\pi(S-1)},
\end{align}
so that 
\begin{align}
\sigma_{\rm min}(V_{ST-1}) \geq (2\sqrt{S})^{-1} (\pi \Delta e^{-2\pi})^{(S-1)},
\end{align}
for any integer $T$. We note that this lower bound on $\sigma_{\rm min}$ is exponentially small in $S$ as $\Delta < 1$.

Third, we need an upper bound on $\norm{V_L}$ in order to use Theorem \ref{thm:general-intermed}. For $z_j \in (0,1]$, we have $\norm{V_L}\leq \norm{V_L}_F=\left(\sum_{i,j} z_j^{2(i-1)}\right)^{1/2}\leq S^{1/2} \left(\sum_{i=1}^{L+1} z_{\rm max}^{2(i-1)}\right)^{1/2}=S^{1/2}\left(\frac{1-z_{\rm max}^{2(L+1)}}{1-z_{\rm max}^2}\right)^{1/2}$ which can tend to $(S(L+1))^{1/2}$ when $z_{\rm max} \rightarrow 1$, so we use the simple upper bound $(S(L+1))^{1/2}\leq (SK)^{1/2}$. Putting all this together allows to translate Eq.~\eqref{eq:upper-gen} to Eq.~\eqref{eq:upper-app}. 
The condition on $\norm{H(\eta)}$ in Eq.~\eqref{eq:noise-cond} then translates to the sufficient condition in Eq.~\eqref{eq:suff-cond-app} using that $\sigma_{\rm min}(V_{K/2})\geq \sigma_{\rm min}(V_{K/2-1})$, the lower bound on $\sigma_{\rm min}(V_{K/2-1})$, and the upper bound on $\norm{V_{K/2}}$.
\end{proof}

The adapted version, Theorem \ref{thm:final-D}, is proved almost identically (but requires that all $E_i$ are in principle bounded away from $2\pi$):

\begin{tada2}
Let $(g+\eta)(k)$ be an imaginary-time decaying signal (of length $K$) with $g(k) = \sum_{i=1}^{S}c_iz_i^k$, $c_i > 0, \:\forall i$, $c_{\min}=\min_i c_i$, and $\eta(k)$ a small noise vector. Let $z_i = 1-E_i/2\pi$ with $E_{i} \in [0, \pi]$ and given eigenvalue gap $\Delta < 1$ in Eq.~\eqref{eq:gap}, and $\{\tilde{E}_i\}$ the energy estimates of ESPRIT with $L=K/2$. Let $K+1 \geq 2S$, $K$ even and $K=TS$ for some positive integer $T$. 
If we have
\begin{equation}
    \norm{H(\eta)}\leq \frac{c_{\rm min}}{\sqrt{K}} \tilde{g}_1(S,\Delta),
    \label{eq:suff-cond}
\end{equation}
with 
\begin{align}
    \tilde{g}_1(S,\Delta)=\frac{1}{32 S^2}\,  \Delta^{3(S-1)}, \end{align}
then 
\begin{equation}
    d(\{\tilde{E}_i\}, \{E_j\}) \leq  \norm{H(\eta)}\,c_{\min}^{-1} K\sqrt{K} \tilde{g}_2(S,\Delta),
    \label{eq:upper-d}
\end{equation}
with
\begin{align}
\tilde{g}_2(S, \Delta)= 640 \sqrt{2}  \,S^{5.5} \,  \Delta^{-5(S-1)}.
\end{align}
\end{tada2}

\begin{proof}
First, we convert the eigenvalue distance $d(\{z_i\}, \{\tilde z_j\})$ to one on energies $d(\{E_i\}, \{\tilde E_j\})$ defined in Eq.~\eqref{eq:matching_error}, with $\tilde{E}_i := 2\pi(1-\tilde{z}_i)$, so 
\begin{align}
    \frac{1}{2\pi}|E_{\pi(i)}- \tilde{E}_i| = |z_{\pi(i)}-\tilde{z}_i|.
    \label{eq:convEz-D}
\end{align}
Second, let us now use the gap condition $|E_i - E_j|\geq 2\pi \Delta$ in Eq.~\eqref{eq:gap}. This leads to a gap condition on the $z_i$ themselves through:
\begin{equation}
|z_i-z_j| \geq \Delta.
\end{equation}
This implies through Eq.~\eqref{eq:gaut} that
\begin{align}
\norm{V_{S-1}^{-1}}_{\infty}\leq \Delta^{-(S-1)} 
\end{align}
so that 
\begin{align}
\sigma_{\rm min}(V_{ST-1}) \geq (2\sqrt{S})^{-1}\Delta ^{(S-1)},
\end{align}
for any integer $T$. Following identical steps as in the proof of the previous Theorem then leads to the final statements.
\end{proof}

\printbibliography

@article{somma:njp,
	doi = {10.1088/1367-2630/ab5c60},
	%url = {https://doi.org/10.1088/1367-2630/ab5c60},
	year = {2019},
	month = dec,
	journal = {New Journal of Physics},
	publisher = {{IOP} Publishing},
	volume = {21},
	number = {12},
	pages = {123025},
	author = {Somma, RD},
	title = {Quantum eigenvalue estimation via time series analysis}
}

@article{stoica1989music,
    title={MUSIC, maximum likelihood, and Cramer-Rao bound},
    author={Stoica, P and Nehorai, A},
    journal={IEEE Transactions on Acoustics, speech, and signal processing},
    volume={37},
    number={5},
    pages={720--741},
    year={1989},
    month = {5},
    publisher={IEEE},
    doi={10.1109/29.17564}
}

@article{li1998relative,
    title={Relative perturbation theory: II. Eigenspace and singular subspace variations},
    author={Li, RC},
    journal={SIAM Journal on Matrix Analysis and Applications},
    volume={20},
    number={2},
    pages={471--492},
    year={1998},
    publisher={SIAM},
    url = {https://doi.org/10.1137/S0895479896298506}
}

@book{book:QMC,
    title={Quantum {M}onte {C}arlo Methods},
    subtitle={Algorithms for lattice models},
    author={Gubernatis, J and Kawashima, N and Werner, P},
    year={2016},
    publisher={Cambridge University Press},
    doi = {10.1017/CBO9780511902581}
}

@book{bhatia2013matrix,
    title={Graduate Texts in Mathematics: Matrix analysis},
    author={Bhatia, R},
    volume={169},
    year={1997},
    publisher={Springer Science \& Business Media},
    doi = {10.1007/978-1-4612-0653-8}
}

@article{pan:bad,
	author = {Pan, VY},
	doi = {10.1137/15M1030170},
	%eprint = {https://doi.org/10.1137/15M1030170},
	journal = {SIAM Journal on Matrix Analysis and Applications},
	number = {2},
	pages = {676-694},
	title = {How Bad Are Vandermonde Matrices?},
	%url = {https://doi.org/10.1137/15M1030170},
	volume = {37},
	year = {2016},
}

@article{gautschi1962inverses,
  title={On the inverses of Vandermonde and confluent Vandermonde matrices I.},
  author={Gautschi, W},
  journal={Numer. Math},
  volume={4},
  pages={117--123},
  year={1962},
  publisher={Springer},
  url={https://www.cs.purdue.edu/homes/wxg/selected_works/section_01/016.pdf}
}

@article{wedin1973perturbation,
    title={Perturbation theory for pseudo-inverses},
    author={Wedin, P{\AA}},
    journal={BIT Numerical Mathematics},
    volume={13},
    number={2},
    pages={217--232},
    year={1973},
    month = {6},
    publisher={Springer},
    url = {https://doi.org/10.1007/BF01933494}
}

@article{QMCreview,
    doi = {10.1103/RevModPhys.73.33},
    %url = {https://journals.aps.org/rmp/abstract/10.1103/RevModPhys.73.33},
    title = {Quantum Monte Carlo simulations of solids},
    author = {Foulkes, WMC and Mitas, L and Needs, RJ and Rajagopal, G},
    journal = {Reviews of Modern Physics},
    publisher = {{American Physical Society}},
    volume = {73},
    %pages = {1122--1140},
    month = jan,
    year = {2001},
    number = {1}
}

@book{MikeIke,
    author = {Nielsen, MA and Chuang, IL},
    title = {Quantum computation and quantum information: 10th Anniversary Edition},
    year = {2010},
    publisher = {Cambridge University Press},
    doi = {10.1017/CBO9780511976667}
}

@article{SHF:QPE,
    author = {Svore, KM and Hastings, MB and Freedman, MH},
    title = {Faster phase estimation},
    journal = {Quantum Inf. Comput.},
    volume = {14},
    number = {3-4},
    pages = {306--328},
    year = {2014},
    month = {3},
    %url = {https://doi.org/10.26421/QIC14.3-4-7},
    doi = {10.26421/QIC14.3-4-7},
    timestamp = {Thu, 29 Apr 2021 18:05:35 +0200},
    biburl = {https://dblp.org/rec/journals/qic/SvoreHF14.bib},
    bibsource = {dblp computer science bibliography, https://dblp.org}
}

@article{TomBarbara,
    doi = {10.1088/1367-2630/aafb8e/meta},
    %url = {https://iopscience.iop.org/article/10.1088/1367-2630/aafb8e/meta},
    title = {Quantum phase estimation of multiple eigenvalues for small-scale (noisy) experiments},
    author = {O'Brien, TE and Tarasinski, B and Terhal, BM},
    journal = {New Journal of Physics},
    publisher = {{IOP Publishing}},
    volume = {21},
    %pages = {1122--1140},
    month = 2,
    year = {2019},
    number = {023022}
}

@article{Sarkar,
    doi = {10.1109/74.370583},
    %url = {https://ieeexplore.ieee.org/abstract/document/370583},
    title = {Using the matrix pencil method to estimate the parameters of a sum of complex exponentials},
    author = {Sarkar, TK and Pereira, O},
    journal = {IEEE Antennas and Propagation Magazine},
    publisher = {{IEEE}},
    volume = {37},
    pages = {48--55},
    month = 2,
    year = {1995},
    number = {1}
}

@article{greville1966note,
  title = {Note on the generalized inverse of a matrix product},
  author = {Greville, TNE},
  journal = {SIAM Review},
  volume = {8},
  number = {4},
  pages = {518--521},
  year = {1966},
  month = {10},
  publisher = {SIAM},
  doi = {10.1137/1008107}
}

@article{Sarkar2,
    doi = {10.1109/78.80911},
    %url = {https://ieeexplore.ieee.org/abstract/document/80911},
    title = {On SVD for estimating generalized eigenvalues of singular matrix pencil in noise},
    author = {Hua, Y and Sarkar, TK},
    journal = {IEEE Transactions on Signal Processing},
    publisher = {{IEEE}},
    volume = {39},
    pages = {892--900},
    month = apr,
    year = {1991},
    number = {4}
}

@article{pt:prony,
    doi = {10.1016/j.laa.2012.10.036},
    %url = {https://www.sciencedirect.com/science/article/pii/S0024379512007665},
	Author = {Potts, D and Tasche, M},
	Journal = {Linear Algebra and its Applications},
	Number = {4},
	Pages = {1024--1039},
	Title = {Parameter estimation for nonincreasing exponential sums by {P}rony-like methods},
	Volume = {439},
	year = {2013},
	month = aug,
	%eprint = {\url{https://doi.org/10.1016/j.laa.2012.10.036}}
}

@article{motta:imag,
    title = {Determining eigenstates and thermal states on a quantum computer using quantum imaginary time evolution},
    volume = {16},
    ISSN = {1745-2481},
    %url = {http://dx.doi.org/10.1038/s41567-019-0704-4},
    DOI = {10.1038/s41567-019-0704-4},
    number = {2},
    journal = {Nature Physics},
    publisher = {Springer Science and Business Media LLC},
    author = {Motta, M and Chong, S and Tan, ATK and O’Rourke, MJ and Ye, E and Minnich, AJ and Brandão, FGSL and Chan, GK},
    year = {2019},
    month = nov,
    pages = {205–210}
}

@article{OWE:RB,
    title = {Randomized benchmarking for individual quantum gates},
    volume = {123},
    ISSN = {1079-7114},
    %url = {http://dx.doi.org/10.1103/PhysRevLett.123.060501},
    DOI = {10.1103/physrevlett.123.060501},
    number = {6},
    journal = {Physical Review Letters},
    publisher = {American Physical Society (APS)},
    author = {Onorati, E and Werner, AH and Eisert, J},
    year = {2019},
    month = aug
}

@misc{helsen2020general,
    title = {\mkbibemph{``A general framework for randomized benchmarking"}}, 
    author = {Helsen, J and Roth, I and Onorati, E and Werner, AH and Eisert, J},
    year = {2020},
    month = oct,
    eprint = {2010.07974},
    archivePrefix = {arXiv},
    primaryClass = {quant-ph}
}

@misc{hangleiter2021precise,
    title = {\mkbibemph{``Precise Hamiltonian identification of a superconducting quantum processor"}}, 
    author = {Hangleiter, D and Roth, I and Eisert, J and Roushan, P},
    year = {2021},
    month = aug,
    eprint = {2108.08319},
    archivePrefix = {arXiv},
    primaryClass = {quant-ph}
}

@article{sandvik:analytic,
  title = {Stochastic method for analytic continuation of quantum Monte Carlo data},
  author = {Sandvik, Anders W.},
  journal = {Phys. Rev. B},
  volume = {57},
  issue = {17},
  pages = {10287--10290},
  numpages = {0},
  year = {1998},
  month = {5},
  publisher = {American Physical Society},
  doi = {10.1103/PhysRevB.57.10287},
  url = {https://link.aps.org/doi/10.1103/PhysRevB.57.10287}
}

@article{lin2021heisenberglimited,
    title = {Heisenberg-limited ground-state energy estimation for early fault-tolerant quantum computers},
    author = {Lin, L and Tong, Y},
    journal = {PRX Quantum},
    volume = {3},
    issue = {1},
    pages = {010318},
    numpages = {21},
    year = {2022},
    month = {2},
    publisher = {American Physical Society},
    doi = {10.1103/PRXQuantum.3.010318},
    %url = {https://link.aps.org/doi/10.1103/PRXQuantum.3.010318}
}

@misc{WBC:random,
    title = {\mkbibemph{``A randomized quantum algorithm for statistical phase estimation"}}, 
    author = {Wan, K and Berta, M and Campbell, ET},
    year = {2021},
    month = {10},
    eprint = {2110.12071},
    archivePrefix = {arXiv},
    primaryClass = {quant-ph}
}

@article{harrow+:complex,
   title={Classical algorithms, correlation decay, and complex zeros of partition functions of quantum many-body systems},
   url={http://dx.doi.org/10.1145/3357713.3384322},
   DOI={10.1145/3357713.3384322},
   journal={Proceedings of the 52nd Annual ACM SIGACT Symposium on Theory of Computing},
   publisher={ACM},
   author={Harrow, Aram W. and Mehraban, Saeed and Soleimanifar, Mehdi},
   year={2020},
   month={6}
}

@misc{dequant:GG,
  author = {Gharibian, S and Le Gall, F},
  title = {\mkbibemph{"Dequantizing the Quantum Singular Value Transformation: Hardness and Applications to Quantum Chemistry and the Quantum PCP Conjecture"}},
  year = {2021},
  month = {Nov},
  archivePrefix = {arXiv},
  primaryClass = {quant-ph},
  eprint = {2111.09079}
}

@article{francesco,
    author = {Helsen, J and Battistel, F and Terhal, BM},
    title = {Spectral quantum tomography},
    journaltitle = {npj Quantum Information},
    year = {2019},
    month = sep,
    day = {2},
    volume = {5},
    number = {74},
    publisher = {Nature Publishing Group},
    doi = {10.1038/s41534-019-0189-0},
    %url = {https://www.nature.com/articles/s41534-019-0189-0}
}

@article{Moitra,
    doi = {10.1145/2746539.2746561},
    %url = {https://dl.acm.org/doi/abs/10.1145/2746539.2746561},
    title = {Super-resolution, extremal functions and the condition number of Vandermonde matrices},
    author = {Moitra, A},
    journal = {Proceedings of the forty-seventh annual ACM symposium on Theory of Computing},
    publisher = {{Association for Computing Machinery}},
    %volume = {119},
    pages = {821--830},
    month = jun,
    year = {2015},
    %number = {10}
}

@article{Lloyd,
    url = {http://www.jstor.org/stable/2899535},
    title = {Universal quantum simulators},
    author = {Lloyd, S},
    journal = {Science},
    publisher = {{American Association for the Advancement of Science}},
    volume = {273},
    pages = {1073--1078},
    month = aug,
    year = {1996},
    number = {5278}
}

@article{childs+:trotter,
    title = {Theory of Trotter error with commutator scaling},
    author = {Childs, AM and Su, Y and Tran, MC and Wiebe, N and Zhu, S},
    journal = {Phys. Rev. X},
    volume = {11},
    issue = {1},
    pages = {011020},
    numpages = {49},
    year = {2021},
    month = {2},
    publisher = {{American Physical Society}},
    doi = {10.1103/PhysRevX.11.011020},
}

@article{Stoq,
    doi = {10.5555/2011772.2011773},
    %url = {https://dl.acm.org/doi/abs/10.5555/2011772.2011773},
    title = {The complexity of stoquastic local Hamiltonian problems},
    author = {Bravyi, S and DiVincenzo, DP and Oliveira, R and Terhal, BM},
    journal = {Quantum Inf. Comput.},
    publisher = {{Rinton Press Inc.}},
    volume = {8},
    pages = {361--385},
    month = may,
    year = {2008},
    number = {5}
}

@book{Sachdev,
    author = {Sachdev, S},
    title = {Quantum phase transitions},
    edition = {2},
    year = {2011},
    publisher = {Cambridge University Press},
    doi = {10.1017/CBO9780511973765}
}

@article{Suzuki,
    doi = {10.1063/1.529425},
    %url = {https://aip.scitation.org/doi/abs/10.1063/1.529425},
    author = {Suzuki, M},
    title = {General theory of fractal path integrals with applications to many-body theories and statistical physics},
    journaltitle = {Journal of Mathematical Physics},
    publisher = {{American Institute of Physics}},
    year = {1991},
    month = feb,
    volume = {32},
    number = {2},
    pages = {400--407}
}

@article{Sarkar3,
    doi = {10.1109/29.56027},
    %url = {https://ieeexplore.ieee.org/abstract/document/56027},
    title = {Matrix pencil method for estimating parameters of exponentially damped/undamped sinusoids in noise},
    author = {Hua, Y and Sarkar, TK},
    journal = {IEEE Transactions on Acoustics, Speech and Signal Processing},
    publisher = {{IEEE}},
    volume = {38},
    pages = {814--824},
    month = may,
    year = {1990},
    number = {5}
}

@article{BravyiGosset,
    doi = {10.1103/PhysRevLett.119.100503},
    %url = {https://journals.aps.org/prl/abstract/10.1103/PhysRevLett.119.100503},
    title = {Polynomial-Time classical simulation of quantum ferromagnets},
    author = {Bravyi, S and Gosset, D},
    journal = {Physical Review Letters},
    publisher = {{American Physical Society}},
    volume = {119},
    %pages = {1122--1140},
    month = sep,
    year = {2017},
    number = {10}
}

@article{stewart1991perturbation,
    author = {Stewart, G. W.},
  journal = {SVD and signal processing, II: algorithms, analysis and applications},
  pages = {99--109},
  publisher = {Citeseer},
  title = {{Perturbation theory for the singular value decomposition}},
  year = 1991,
  url={https://users.math.msu.edu/users/iwenmark/Teaching/MTH995/Papers/SVD_Stewart.pdf}
}

@INPROCEEDINGS{AB:stoq,
    author = {Aharonov, D and Bredariol Grilo, A},
    booktitle = {2019 IEEE 60th Annual Symposium on Foundations of Computer Science (FOCS)}, 
    title = {Stoquastic PCP vs. randomness}, 
    year = {2019},
    volume = {},
    number = {},
    pages = {1000-1023},
    doi = {10.1109/FOCS.2019.00065}
}

@article{CH:mixing,
    doi = {10.22331/q-2021-02-11-395},
    %url = {https://doi.org/10.22331/q-2021-02-11-395},
    title = {Rapid mixing of path integral {M}onte {C}arlo for 1{D} stoquastic {H}amiltonians},
    author = {Crosson, E and Harrow, AW},
    journal = {{Quantum}},
    issn = {2521-327X},
    publisher = {{Verein zur F{\"{o}}rderung des Open Access Publizierens in den Quantenwissenschaften}},
    volume = {5},
    pages = {395},
    month = feb,
    year = {2021}
}

@misc{Crosson,
    title = {\mkbibemph{``Classical simulation of high temperature quantum Ising models"}},
    author = {Crosson, E and Slezak, S},
    year = {2020},
    month = feb,
    eprint = {2002.02232v1},
    archivePrefix = {arXiv},
    primaryClass = {quant-ph}
}

@misc{ComplStoq,
    title = {\mkbibemph{"Sign-curing local Hamiltonians: termwise versus global stoquasticity and the use of Clifford transformations"}},
    author = {Ioannou, M and Piddock, S and Marvian, M and Klassen, J and Terhal, BM},
    year = {2020},
    month = {7},
    eprint = {2007.11964},
    archivePrefix = {arXiv},
    primaryClass = {quant-ph}
}

@book{HornJohnson,
    author = {Horn, RA and Johnson, CR},
    title = {Matrix analysis},
    year = {2012},
    publisher = {Cambridge University Press},
    edition = {2},
    doi = {10.1017/9781139020411},
}

@article{SergeyQPEvsMC,
    doi = {10.5555/2871363.2871366},
    %url = {https://dl.acm.org/doi/abs/10.5555/2871363.2871366},
    title = {Monte Carlo simulation of stoquastic Hamiltonians},
    author = {Bravyi, S},
    journal = {Quantum Inf. Comput.},
    publisher = {{Rinton Press Inc.}},
    volume = {15},
    pages = {1122--1140},
    month = oct,
    year = {2015},
    number = {13-14}
}

@article{bazan, author = {Baz\'{a}n, FSV}, title = {Conditioning of rectangular Vandermonde matrices with nodes in the unit disk}, year = {1999}, issue_date = {Oct.-Jan. 2000}, publisher = {Society for Industrial and Applied Mathematics}, address = {USA}, volume = {21}, number = {2}, issn = {0895-4798}, 
%url = {https://doi.org/10.1137/S0895479898336021}, 
doi = {10.1137/S0895479898336021}, journal = {SIAM J. Matrix Anal. Appl.}, month = {10}, pages = {679–693}, numpages = {15}, keywords = {almost normal matrices, singular values, Vandermonde matrices, exponential modeling} }

@article{li2019superresolution,
    author = {Li, W and Liao, W and Fannjiang, A},
    journal = {IEEE Transactions on Information Theory}, 
    title = {Super-resolution limit of the ESPRIT algorithm}, 
    year = {2020},
    month = {2},
    volume = {66},
    number = {7},
    pages = {4593-4608},
    doi = {10.1109/TIT.2020.2974174},
    publisher = {IEEE}
}

@misc{GithubMaarten,
    author = {Stroeks, MEHM},
    title = {\mkbibemph{``ClassQuantSimStoqHam"}},
    year = {2022},
    publisher = {GitHub},
    journal = {GitHub repository},
    howpublished = {\url{https://github.com/MStroeks/ClassQuantSimStoqHam}}
}

@article{MOM,
    doi = {10.1007/s10208-019-09427-x},
    title = {Mean estimation and regression under heavy-tailed distributions: a survey},
    author = {Lugosi, G and Mendelson, S},
    journal = {Foundations of Computational Mathematics},
    volume = {19},
    pages = {1145–-1190},
    month = {8},
    year = {2019}
}

\end{document}